\documentclass [12pt]{amsart}
\usepackage[utf8]{inputenc}
\pdfoutput=1
\usepackage{amsmath,amssymb,amsthm,amsfonts}
\usepackage{mathtools}

\DeclareRobustCommand{\gobblefour}[5]{}
\newcommand*{\SkipTocEntry}{\addtocontents{toc}{\gobblefour}}

\usepackage[english]{babel}
\usepackage{comment}
\usepackage{hyperref}
\usepackage{bbm}
\usepackage{mathrsfs}

\usepackage{tikz,graphicx,color}
\usepackage{tikz-cd}
\usepackage{tikz-3dplot}
\usetikzlibrary{calc}
\usetikzlibrary{arrows}
\usetikzlibrary{shapes}
\usetikzlibrary{patterns}
\usetikzlibrary{positioning}
\usetikzlibrary{arrows.meta}
\usetikzlibrary{decorations.markings}

\pdfpageattr{/Group <</S /Transparency /I true /CS /DeviceRGB>>}

\usepackage{epstopdf}

\usepackage{enumerate}
\usepackage[arrow]{xy}
\usepackage{diagbox}
\usepackage{subfig}
\usepackage{arcs}
\usepackage{xcolor}
\usepackage{tabu}
\usepackage{booktabs}

\usepackage[margin=0.85in]{geometry}

\newtheorem{theorem}{Theorem}[section]

\newtheorem{lemma}[theorem]{Lemma}
\newtheorem{proposition}[theorem]{Proposition}
\newtheorem{corollary}[theorem]{Corollary}
\theoremstyle{definition}
\newtheorem{remark}[theorem]{Remark}
\theoremstyle{definition}
\newtheorem{definition}[theorem]{Definition}
\newtheorem{conjecture}[theorem]{Conjecture}
\newtheorem{question}[theorem]{Question}
\theoremstyle{definition}
\newtheorem{problem}[theorem]{Problem}
\theoremstyle{definition}
\newtheorem{example}[theorem]{Example}

\def\Acal{\mathcal{A}}\def\Fcal{\mathcal{F}}\def\Ical{\mathcal{I}}\def\Mcal{\mathcal{M}}\def\Xcal{\mathcal{X}}

% ...

\def\one{{\mathbbm{1}}}

\def\R{\mathbb{R}}

\def\P{\mathbb{P}}
% ...

% ...

\def\<{{\langle}}
\def\>{{\rangle}}

\def\RP{{\R P}}

\def\det{{ \operatorname{det}}}

\def\diag{{ \operatorname{diag}}}

\def\proj{ \operatorname{proj}}

\def\wt{\operatorname{wt}}

\def\RR{{\mathbb R}}

\def\RP{{\RR\mathbb P}}

\def\GL{\operatorname{GL}}
\def\SL{\operatorname{SL}}

\def\Mat{\operatorname{Mat}}
\def\Gr{\operatorname{Gr}}
\def\Grtnn{\Gr_{\ge 0}}
\def\Grtp{\Gr_{>0}}
\def\alt{\operatorname{alt}}

\def\medpa{\tau}
\def\xing{{\operatorname{xing}}}

\def\n{N}
\def\pc{\Pi^{>0}}

\def\Gmed{{G^\times}}
\def\Gbip{{G^{\square}}}
\def\Gdub{\widehat{G}^\square}
\def\Edub{\widehat{E}^\square}
\def\Vdub{\widehat{V}^\square}
\def\bbip{d}

\def\vscl{0.3}
\def\vsclG{0.3}
\def\scsc{0.6}
\def\tikzscl{1}
\def\bscl{0.6}
\def\sclbx{1.4}
\newcommand\edgepl[2]{\draw[line width=1pt,blue] (#1) -- (#2);}
\newcommand\edgemed[2]{\draw[line width=1.5pt,brown] (#1) -- (#2);}

\newcommand\edge[2]{\draw[line width=1.5pt] (#1) -- (#2);}
\newcommand\edgedir[2]{\draw[-{Stealth[scale=1.0]},line width=1.5pt] (#1) -- (#2);}
\newcommand\edgeop[2]{\draw[dashed, line width=1.5pt,opacity=0.4] (#1) -- (#2);}
\newcommand\disk[1]{\draw[dashed, opacity=0.3] (#1) circle (1);}
\newcommand\vertx[2]{\node[draw,circle,fill=black,scale=\vsclG] (#1) at (#2) {};}
\newcommand\vwh[2]{\node[draw,circle,fill=white,scale=\vscl] (#1) at (#2) {};}
\newcommand\vbl[2]{\node[draw,circle,fill=black!50,scale=\vscl] (#1) at (#2) {};}

\newcommand\edgeplop[2]{\draw[line width=0.3pt,black,opacity=0.5] (#1) -- (#2);}
\newcommand\edgeplmatch[2]{\draw[line width=2pt,blue] (#1) -- (#2);}

\def\side{0.3}

\newcommand\drawbbip{
    \node[anchor=30,scale=\bscl] (d1) at (210:1) {$\bbip_4$};
    \node[anchor=-30,scale=\bscl] (d2) at (-210:1) {$\bbip_3$};
    \node[anchor=210,scale=\bscl] (d3) at (30:1) {$\bbip_2$};
    \node[anchor=-210,scale=\bscl] (d3) at (-30:1) {$\bbip_1$};
}

    \newcommand\drawgraph{
    \coordinate (zero) at (0,0);
    \disk{0,0}
    \edgeop{-1,0}{1,0}
    \vwh{A}{\side,\side}
    \vwh{C}{-\side,-\side}
    \vbl{B}{\side,-\side}
    \vbl{D}{-\side,\side}
    \vbl{AA}{30:1}
    \vwh{BB}{-30:1}
    \vbl{CC}{-150:1}
    \vwh{DD}{150:1}
    \edgeplop{A}{AA}
    \edgeplop{B}{BB}
    \edgeplop{C}{CC}
    \edgeplop{D}{DD}
    \edgeplop{A}{B}
    \edgeplop{C}{B}
    \edgeplop{C}{D}
    \edgeplop{A}{D}
    \drawbbip
    
    \node[anchor=south,scale=\scsc] (stop) at (0,\side) {$c_e$};
    \node[anchor=north,scale=\scsc] (sbot) at (0,-\side) {$c_e$};
    \node[anchor=west,scale=\scsc] (cright) at (\side,0) {$s_e$};
    \node[anchor=east,scale=\scsc] (cleft) at (-\side,0) {$s_e$};
    }

\def\Prob{\mathbf{P}}
\def\Zpart{Z}

\def\Space{\Xcal}
\def\Closure{\overline{\Space}}
\def\OG{\operatorname{OG}}
\def\OGtnn{\OG_{\ge0}}
\def\OGtp{\OG_{>0}}
\def\doublemap{\phi}
\newcommand\double[1]{\widetilde{#1}}
\newcommand\Matsymx[1]{\Mat^{\operatorname{sym}}_{#1}(\R,1)}
\def\Matsym{\Matsymx{n}}

\def\insideprod{q}

\def\P{P}

\begin{document}
\numberwithin{equation}{section}

\title{Ising model and the positive orthogonal Grassmannian}
\author{Pavel Galashin}
\address{Department of Mathematics, Massachusetts Institute of Technology,
Cambridge, MA 02139, USA}
\email{{\href{mailto:galashin@mit.edu}{galashin@mit.edu}}}

\author{Pavlo Pylyavskyy}
\address{Department of Mathematics, University of Minnesota,
Minneapolis, MN 55414, USA}
\email{{\href{mailto:ppylyavs@umn.edu}{ppylyavs@umn.edu}}}
\thanks{P. P. was partially supported by NSF grants DMS-1148634 and DMS-1351590.}
\date{\today}



\begin{abstract}
We completely describe by inequalities the set of boundary correlation matrices of planar Ising networks embedded in a disk. Specifically, we build on a recent result of M.~Lis to give a simple bijection between such correlation matrices and points in the totally nonnegative part of the orthogonal Grassmannian, which has been introduced in 2013 in the study of the scattering amplitudes of ABJM theory. We also show that the edge parameters of the Ising model for reduced networks can be uniquely recovered from boundary correlations, solving the inverse problem. Under our correspondence, the Kramers--Wannier high/low temperature duality transforms into the cyclic symmetry of the Grassmannian, and using this cyclic symmetry, we prove that the spaces under consideration are homeomorphic to closed balls. 
\end{abstract}

\maketitle

\setcounter{tocdepth}{2}
\tableofcontents

\newgeometry{margin=1in}
\section{Introduction}\label{sec:intro}

The Ising model, introduced by Lenz in~1920 as a model for ferromagnetism and solved by Ising~\cite{Ising} in dimension $1$, plays a central role in statistical mechanics and conformal field theory. One of the main features of this model is that it undergoes a phase transition in dimensions larger than $1$. In particular, the critical temperature $\frac12\log(\sqrt2+1)$ for the case of the two-dimensional square lattice has been computed by Kramers and Wannier~\cite{KrWa}, who found a duality transformation exchanging subcritical and supercritical temperatures. The free energy of the model was computed by Onsager~\cite{Onsager} and Yang~\cite{Yang}, and since then it became a subject of active mathematical and physical research. Conformal invariance of the scaling limit was conjectured in~\cite{BPZ1,BPZ2} in relation to conformal field theory, and proven more recently as a part of a series of groundbreaking results by Smirnov, Chelkak, Hongler, Izyurov, and others~\cite{Smirnov,CHI,CS,HS,CDCHKS}.

Among the most important quantities associated with the Ising model are two-point and multi-point correlation functions. In particular, their limit at criticality exists and is conformally invariant~\cite{CHI}. It was shown in~\cite{Griffiths} and later generalized in~\cite{KS} that these correlation functions satisfy natural inequalities, and in particular, an important question of characterizing correlation functions coming from the Ising model was raised in the appendix of~\cite{KS}. 

A starting point for our results was recent insightful work of Lis~\cite{Lis}, where he discovered a deep connection between the planar Ising model and total positivity, and used it to prove new inequalities on boundary two-point correlation functions in the planar case. He relied on the results of Postnikov~\cite{Pos} and Talaska~\cite{Talaska} on the \emph{totally nonnegative Grassmannian} $\Grtnn(k,n)$, which is the subset of the Grassmannian $\Gr(k,n)$ of $k$-dimensional subspaces of $\R^n$ where all Pl\"ucker coordinates are nonnegative. The space $\Grtnn(k,n)$ is a special case of the totally positive part of a partial flag variety introduced by Lusztig~\cite{Lus2,Lus98} as an application of his theory of canonical bases~\cite{LusztigBases}. The totally nonnegative Grassmannian was studied from combinatorial point of view in~\cite{Pos}, and since then it has attracted lots of attention due to its unexpected connections to various areas such as cluster algebras and the physics of scattering amplitudes~\cite{FZ,abcgpt,AHT,Scott_06}.

Despite the enormous amount of research on the planar Ising model, some basic questions seem to have remained unanswered. Let us denote by $\Space_n\subset \Mat_n(\R)$ the space of all boundary correlation matrices of planar Ising networks with $n$ boundary nodes embedded in a disk. This is a subspace of the space $\Mat_n(\R)$ of $n\times n$ matrices with real entries. Every matrix in $\Space_n$ is symmetric and has diagonal entries equal to $1$, but $\Space_n$ is neither a closed nor an open subset of the space of such matrices. Let $\Closure_n$ denote the closure of $\Space_n$ inside $\Mat_n(\R)$, i.e., $\Closure_n$ is the space of boundary correlation matrices of a slightly more general class of planar Ising networks, as discussed in Section~\ref{sec:ising_to_OG}. Two fundamental questions about $\Closure_n$ that we answer in this paper (see Theorem~\ref{thm:main}) are:
\begin{itemize}
\item Describe $\Closure_n$ by equalities and inequalities inside $\Mat_n(\R)$.
\item Describe the topology of $\Closure_n$.
\end{itemize}
Using a construction similar to the one in~\cite{Lis}, we give a simple embedding $\doublemap$ of the space $\Closure_n$ into a  subset of $\Grtnn(n,2n)$ which turns out to be precisely the \emph{totally nonnegative orthogonal Grassmannian}, introduced in~\cite{HW,HWX} in the study of ABJM scattering amplitudes. This gives a solution to the first question, and then we show that $\Closure_n$ is homeomorphic to an $n\choose2$-dimensional closed ball using the techniques developed in~\cite{GKL}, where an analogous result (conjectured by Postnikov~\cite{Pos}) was shown for $\Grtnn(k,n)$.

We then apply our construction to study some other aspects of the planar Ising model. For instance, we recognize (Theorem~\ref{thm:planar_dual}) the Kramers--Wannier duality~\cite{KrWa} as the well studied cyclic shift operation on $\Grtnn(n,2n)$, and explain (Remark~\ref{rmk:very_close}) the connection between the planar Ising model at critical temperature and the unique cyclically symmetric point inside $\Grtnn(n,2n)$, studied in~\cite{GKL,KarpCS}. We also express (Theorem~\ref{thm:generalized_Griffiths}) generalized Griffiths' inequalities of~\cite{Griffiths,KS} as manifestly positive linear combinations of the Pl\"ucker coordinates of our embedding. We explain in Corollary~\ref{cor:dimers} how the known formula for Pl\"ucker coordinates in terms of the dimer model recovers one of  Dub\'edat's \emph{bosonization identities}~\cite{Dubedat}. Finally, we solve the \emph{inverse problem} in Section~\ref{sec:inverse_problem_intro}: given a boundary correlation matrix $M\in \Mat_n(\R)$ of the Ising model on a planar graph $G$ embedded in a disk, we show that if $G$ is \emph{reduced} then the edge weights of the Ising model are uniquely and explicitly determined by $M$.

In many aspects, our results for the planar Ising model are analogous to known results for planar electrical networks, see e.g.~\cite{CGV,CIM,Lam,KenyonCDM}. However, the precise relationship between the two models remains completely mysterious to us. See Section~\ref{sec:conjectures} for a discussion of this and other open problems motivated by our approach.

This paper is organized as follows. We state our main result (Theorem~\ref{thm:main}) in Section~\ref{sec:main_results}, and then list several applications of our construction in Section~\ref{sec:applications}. We give some background on the totally nonnegative Grassmannian in Section~\ref{sec:tnn_OG}, and study the totally nonnegative orthogonal Grassmannian  in Section~\ref{sec:tnn_OG_2}. After that, we prove our main results. In Section~\ref{sec:ising_to_OG}, we show that the formula for boundary correlations in terms of the dimer model indeed yields the same result as the embedding $\doublemap$ from Section~\ref{sec:main_results}. In Section~\ref{sec:ball}, we prove that $\Closure_n$ is homeomorphic to a ball and discuss the cyclic symmetry of this space. We explain how to express generalized Griffiths' inequalities as positive sums of Pl\"ucker coordinates in Section~\ref{sec:griffiths}, and list several conjectures in Section~\ref{sec:conjectures}.

\SkipTocEntry\section*{Acknowledgments}
We thank Marcin Lis, David Speyer, Thomas Lam, Dmitry Chelkak, and George Lusztig for their valuable comments on the first version of this manuscript. We are also grateful to the anonymous referees for helpful suggestions.

\section{Main results}\label{sec:main_results}
\subsection{The Ising model}
A \emph{planar Ising network} is a pair $N=(G,J)$ where $G=(V,E)$ is a planar graph embedded in a disk and $J:E\to \R_{>0}$ is a function assigning positive real numbers to the edges of $G$.
We always label the vertices of $G$ on the boundary of the disk by $b_1,\dots,b_n\in V$ in counterclockwise order. Given a planar Ising network $N=(G,J)$, the \emph{Ising model} on $N$ (with \emph{no external field} and \emph{free boundary conditions}) is a probability measure on the space $\{-1,1\}^V$ of \emph{spin configurations} on the vertices of $G$. Given a spin configuration $\sigma:V\to\{-1,1\}$, its probability is given by
\begin{equation}\label{eq:dfn:Prob}
\Prob(\sigma):=\frac1{\Zpart} \prod_{\{u,v\}\in E}\exp \left( J_{\{u,v\}} \sigma_u\sigma_v\right),
\end{equation}
where $\Zpart$ is the \emph{partition function}:
\begin{equation}\label{eq:dfn:Zpart}
\Zpart:=\sum_{\sigma\in\{-1,1\}^V}  \prod_{\{u,v\}\in E}\exp \left( J_{\{u,v\}} \sigma_u\sigma_v\right).
\end{equation}
Our main focus will be \emph{boundary two-point correlation functions}. Let $[n]:=\{1,2,\dots,n\}$. Given $i,j\in[n]$, we define the corresponding correlation function by
\begin{equation}\label{eq:dfn:Corr}
\<\sigma_i\sigma_j\>:=\sum_{\sigma\in\{-1,1\}^V}\Prob(\sigma)\sigma_{b_i}\sigma_{b_j}.
\end{equation}
Clearly, we have $\<\sigma_i\sigma_j\>=\<\sigma_j\sigma_i\>$, and if $i=j$ then the correlation function $\<\sigma_i\sigma_i\>$ is equal to $1$. We denote by $\Matsym\subset\Mat_n(\R)$ the space of all $n\times n$ symmetric real matrices with ones on the diagonal. Thus we obtain a matrix $M=M(G,J):=(m_{i,j})\in\Matsym$ given by $m_{i,j}:=\<\sigma_i\sigma_j\>$. Let us denote 
\[\Space_n:=\{M(G,J)\mid (G,J)\text{ is a planar Ising network with $n$ boundary vertices}\}.\]
Denote by $\Closure_n$ the closure of $\Space_n$ in the space $\Mat_n(\R)$ of $n\times n$ real matrices. (In other words, $\Closure_n$ can be defined as the space of all boundary correlation matrices $M(G,J)$ where $J$ is allowed to take values in $[0,\infty]$, or equivalently where $G$ is obtained from a planar graph embedded in a disk by contracting some edges that may connect boundary vertices, as we discuss in Section~\ref{sec:ising_to_OG}.) We will see later (Proposition~\ref{prop:Closure_cells}) that $\Closure_n$ admits a natural stratification into cells indexed by \emph{matchings on $[2n]$}, that is, by perfect matchings of the complete graph $K_{2n}$ (also called \emph{medial pairings}). For $n=3$, all matchings on $[2n]$ are shown in Figure~\ref{fig:P_3}.

\subsection{The orthogonal Grassmannian}
The \emph{Grassmannian} $\Gr(k,n)$ is the space of $k$-dimensional linear subspaces of $\R^n$. We always view an element $X\in\Gr(k,n)$ as the row span of a real $k\times n$ matrix of rank $k$, thus we think of $\Gr(k,n)$ as the space of full rank $k\times n$ matrices modulo row operations. Let us denote by ${[n]\choose k}$ the set of all $k$-element subsets of $[n]$. Given a set $I\in{[n]\choose k}$ and a point $X\in\Gr(k,n)$, the corresponding \emph{Pl\"ucker coordinate} $\Delta_I(X)$ is defined to be the determinant of the $k\times k$  submatrix of $X$ with column set $I$. (Such determinants are also called \emph{maximal minors} of $X$.) Pl\"ucker coordinates are defined up to a simultaneous rescaling, giving rise to the \emph{Pl\"ucker embedding} of $\Gr(k,n)$ into the $\left({n\choose k}-1\right)$-dimensional real projective space, see e.g.~\cite[Section~2.4]{RedBook}.

Define the \emph{totally nonnegative Grassmannian} $\Grtnn(k,n)\subset \Gr(k,n)$ as follows:
\[\Grtnn(k,n):=\left\{X\in\Gr(k,n)\mid \Delta_I(X)\geq0 \text{ for all $I\in{[n]\choose k}$}\right\}.\]
\begin{definition}\label{dfn:OG}
  The \emph{orthogonal Grassmannian} $\OG(n,2n)\subset \Gr(n,2n)$ is defined by
  \[\OG(n,2n):=\left\{X\in\Gr(n,2n)\mid \Delta_I(X)=\Delta_{[2n]\setminus I}(X) \text{ for all $I\in{[2n]\choose n}$}\right\}.\]
  Its \emph{totally nonnegative part} $\OGtnn(n,2n)\subset \Grtnn(n,2n)$ is the intersection
  \[\OGtnn(n,2n):=\OG(n,2n)\cap\Grtnn(n,2n).\]
\end{definition}
The space $\OGtnn(n,2n)$ has been first considered in~\cite{HW} in the context of the scattering amplitudes of ABJM theory. Postnikov defined a stratification of $\Grtnn(k,n)$ into \emph{positroid cells}, which induces a stratification of $\OGtnn(n,2n)$. As it was observed in~\cite{HW,HWX}, the strata of $\OGtnn(n,2n)$ are also naturally labeled by matchings on $[2n]$. We prove this in Section~\ref{sec:tnn_OG_2}.

\subsection{An embedding}\label{sec:embedding}

 Given a matrix $M=(m_{i,j})\in\Matsym$, one can construct an element $\doublemap(M)\in\OG(n,2n)$ using the following rules. We will describe an $n\times 2n$ matrix $\double M=(\double m_{i,j})$, so that for all $i,j\in[n]$, each of $\double m_{i,2j-1}$ and $\double m_{i,2j}$ is equal to either  $m_{i,j}$ or $-m_{i,j}$, as in Figure~\ref{fig:double}. Explicitly, for $i=j$ we put $\double m_{i,2i-1}=\double m_{i,2i}=m_{i,i}=1$, and for $i\neq j$ we set
\begin{equation}\label{eq:double_signs}
\double m_{i,2j-1}=-\double m_{i,2j}=(-1)^{i+j+\one(i<j)}m_{i,j},
\end{equation}
where $\one(i<j)$ denotes $1$ if $i<j$ and $0$ otherwise. 
\begin{figure}
\scalebox{0.86}{
  $   \displaystyle M=\begin{pmatrix}
    1 & m_{12} & m_{13} & m_{14} \\
m_{12} & 1 & m_{23} & m_{24} \\
m_{13} & m_{23} & 1 & m_{34} \\
m_{14} & m_{24} & m_{34} & 1
  \end{pmatrix}\quad \mapsto\quad \double M=
\begin{pmatrix}
1 & 1 & m_{12} & -m_{12} & -m_{13} & m_{13} & m_{14} & -m_{14} \\
-m_{12} & m_{12} & 1 & 1 & m_{23} & -m_{23} & -m_{24} & m_{24} \\
m_{13} & -m_{13} & -m_{23} & m_{23} & 1 & 1 & m_{34} & -m_{34} \\
-m_{14} & m_{14} & m_{24} & -m_{24} & -m_{34} & m_{34} & 1 & 1
\end{pmatrix} $}
  \caption{\label{fig:double}An example of the map $M\mapsto \double M$ for $n=4$.}
\end{figure}

\begin{remark}\label{rmk:full_rank}
For each $i\in[n]$, the sum of columns $2i-1$ and $2i$ of $\double M$ is equal to $2 e_i$, where $e_i$ is the $i$-th standard basis vector in $\R^n$. Thus the matrix $\double M$ has full rank, and we denote by $\doublemap(M)\in \Gr(n,2n)$ its row span.
\end{remark}

One can check that in fact $\doublemap(M)$ belongs to $\OG(n,2n)$, see Corollary~\ref{cor:doublemap_subset_OG}. We have thus constructed a map $\doublemap: \Matsym\to\OG(n,2n)$. Since boundary correlation matrices of planar Ising networks belong to the space $\Matsym$, $\doublemap$ restricts to a map $\doublemap:\Closure_n\to\OG(n,2n)$. We are ready to state our main result.

\begin{theorem}\label{thm:main}
The restriction $\doublemap:\Closure_n\to\OG(n,2n)$ is a stratification-preserving homeomorphism between $\Closure_n$ and $\OGtnn(n,2n)$. Moreover, both spaces are homeomorphic to an $n\choose 2$-dimensional closed ball.
\end{theorem}

We prove the second part of Theorem~\ref{thm:main} in Section~\ref{sec:ball}, where we also deduce its first part from Theorems~\ref{thm:Gmed_parametrization} and~\ref{thm:X=X'}.

\begin{remark}
  The second sentence of  Theorem~\ref{thm:main} is an application of the machinery developed in~\cite{GKL}. The fact that the image $\doublemap(\Closure_n)$ is a subset of $\Grtnn(n,2n)$ can be deduced in a straightforward fashion from the work of Lis~\cite{Lis}, see Section~\ref{sec:alternating_flows}.
\end{remark}

\begin{example}
  We illustrate Theorem~\ref{thm:main} in the case $n=2$. Let $\sigma_{12}:=\<\sigma_1\sigma_2\>$, then the boundary correlation matrix $M$ has the form
\begin{equation}\label{eq:M_n=2}
  M=\begin{pmatrix}
      1 & \sigma_{12}\\
      \sigma_{12} & 1
    \end{pmatrix}.
\end{equation}
By definition, $\sigma_{12}\leq 1$, and we also have $\sigma_{12}\geq0$ by one of the \emph{Griffiths' inequalities}~\cite[Theorem~1]{Griffiths}. In fact, if $G$ has a single edge connecting the vertices $b_1$ and $b_2$ then it is easy to check that depending on $J_{\{b_1,b_2\}}$, $\sigma_{12}$ can be any number strictly between $0$ and $1$. If we remove the edge $\{b_1,b_2\}$ from $G$, we get $\sigma_{12}=0$. Thus $\Space_n$ consists of all matrices $M$ of the form~\eqref{eq:M_n=2} for $0\leq \sigma_{12}< 1$. If we contract the edge $\{b_1,b_2\}$, we get $\sigma_{12}=1$. The resulting graph will no longer be embedded in a disk, because the boundary vertices $b_1$ and $b_2$ will get identified. This is an example of a \emph{generalized planar Ising network} that we introduce in Section~\ref{sec:ising_to_OG}. We see that the closure $\Closure_n$ of $\Space_n$ consists of all matrices $M$ of the form~\eqref{eq:M_n=2} for $0\leq \sigma_{12}\leq 1$, and is stratified into three cells $\{\sigma_{12}=0\}$, $\{0<\sigma_{12}<1\}$, and $\{\sigma_{12}=1\}$. These three cells correspond to three possible matchings on $\{1,2,3,4\}$, namely, $\{\{1,2\},\{3,4\}\}$, $\{\{1,3\},\{2,4\}\}$, and $\{\{1,4\},\{2,3\}\}$, respectively.

We have
\begin{equation}\label{eq:double_M_example}
\double M=\begin{pmatrix}
    1 & 1 & \sigma_{12} & -\sigma_{12}\\
    -\sigma_{12} & \sigma_{12} & 1 & 1
  \end{pmatrix},
\end{equation}
and $\doublemap(M)\in \Gr(n,2n)$ is the row span of $\double M$. The maximal minors of $\double M$ are
\[\Delta_{12}(\double M)=\Delta_{34}(\double M)=2\sigma_{12},\quad \Delta_{14}(\double M)=\Delta_{23}(\double M)=1-\sigma_{12}^2,\]
\[\Delta_{13}(\double M)=\Delta_{24}(\double M)=1+\sigma_{12}^2.\]
It follows that $\doublemap(M)$ belongs to $\OG(n,2n)$ for all $\sigma_{12}\in\R$, and moreover, we get $\doublemap(M)\in \OGtnn(n,2n)$ precisely when $0\leq \sigma_{12}\leq 1$. Note that $\double M$ is a matrix but $\doublemap(M)$ is an element of the Grassmannian, and thus the Pl\"ucker coordinates of $\doublemap(M)$ are only defined up to rescaling. Nevertheless, we can recover $\sigma_{12}$ from these minors as follows:
\begin{equation}\label{eq:sigma_12_in_terms_of_minors}
\sigma_{12}=\frac{\Delta_{12}(\doublemap(M))}{\Delta_{13}(\doublemap(M))+\Delta_{14}(\doublemap(M))}.
\end{equation}
Thus we see that for $n=2$, the map $\doublemap$ is indeed a homeomorphism, and both spaces $\Closure_n$ and $\OGtnn(n,2n)$ are homeomorphic to $[0,1]$, which is an ${n\choose 2}=1$-dimensional closed ball.
\end{example}

\section{Consequences of the main construction}\label{sec:applications}
In this section, we give further results on the relationship between the Ising model and the orthogonal Grassmannian.

\newcommand\OddEven[1]{\mathcal{E}_n(#1)}
\newcommand\OddEvenPrime[1]{\mathcal{E}'_n(#1)}
\subsection{Reconstructing correlations from minors}
Our first goal is, given an element $X\in\OGtnn(n,2n)$, to find explicitly a matrix $M=(m_{i,j})\in\Matsym$ such that $X$ is the row span of $\double M$. For the case $n=2$, this was done in~\eqref{eq:sigma_12_in_terms_of_minors}. In order to deal with the general case, we give the following important definition.

\begin{definition}\label{dfn:OddEven}
Given a subset $S\subset[n]$, we denote by $\OddEven S\subset {[2n]\choose n}$ the collection of $n$-element subsets $I$ of $[2n]$ such that for each $i\in [n]$, the intersection $I\cap\{2i-1,2i\}$ has even size if and only if $i\in S$.
\end{definition}

The following result, proved in Section~\ref{sec:ising_to_OG}, is a simple consequence of Remark~\ref{rmk:full_rank}.
\begin{lemma}\label{lemma:from_minors_to_correlations}
  Let $M=(m_{i,j})\in\Matsym$ be a matrix. Then for each $i,j\in[n]$, we have
  \begin{equation}\label{eq:from_minors_to_correlations}
m_{i,j}=\frac{\sum_{I\in \OddEven{\{i,j\}}}\Delta_I(\doublemap(M))}{\sum_{I\in \OddEven{\emptyset}}\Delta_I(\doublemap(M))}=2^{-n}\sum_{I\in \OddEven{\{i,j\}}}\Delta_I(\double M).
  \end{equation}
\end{lemma}
We stress again that unlike $\double M$, the maximal minors of $\doublemap(M)$ are defined up to a common scalar, so it only makes sense to talk about their ratios. However, for the specific matrix $\double M$, we have
\begin{equation}\label{eq:sum_of_minors_2^n}
\sum_{I\in \OddEven{\emptyset}}\Delta_I(\double M)=2^n,
\end{equation}
by the multilinearity of the determinant, see Remark~\ref{rmk:full_rank}. 
 Thus~\eqref{eq:sum_of_minors_2^n} explains why the two expressions for $m_{i,j}$ given in~\eqref{eq:from_minors_to_correlations} are actually equal.

For example, for $n=2$, \eqref{eq:sum_of_minors_2^n} becomes
\[\Delta_{13}(\double M)+\Delta_{14}(\double M)+\Delta_{23}(\double M)+\Delta_{24}(\double M)=4,\]
and for $i=1$ and $j=2$, Lemma~\ref{lemma:from_minors_to_correlations} gives another expression for $\sigma_{12}$:
\[\sigma_{12}=\frac{\Delta_{12}(\doublemap(M))+\Delta_{34}(\doublemap(M))}{\Delta_{13}(\doublemap(M))+\Delta_{14}(\doublemap(M))+\Delta_{23}(\doublemap(M))+\Delta_{24}(\doublemap(M))},\]
which is easily seen to be equivalent to~\eqref{eq:sigma_12_in_terms_of_minors}.

\subsection{Cyclic symmetry and the Kramers--Wannier duality}\label{sec:cyclic_intro}
A nice application of Theorem~\ref{thm:main} is a cyclic symmetry of the space $\Closure_n$, which comes from the cyclic symmetry of $\OGtnn(n,2n)$. It turns out that the cyclic shift operation on $\OGtnn(n,2n)$ corresponds to the  Kramers-Wannier duality~\cite{KrWa} that switches between the \emph{high and low temperature expansions} for the Ising model.

Let $k\leq \n$, and consider a linear operator $S:\R^\n\to \R^\n$ mapping a row vector $v=(v_1,\dots,v_\n)\in\R^\n$ to $v\cdot S=(v_2,v_3,\dots,v_\n,(-1)^{k-1}v_1)$. As a matrix, $S$ is given by $S_{i+1,i}=1$ for $i\in[\n-1]$, and $S_{1,\n}=(-1)^{k-1}$. A simple observation is that multiplying a $k\times \n$ matrix $A$ with nonnegative maximal minors by $S$ on the right yields another $k\times \n$ matrix with nonnegative maximal minors. Since multiplication on the right commutes with the left $\GL_k(\R)$-action, we get a cyclic shift operator on $\Grtnn(k,\n)$ mapping $X\in\Grtnn(k,\n)$ to $X\cdot S\in\Grtnn(k,\n)$. It is clear from the definitions that for $\Grtnn(n,2n)$, this action restricts to a cyclic shift action on $\OGtnn(n,2n)$. For example, if $X\in\OGtnn(2,4)$ is the row span of the matrix $\double M$ given in~\eqref{eq:double_M_example} then $X\cdot S$ is represented by 
\begin{equation}\label{eq:XcdotS_example}
\double M\cdot S=\begin{pmatrix}
    1 & \sigma_{12} & -\sigma_{12} & -1\\
  \sigma_{12} & 1 & 1 &    \sigma_{12}
\end{pmatrix}.
\end{equation}
One can check that the row span $X\cdot S$ of this matrix again belongs to $\OGtnn(2,4)$. By Theorem~\ref{thm:main}, there must exist a matrix $M'=\begin{pmatrix}
1 & \sigma'_{12}\\
\sigma'_{12} & 1
\end{pmatrix}$ such that $\doublemap(M')=X\cdot S$ in $\OGtnn(2,4)$ (i.e., such that $\double{M'}$ is obtained from the matrix in~\eqref{eq:XcdotS_example} by row operations). The value of $\sigma'_{12}$ can be found from the minors of $X\cdot S$ using~\eqref{eq:sigma_12_in_terms_of_minors}:
\[\sigma'_{12}=\frac{\Delta_{12}(X\cdot S)}{\Delta_{13}(X\cdot S)+\Delta_{14}(X\cdot S)}=\frac{1-\sigma_{12}^2}{1+\sigma_{12}^2+2\sigma_{12}}=\frac{1-\sigma_{12}}{1+\sigma_{12}}.\]
Thus the cyclic shift operation on $\OGtnn(n,2n)$ yields an automorphism of $\Closure_n$ which has order $2n$ for $n>2$ and order $n$ for $n=1,2$. For $n=2$, it sends $\sigma_{12}$ to $\frac{1-\sigma_{12}}{1+\sigma_{12}}$.

\def\Ndual{N^\ast}
\def\Mdual{M^\ast}
\def\Gdual{G^\ast}
\def\Vdual{V^\ast}
\def\Jdual{J^\ast}
\def\Edual{E^\ast}
\def\bdual{b^\ast}
\def\edual{{e^\ast}}

Let us now recall the construction of the duality of~\cite{KrWa}.
\begin{definition}\label{dfn:dual}
  Let $N=(G,J)$ be a connected\footnote{This definition can be easily extended to all (not necessarily connected) generalized planar Ising networks.} planar Ising network. The \emph{dual planar Ising network} $\Ndual:=(\Gdual,\Jdual)$ is defined as follows. The graph $\Gdual=(\Vdual,\Edual)$ is the planar dual graph of $G$, with boundary vertices $\bdual_1,\dots,\bdual_n$ placed counterclockwise on the boundary of the disk so that $\bdual_i$ is between $b_i$ and $b_{i+1}$. For $e\in E$, we denote by $\edual$ the edge of $\Gdual$ that crosses $e\in E$, and thus we have $\Edual=\{\edual\mid e\in E\}$. The edge parameters $\Jdual_\edual\in\R_{>0}$ are defined uniquely by the condition that
\begin{equation}\label{eq:duality}
\sinh(2\Jdual_\edual)=\frac1{\sinh(2J_e)}
\end{equation}
for all $e\in E$.
\end{definition}
For example, if $G$ is the graph in Figure~\ref{fig:Gbip_big} (left) then its dual $\Gdual$ is shown in Figure~\ref{fig:Gbip_big} (middle). Note that we have $\sinh(2t)=1/\sinh(2t)$ if and only if $t=\frac12\log(\sqrt2+1)$ is the \emph{critical temperature} of the Ising model on the square lattice. We also remark that applying the duality twice yields the same planar Ising network except that its boundary vertex labels are cyclically shifted: $(\bdual_i)^\ast=b_{i+1}$. We prove the following result in Section~\ref{sec:ball}.

\begin{theorem}\label{thm:planar_dual}
  Let $N=(G,J)$ be a connected planar Ising network with dual planar Ising network $\Ndual:=(\Gdual,\Jdual)$. Then the correlation matrices $M:=M(G,J)$ and $\Mdual:=M(\Gdual,\Jdual)$ are related by the cyclic shift on $\OGtnn(n,2n)$:
  \[\doublemap(M)\cdot S=\doublemap(\Mdual).\]
\end{theorem}

According to~\cite[Theorem~1.1]{GKL}, the space $\Grtnn(k,\n)$ is homeomorphic to a closed ball. The main ingredient of the proof of this result is the cyclic symmetry of $\Grtnn(k,\n)$. In Section~\ref{sec:ball}, we use the above cyclic symmetry of $\OGtnn(n,2n)$ in a similar way to show that it is a closed ball, which by the first part of Theorem~\ref{thm:main} implies that the space $\Closure_n$ of boundary correlation matrices is homeomorphic to a closed ball as well.

The following fact can be found in~\cite[Lemma~3.1]{GKL} or~\cite[Eq.~(1.4)]{KarpCS}, where it is stated more generally for all $\Grtnn(k,\n)$.
\begin{proposition}\label{prop:X_0}
  There exists a unique point $X_0\in\Grtnn(n,2n)$ that is \emph{cyclically symmetric}, i.e., satisfies $X_0\cdot S=X_0$. Its Pl\"ucker coordinates $\Delta_I(X_0)$ for $I\in{[2n]\choose n}$ are given by
  \begin{equation}\label{eq:sines}
\Delta_I(X_0)=\prod_{i,j\in I: i<j}\sin \left(\frac{(j-i)\pi}{2n}\right).
  \end{equation}
\end{proposition}

It will follow from our proof of Theorem~\ref{thm:main} that this point $X_0$ actually belongs to $\OGtnn(n,2n)$, and thus corresponds to some special planar Ising network with $n$ boundary vertices. For instance, for $n=2$ this is the Ising network $N$ with one edge $e$ such that $J_e=\frac12\log(\sqrt2+1)$. This planar Ising network is \emph{self-dual}, i.e., satisfies $N=\Ndual$. However, it is easy to see that for $n=3$ there are no self-dual planar Ising networks. Nevertheless, as our next result shows, for each $n$, there exists a (usually not unique) planar Ising network $N=(G,J)$ with $n$ boundary vertices such that the boundary correlation matrices of $N$ and $\Ndual$ coincide: $M(G,J)=M(\Gdual,\Jdual)$.
\begin{proposition}\label{prop:cyclic_fixed_pt}
For each $n\geq1$, there exists a unique boundary correlation matrix $M_0\in\Closure$ of some planar Ising network such that the element $\doublemap(M_0)\in\OGtnn(n,2n)$ is \emph{cyclically symmetric}, i.e., satisfies $\doublemap(M_0)\cdot S=\doublemap(M_0)$. For any planar Ising network $N=(G,J)$ satisfying $M(G,J)=M_0$, we have $M(G,J)=M(\Gdual,\Jdual)$.
\end{proposition}

See Section~\ref{sec:ball} for the proof.

\begin{remark}\label{rmk:very_close}
Consider a planar Ising network $N=(G,J)$ such that $G$ is the intersection of the square lattice of small side length $\delta$ with a disk, and let $J_e=\frac12\log(\sqrt2+1)$ be critical for all $e\in E$. The dual network $\Ndual=(\Gdual,\Jdual)$ is ``very close'' to $N$ in the sense that it is obtained by shifting $N$ by $(\delta/2,\delta/2)$ and making some adjustments near the boundary of the disk. Thus one could argue that the boundary correlations of $N$ are ``very close'' to being cyclically symmetric, in which case we can find them explicitly from~\eqref{lemma:from_minors_to_correlations} and~\eqref{eq:sines}. The notion of being ``very close'' is asymptotic and thus is left beyond the scope of this paper. We plan to apply this approach to studying the universality of the scaling limit as $\delta\to0$ in future work. 
\end{remark}

\subsection{Reduction to the dimer model}\label{sec:dimer_model_intro}
Lemma~\ref{lemma:from_minors_to_correlations} shows that each two-point correlation function is a ratio of two sums of minors of an element of $\OGtnn(n,2n)$. In the next section, we apply a well known result that each minor of an element of $\Grtnn(k,\n)$ is equal to a weighted sum of matchings in a certain planar bipartite graph.

\def\sech{\operatorname{sech}}
\def\tanh{\operatorname{tanh}}
Suppose that we are given a planar Ising network $N=(G,J)$. We introduce two functions $s,c:E\to (0,1)$ satisfying $s_e^2+c_e^2=1$ for each $e\in E$, as follows. Given $e\in E$, we set
\begin{equation}\label{eq:sinh_tanh}
s_e:=\sech(2J_e)=\frac{2}{\exp(2J_e)+\exp(-2J_e)};\quad c_e:=\tanh(2J_e)=\frac{\exp(2J_e)-\exp(-2J_e)}{\exp(2J_e)+\exp(-2J_e)}.
\end{equation}

    \begin{figure}
      
\renewcommand{\arraystretch}{1.6}

\begin{tabular}{ccc}
\scalebox{\sclbx}{
  \begin{tikzpicture}[baseline=(zero.base),scale=\tikzscl]
    \coordinate (zero) at (0,0);
    \disk{0,0}
    \edge{-1,0}{1,0}
    \node[anchor=south,scale=\scsc] (etop) at (0,0) {$e$};
    \vertx{V1}{-1,0}
    \vertx{V2}{1,0}
    \node[anchor=east,scale=\bscl] (b1) at (-1,0) {$b_2$};
    \node[anchor=west,scale=\bscl] (b2) at (1,0) {$b_1$};
  \end{tikzpicture}
  }
&
\scalebox{\sclbx}{
  \begin{tikzpicture}[baseline=(zero.base),scale=\tikzscl]
    \drawgraph
    \edgepl{A}{AA}
    \edgepl{B}{BB}
    \edgepl{C}{CC}
    \edgepl{D}{DD}
    \edgepl{A}{B}
    \edgepl{C}{B}
    \edgepl{C}{D}
    \edgepl{A}{D}
  \end{tikzpicture}}
&
\scalebox{\sclbx}{
  \begin{tikzpicture}[baseline=(zero.base),scale=\tikzscl]
    \coordinate (zero) at (0,0);
    \disk{0,0}
    \edgeop{-1,0}{1,0}
    \node[anchor=30,scale=\bscl] (d1) at (210:1) {$\bbip_4$};
    \node[anchor=-30,scale=\bscl] (d2) at (-210:1) {$\bbip_3$};
    \node[anchor=210,scale=\bscl] (d3) at (30:1) {$\bbip_2$};
    \node[anchor=-210,scale=\bscl] (d3) at (-30:1) {$\bbip_1$};

    \vertx{AA}{30:1}
    \vertx{BB}{-30:1}
    \vertx{CC}{-150:1}
    \vertx{DD}{150:1}
    \vertx{zero}{0,0}
    \edgemed{AA}{zero}
    \edgemed{BB}{zero}
    \edgemed{CC}{zero}
    \edgemed{DD}{zero}
  \end{tikzpicture}
}
\\

$G$
&
$\Gbip$
&
$\Gmed$
\\

\end{tabular}

  \caption{\label{fig:Gbip_small}Transforming a graph $G$ (left) with one edge $e$ connecting two boundary vertices $b_1$ and $b_2$ into a bipartite graph $\Gbip$ (middle) with eight edges. Four of those edges are incident to the boundary and have weight $1$, and the rest have weights $s_e$, $c_e$, $s_e$, $c_e$, as shown in the figure. The corresponding medial graph $\Gmed$ from Section~\ref{sec:inverse_problem_intro} is shown on the right.}
\end{figure}
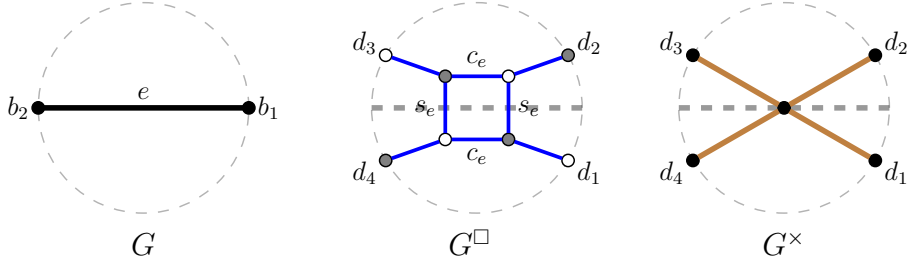

Next, we transform $G$ into a weighted planar bipartite graph (a \emph{plabic graph} in the sense of~\cite{Pos}) $\Gbip$ embedded in a disk, as in Figures~\ref{fig:Gbip_small} and~\ref{fig:Gbip_big}: we replace each edge $e\in E$ of $G$ by a bipartite square as in Figure~\ref{fig:Gbip_small} (middle), and connect two such squares if the corresponding edges of $G$ share both a vertex and a face of $G$. Additionally, we connect each of the  $2n$ boundary vertices of $\Gbip$, which we label $d_1,\dots,d_{2n}$ in counterclockwise order, to a unique vertex of $\Gbip$ in the interior of the disk in an obvious way (as in Figure~\ref{fig:Gbip_big}). Thus $d_i$ is a white (resp., black) vertex if $i$ is odd (resp., even). See Definitions~\ref{dfn:medial_to_plabic} and~\ref{dfn:Ising_to_medial} for a precise description of the rules for constructing $\Gbip$ from $G$. We call $\Gbip$ \emph{the plabic graph associated with $G$}.

\def\match{\mathcal{A}}
Let us now describe the \emph{boundary measurement map} of~\cite{Pos,Talaska}, as explained in~\cite{LamCDM}.
\begin{definition}\label{dfn:almost_perfect}
An \emph{almost perfect matching} of $\Gbip$ is a collection $\match$ of edges of $\Gbip$ such that every vertex of $\Gbip$ is incident to at most one edge in $\match$, and every non-boundary vertex of $\Gbip$ is incident to exactly one edge in $\match$. The \emph{boundary} of $\match$ is a subset $\partial(\match)\subset[2n]$ which consists of all odd indices $i$ such that $d_i$ is not incident to an edge of $\match$ together with all even indices $i$ such that $d_i$ is incident to an edge of $\match$. We define the \emph{weight} $\wt(\match)$ of $\match$ to be the product of weights of all edges in $\match$.
\end{definition}

It is not hard to see that $\partial(\match)$ has size $n$ for any almost perfect matching $\match$ of $\Gbip$. We are prepared to give a formula for the boundary correlation functions which is very similar to Kenyon and Wilson's \emph{grove measurement} formula~\cite{KW}. See Section~\ref{sec:ball} for the proof.

\begin{theorem}\label{thm:bound_measurements}
  Let $N=(G,J)$ be a planar Ising network and $M=M(G,J)$ be its boundary correlation matrix. Consider the element  $\doublemap(M)\in\OGtnn(n,2n)$, and let $\Gbip$ be the weighted planar bipartite graph described above. Then up to a common rescaling,  for every $I\in{[2n]\choose n}$ we have
  \begin{equation}\label{eq:dimer_model}
\Delta_I(\doublemap(M))=\sum_{\match:\partial(\match)=I}\wt(\match),
  \end{equation}
  where the sum is over almost perfect matchings $\match$ of $\Gbip$ with boundary $I$.
\end{theorem}

\begin{figure}
  \def\tikzscl{1}
  \def\sclbx{1.4}

  \begin{tabular}{cccc}
 \scalebox{\sclbx}{
   \begin{tikzpicture}[baseline=(zero.base),scale=\tikzscl]
    \drawgraph
    \edgeplmatch{A}{AA}
    \edgeplmatch{D}{DD}
    \edgeplmatch{B}{C}
  \end{tikzpicture}
}    & 
 \scalebox{\sclbx}{
  \begin{tikzpicture}[baseline=(zero.base),scale=\tikzscl]
    \drawgraph
    \edgeplmatch{A}{B}
    \edgeplmatch{D}{C}
  \end{tikzpicture}
}    & 
 \scalebox{\sclbx}{
  \begin{tikzpicture}[baseline=(zero.base),scale=\tikzscl]
    \drawgraph
    \edgeplmatch{A}{D}
    \edgeplmatch{C}{B}
  \end{tikzpicture}
}    & 
 \scalebox{\sclbx}{
  \begin{tikzpicture}[baseline=(zero.base),scale=\tikzscl]
    \drawgraph
    \edgeplmatch{C}{CC}
    \edgeplmatch{D}{DD}
    \edgeplmatch{A}{B}
  \end{tikzpicture}
}    \\
    $\partial(\match)=\{1,2\}$ &     $\partial(\match)=\{1,3\}$ &     $\partial(\match)=\{1,3\}$ &     $\partial(\match)=\{1,4\}$\\
    $\wt(\match)=c_e$ &      $\wt(\match)=s_e^2$ &      $\wt(\match)=c_e^2$ &      $\wt(\match)=s_e$ 
  \end{tabular}
  
  \caption{\label{fig:Gbip_matchings}Some almost perfect matchings of $\Gbip$, together with their boundaries and weights.}
\end{figure}
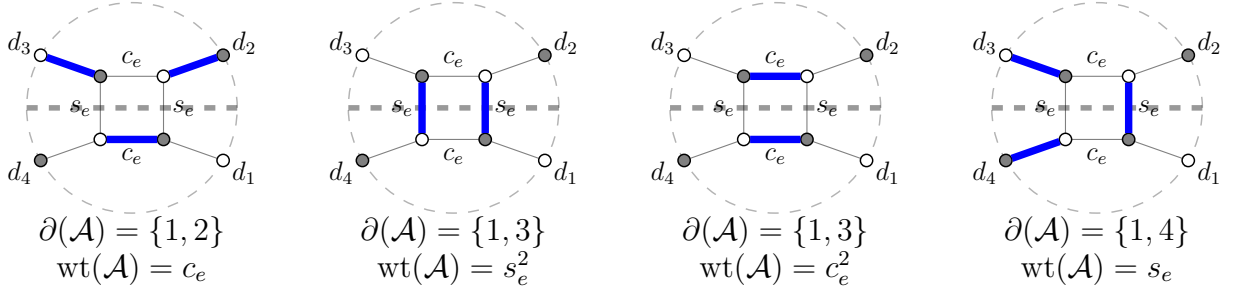

For example, consider the graph $\Gbip$ in Figure~\ref{fig:Gbip_small} (middle). There is a single almost perfect matching of $\Gbip$ with boundary $\{1,2\}$, shown in Figure~\ref{fig:Gbip_matchings} (left). Similarly, there are two almost perfect matchings of $\Gbip$ with boundary $\{1,3\}$ and one almost perfect matching with boundary $\{1,4\}$, also shown in Figure~\ref{fig:Gbip_matchings}. Therefore by Theorem~\ref{thm:bound_measurements} we get
\[\Delta_{12}(\doublemap(M))=c_e,\quad \Delta_{13}(\doublemap(M))=c_e^2+s_e^2=1,\quad \Delta_{14}(\doublemap(M))=s_e,\]
up to a common rescaling. By~\eqref{eq:sigma_12_in_terms_of_minors}, we should have
\begin{equation}\label{eq:sigma_12_s_c}
\sigma_{12}=\frac{\Delta_{12}(\doublemap(M))}{\Delta_{13}(\doublemap(M))+\Delta_{14}(\doublemap(M))}=\frac{c_e}{1+s_e},
\end{equation}
where $s_e=\sech(2J_e)$ and $c_e=\tanh(2J_e)$ are expressed in terms of $J_e$ as in~\eqref{eq:sinh_tanh}. Simplifying the expressions, we get
\begin{equation}\label{eq:sigma_12_tanh}
\sigma_{12}=\frac{\exp(J_e)-\exp(-J_e)}{\exp(J_e)+\exp(-J_e)}.
\end{equation}
On the other hand, by the definition of the Ising model, the partition function is equal to $Z=2(\exp(J_e)+\exp(-J_e))$ and thus the correlation $\<\sigma_1\sigma_2\>$ is
\[\<\sigma_1\sigma_2\>=\frac2Z \left(\exp(J_e)-\exp(-J_e)\right)=\frac{\exp(J_e)-\exp(-J_e)}{\exp(J_e)+\exp(-J_e)},\]
in agreement with Theorem~\ref{thm:bound_measurements}.

Lemma~\ref{lemma:from_minors_to_correlations} and Theorem~\ref{thm:bound_measurements} together recover the following elegant formula of~\cite{Dubedat} for expressing boundary correlations in terms of almost perfect matchings.\footnote{We thank Marcin Lis for bringing the paper~\cite{Dubedat} to our attention. We also thank one of the referees for explaining to us that Corollary~\ref{cor:dimers} is a simple consequence of the results of~\cite{Dubedat}. See Section~\ref{sec:Dubedat} for further discussion.}

\begin{corollary}\label{cor:dimers}
  Let $N=(G,J)$ be a planar Ising network. Then for all $i\neq j\in[n]$, the corresponding boundary correlation function is given by
  \begin{equation}\label{eq:from_dimers_to_correlations}
    \<\sigma_i\sigma_j\>=\frac{\sum_{\match:\partial(\match)\in\OddEven{\{i,j\}}}\wt(\match)}
    {\sum_{\match:\partial(\match)\in\OddEven{\emptyset}}\wt(\match)},
  \end{equation}
  where the sums in the numerator and the denominator are over almost perfect matchings $\match$ of $\Gbip$.
\end{corollary}

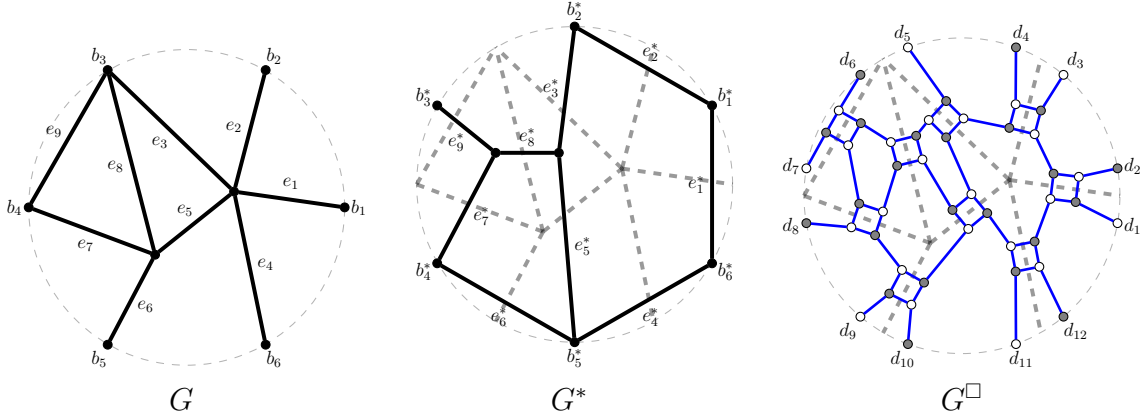
\begin{figure}

\begin{tabular}{ccc}
\scalebox{1.0}{
\begin{tikzpicture}[scale=2.1]
\coordinate (zero) at (0:0.00);
\disk{0,0}
\vertx{V0}{1.00,0.00}
\node[anchor=180, scale=\bscl] (b0) at (0:1) {$b_1$};
\vertx{V1}{0.50,0.87}
\node[anchor=240, scale=\bscl] (b1) at (60:1) {$b_2$};
\vertx{V2}{-0.50,0.87}
\node[anchor=300, scale=\bscl] (b2) at (120:1) {$b_3$};
\vertx{V3}{-1.00,0.00}
\node[anchor=360, scale=\bscl] (b3) at (180:1) {$b_4$};
\vertx{V4}{-0.50,-0.87}
\node[anchor=420, scale=\bscl] (b4) at (240:1) {$b_5$};
\vertx{V5}{0.50,-0.87}
\node[anchor=480, scale=\bscl] (b5) at (300:1) {$b_6$};
\vertx{V6}{0.30,0.10}
\vertx{V7}{-0.20,-0.30}
\edge{V6}{V0}
\node[scale=\scsc] (E0) at (0.66,0.15) {$e_{1}$};
\edge{V6}{V1}
\node[scale=\scsc] (E1) at (0.30,0.51) {$e_{2}$};
\edge{V6}{V2}
\node[scale=\scsc] (E2) at (-0.17,0.41) {$e_{3}$};
\edge{V6}{V5}
\node[scale=\scsc] (E3) at (0.50,-0.36) {$e_{4}$};
\edge{V7}{V6}
\node[scale=\scsc] (E4) at (-0.01,-0.02) {$e_{5}$};
\edge{V7}{V4}
\node[scale=\scsc] (E5) at (-0.26,-0.63) {$e_{6}$};
\edge{V7}{V3}
\node[scale=\scsc] (E6) at (-0.64,-0.24) {$e_{7}$};
\edge{V7}{V2}
\node[scale=\scsc] (E7) at (-0.45,0.26) {$e_{8}$};
\edge{V3}{V2}
\node[scale=\scsc] (E8) at (-0.84,0.48) {$e_{9}$};
\end{tikzpicture}}
&
\scalebox{1.0}{
\begin{tikzpicture}[scale=2.1]
\coordinate (zero) at (0:0.00);
\disk{0,0}
\coordinate (V0) at (1.00,0.00);
\coordinate (V1) at (0.50,0.87);
\coordinate (V2) at (-0.50,0.87);
\coordinate (V3) at (-1.00,0.00);
\coordinate (V4) at (-0.50,-0.87);
\coordinate (V5) at (0.50,-0.87);
\coordinate (V6) at (0.30,0.10);
\coordinate (V7) at (-0.20,-0.30);
\edgeop{V6}{V0}
\edgeop{V6}{V1}
\edgeop{V6}{V2}
\edgeop{V6}{V5}
\edgeop{V7}{V6}
\edgeop{V7}{V4}
\edgeop{V7}{V3}
\edgeop{V7}{V2}
\edgeop{V3}{V2}
\vertx{VD0}{0.87,0.50}
\node[anchor=210, scale=\bscl] (b0) at (30:1) {$b_1^\ast$};
\vertx{VD1}{0.00,1.00}
\node[anchor=270, scale=\bscl] (b1) at (90:1) {$b_2^\ast$};
\vertx{VD2}{-0.87,0.50}
\node[anchor=330, scale=\bscl] (b2) at (150:1) {$b_3^\ast$};
\vertx{VD3}{-0.87,-0.50}
\node[anchor=390, scale=\bscl] (b3) at (210:1) {$b_4^\ast$};
\vertx{VD4}{-0.00,-1.00}
\node[anchor=450, scale=\bscl] (b4) at (270:1) {$b_5^\ast$};
\vertx{VD5}{0.87,-0.50}
\node[anchor=510, scale=\bscl] (b5) at (330:1) {$b_6^\ast$};
\vertx{VD6}{-0.10,0.20}
\vertx{VD7}{-0.50,0.20}
\edge{VD1}{VD0}
\node[scale=\scsc,anchor=center] (ED0) at (0.48,0.84) {$e_2^\ast$};
\edge{VD5}{VD0}
\node[scale=\scsc,anchor=center] (ED1) at (0.77,-0.00) {$e_1^\ast$};
\edge{VD6}{VD1}
\node[scale=\scsc,anchor=center] (ED2) at (-0.15,0.61) {$e_3^\ast$};
\edge{VD7}{VD6}
\node[scale=\scsc,anchor=center] (ED3) at (-0.30,0.30) {$e_8^\ast$};
\edge{VD7}{VD2}
\node[scale=\scsc,anchor=center] (ED4) at (-0.75,0.27) {$e_9^\ast$};
\edge{VD7}{VD3}
\node[scale=\scsc,anchor=center] (ED5) at (-0.59,-0.20) {$e_7^\ast$};
\edge{VD4}{VD3}
\node[scale=\scsc,anchor=center] (ED6) at (-0.48,-0.84) {$e_6^\ast$};
\edge{VD6}{VD4}
\node[scale=\scsc,anchor=center] (ED7) at (0.05,-0.39) {$e_5^\ast$};
\edge{VD5}{VD4}
\node[scale=\scsc,anchor=center] (ED8) at (0.48,-0.84) {$e_4^\ast$};
\end{tikzpicture}}
&
\scalebox{1.0}{

\begin{tikzpicture}[scale=2.1]
\coordinate (zero) at (0:0.00);
\disk{0,0}
\coordinate (V0) at (1.00,0.00);
\coordinate (V1) at (0.50,0.87);
\coordinate (V2) at (-0.50,0.87);
\coordinate (V3) at (-1.00,0.00);
\coordinate (V4) at (-0.50,-0.87);
\coordinate (V5) at (0.50,-0.87);
\coordinate (V6) at (0.30,0.10);
\coordinate (V7) at (-0.20,-0.30);
\edgeop{V6}{V0}
\edgeop{V6}{V1}
\edgeop{V6}{V2}
\edgeop{V6}{V5}
\edgeop{V7}{V6}
\edgeop{V7}{V4}
\edgeop{V7}{V3}
\edgeop{V7}{V2}
\edgeop{V3}{V2}
\vwh{VWH6EE0x0}{0.74,0.12}
\vbl{VBL6EE0x0}{0.72,-0.04}
\vwh{VWH6EE0x6}{0.56,-0.02}
\vbl{VBL6EE0x6}{0.58,0.14}
\edgepl{VWH6EE0x0}{VBL6EE0x0}
\edgepl{VBL6EE0x0}{VWH6EE0x6}
\edgepl{VWH6EE0x6}{VBL6EE0x6}
\edgepl{VBL6EE0x6}{VWH6EE0x0}
\vwh{VWH6EE1x1}{0.34,0.58}
\vbl{VBL6EE1x1}{0.50,0.54}
\vwh{VWH6EE1x6}{0.46,0.39}
\vbl{VBL6EE1x6}{0.30,0.43}
\edgepl{VWH6EE1x1}{VBL6EE1x1}
\edgepl{VBL6EE1x1}{VWH6EE1x6}
\edgepl{VWH6EE1x6}{VBL6EE1x6}
\edgepl{VBL6EE1x6}{VWH6EE1x1}
\vwh{VWH6EE2x2}{-0.21,0.48}
\vbl{VBL6EE2x2}{-0.10,0.60}
\vwh{VWH6EE2x6}{0.01,0.49}
\vbl{VBL6EE2x6}{-0.10,0.37}
\edgepl{VWH6EE2x2}{VBL6EE2x2}
\edgepl{VBL6EE2x2}{VWH6EE2x6}
\edgepl{VWH6EE2x6}{VBL6EE2x6}
\edgepl{VBL6EE2x6}{VWH6EE2x2}
\vwh{VWH6EE5x5}{0.49,-0.45}
\vbl{VBL6EE5x5}{0.34,-0.48}
\vwh{VWH6EE5x6}{0.31,-0.32}
\vbl{VBL6EE5x6}{0.46,-0.29}
\edgepl{VWH6EE5x5}{VBL6EE5x5}
\edgepl{VBL6EE5x5}{VWH6EE5x6}
\edgepl{VWH6EE5x6}{VBL6EE5x6}
\edgepl{VBL6EE5x6}{VWH6EE5x5}
\vwh{VWH7EE6x6}{0.06,0.01}
\vbl{VBL7EE6x6}{0.16,-0.11}
\vwh{VWH7EE6x7}{0.04,-0.21}
\vbl{VBL7EE6x7}{-0.06,-0.09}
\edgepl{VWH7EE6x6}{VBL7EE6x6}
\edgepl{VBL7EE6x6}{VWH7EE6x7}
\edgepl{VWH7EE6x7}{VBL7EE6x7}
\edgepl{VBL7EE6x7}{VWH7EE6x6}
\vwh{VWH7EE4x4}{-0.32,-0.69}
\vbl{VBL7EE4x4}{-0.46,-0.62}
\vwh{VWH7EE4x7}{-0.38,-0.47}
\vbl{VBL7EE4x7}{-0.24,-0.55}
\edgepl{VWH7EE4x4}{VBL7EE4x4}
\edgepl{VBL7EE4x4}{VWH7EE4x7}
\edgepl{VWH7EE4x7}{VBL7EE4x7}
\edgepl{VBL7EE4x7}{VWH7EE4x4}
\vwh{VWH7EE3x3}{-0.70,-0.20}
\vbl{VBL7EE3x3}{-0.65,-0.05}
\vwh{VWH7EE3x7}{-0.50,-0.10}
\vbl{VBL7EE3x7}{-0.55,-0.25}
\edgepl{VWH7EE3x3}{VBL7EE3x3}
\edgepl{VBL7EE3x3}{VWH7EE3x7}
\edgepl{VWH7EE3x7}{VBL7EE3x7}
\edgepl{VBL7EE3x7}{VWH7EE3x3}
\vwh{VWH7EE2x2}{-0.45,0.34}
\vbl{VBL7EE2x2}{-0.29,0.38}
\vwh{VWH7EE2x7}{-0.25,0.23}
\vbl{VBL7EE2x7}{-0.41,0.19}
\edgepl{VWH7EE2x2}{VBL7EE2x2}
\edgepl{VBL7EE2x2}{VWH7EE2x7}
\edgepl{VWH7EE2x7}{VBL7EE2x7}
\edgepl{VBL7EE2x7}{VWH7EE2x2}
\vwh{VWH3EE2x2}{-0.78,0.54}
\vbl{VBL3EE2x2}{-0.64,0.46}
\vwh{VWH3EE2x3}{-0.72,0.32}
\vbl{VBL3EE2x3}{-0.86,0.40}
\edgepl{VWH3EE2x2}{VBL3EE2x2}
\edgepl{VBL3EE2x2}{VWH3EE2x3}
\edgepl{VWH3EE2x3}{VBL3EE2x3}
\edgepl{VBL3EE2x3}{VWH3EE2x2}
\vwh{VWH0EE0x0}{-10:1}
\vbl{VBL0EE0x0}{10:1}
\node[anchor=170,scale=\scsc] (D1) at (-10:1) {$\bbip_{1}$};
\node[anchor=190,scale=\scsc] (D2) at (10:1) {$\bbip_{2}$};
\vwh{VWH1EE1x1}{50:1}
\vbl{VBL1EE1x1}{70:1}
\node[anchor=230,scale=\scsc] (D3) at (50:1) {$\bbip_{3}$};
\node[anchor=250,scale=\scsc] (D4) at (70:1) {$\bbip_{4}$};
\vwh{VWH2EE2x2}{110:1}
\vbl{VBL2EE2x2}{130:1}
\node[anchor=290,scale=\scsc] (D5) at (110:1) {$\bbip_{5}$};
\node[anchor=310,scale=\scsc] (D6) at (130:1) {$\bbip_{6}$};
\vwh{VWH3EE3x3}{170:1}
\vbl{VBL3EE3x3}{190:1}
\node[anchor=350,scale=\scsc] (D7) at (170:1) {$\bbip_{7}$};
\node[anchor=370,scale=\scsc] (D8) at (190:1) {$\bbip_{8}$};
\vwh{VWH4EE4x4}{230:1}
\vbl{VBL4EE4x4}{250:1}
\node[anchor=410,scale=\scsc] (D9) at (230:1) {$\bbip_{9}$};
\node[anchor=430,scale=\scsc] (D10) at (250:1) {$\bbip_{10}$};
\vwh{VWH5EE5x5}{290:1}
\vbl{VBL5EE5x5}{310:1}
\node[anchor=470,scale=\scsc] (D11) at (290:1) {$\bbip_{11}$};
\node[anchor=490,scale=\scsc] (D12) at (310:1) {$\bbip_{12}$};
\edgepl{VWH6EE0x0}{VBL0EE0x0}
\edgepl{VWH0EE0x0}{VBL6EE0x0}
\edgepl{VWH1EE1x1}{VBL6EE1x1}
\edgepl{VWH6EE1x1}{VBL1EE1x1}
\edgepl{VWH7EE2x2}{VBL3EE2x2}
\edgepl{VWH6EE2x2}{VBL7EE2x2}
\edgepl{VWH2EE2x2}{VBL6EE2x2}
\edgepl{VWH3EE2x2}{VBL2EE2x2}
\edgepl{VWH3EE2x3}{VBL7EE3x3}
\edgepl{VWH3EE3x3}{VBL3EE2x3}
\edgepl{VWH7EE3x3}{VBL3EE3x3}
\edgepl{VWH4EE4x4}{VBL7EE4x4}
\edgepl{VWH7EE4x4}{VBL4EE4x4}
\edgepl{VWH5EE5x5}{VBL6EE5x5}
\edgepl{VWH6EE5x5}{VBL5EE5x5}
\edgepl{VWH6EE5x6}{VBL7EE6x6}
\edgepl{VWH6EE0x6}{VBL6EE5x6}
\edgepl{VWH6EE1x6}{VBL6EE0x6}
\edgepl{VWH6EE2x6}{VBL6EE1x6}
\edgepl{VWH7EE6x6}{VBL6EE2x6}
\edgepl{VWH7EE6x7}{VBL7EE4x7}
\edgepl{VWH7EE2x7}{VBL7EE6x7}
\edgepl{VWH7EE3x7}{VBL7EE2x7}
\edgepl{VWH7EE4x7}{VBL7EE3x7}
\end{tikzpicture}
}
\\

$G$
&
$G^\ast$
&
$\Gbip$
\\

\end{tabular}

\caption{\label{fig:Gbip_big} A planar graph $G$ embedded in a disk (left), its dual $\Gdual$ (middle), and the corresponding plabic graph $\Gbip$ (right).}
\end{figure}


\subsection{Generalized Griffiths' inequalities}
As we have already noted, the fact that we have $J_e>0$ for all edges $e$ implies that all two-point correlation functions $\<\sigma_u\sigma_v\>$ are nonnegative~\cite{Griffiths}. Equation~\eqref{eq:from_minors_to_correlations} shows that $\<\sigma_u\sigma_v\>$ is a positive linear combination of the minors of $\double M$, and thus its nonnegativity follows from Theorem~\ref{thm:main}. More generally, for every subset $A\subset [n]$, define
\[\<\sigma_A\>:=\<\prod_{i\in A}\sigma_{b_i}\>=\sum_{\sigma\in\{-1,1\}^V}\Prob(\sigma)\prod_{i\in A}\sigma_{b_i}\]
  to be the expectation of the product of the spins in $A$. The following \emph{generalized Griffiths' inequalities} were proved in~\cite{KS}.
\begin{proposition}[{\cite{KS}}]\label{prop:generalized_Griffiths}
  For every $A\subset [n]$, we have
  \[\<\sigma_A\>\geq0.\]
  For every $A,B\subset[n]$, we have
  \[\<\sigma_A\sigma_B\>-\<\sigma_A\>\<\sigma_B\>\geq0.\]
\end{proposition}
\newcommand{\symdiff}[2]{#1\oplus #2}
\newcommand{\Sumparity}[2]{\mathcal{D}^{{#2}}(#1)}
\def\doubleA{\widetilde{A}}
Here $\<\sigma_A\sigma_B\>=\<\prod_{i\in \symdiff A B}\sigma_{b_i}\>$, where $\symdiff A B=(A\setminus B)\cup (B\setminus A)$ denotes the symmetric difference of $A$ and $B$. We note that the inequalities of Proposition~\ref{prop:generalized_Griffiths} hold more generally for not necessarily planar graphs $G$ and arbitrary subsets $A$ and $B$ of the vertex set of $G$.

The goal of our next result is to explain how both inequalities in Proposition~\ref{prop:generalized_Griffiths} also arise as positive linear combinations of the minors of the matrix $\double M$, where $M=M(G,J)$ is the boundary correlation matrix.
\begin{definition}\label{dfn:double_sumparity}
For $A\subset [n]$, we define $\doubleA:=\{2i-1\mid i\in A\}\cup \{2i\mid i\in A\}$, and for $\epsilon\in\{0,1\}$, we let $\Sumparity{A}{\epsilon}\subset{[2n]\choose n}$ be the set of all $I\in {[2n]\choose n}$ such that the sum of elements of $I\cap \doubleA$ is equal to $\epsilon$ modulo $2$.
\end{definition}
Recall also the notation $\OddEven{S}$ from Definition~\ref{dfn:OddEven}. We prove the following result in Section~\ref{sec:griffiths}.
\begin{theorem}\label{thm:generalized_Griffiths}
  For every $A\subset [n]$, we have
  \begin{equation}\label{eq:easy_Griffiths}
\<\sigma_A\>=2^{-n}\sum_{I\in \OddEven{A}}\Delta_I(\double M).
  \end{equation}
  For every $A,B\subset[n]$, there exists $\epsilon\in\{0,1\}$, given explicitly in~\eqref{eq:epsilon_Griffiths}, such that
  \begin{equation}\label{eq:Griffiths}
\<\sigma_A\sigma_B\>-\<\sigma_A\>\<\sigma_B\>=2^{-n+1}\sum_{I\in \OddEven{\symdiff{A}{B}}\cap \Sumparity{B}{\epsilon}}\Delta_I(\double M).
  \end{equation}
\end{theorem}

Thus the inequalities of Proposition~\ref{prop:generalized_Griffiths} become manifestly true when expressed in terms of minors of $\double M$, which are nonnegative by Theorem~\ref{thm:main}.

For example, when $n=2$ and $A=B=\{1,2\}$, we have $\epsilon=1$ by~\eqref{eq:epsilon_Griffiths}, and thus~\eqref{eq:Griffiths} becomes
\[1-\sigma_{12}^2=\frac{\Delta_{14}(\double M)+\Delta_{23}(\double M)}2.\]

\begin{figure}

\begin{tabular}{ccc}
\scalebox{1.0}{
\begin{tikzpicture}[scale=2.1]
\coordinate (zero) at (0:0.00);
\disk{0,0}
\vertx{V0}{1.00,0.00}
\node[anchor=180, scale=\bscl] (b0) at (0:1) {$b_1$};
\vertx{V1}{0.50,0.87}
\node[anchor=240, scale=\bscl] (b1) at (60:1) {$b_2$};
\vertx{V2}{-0.50,0.87}
\node[anchor=300, scale=\bscl] (b2) at (120:1) {$b_3$};
\vertx{V3}{-1.00,0.00}
\node[anchor=360, scale=\bscl] (b3) at (180:1) {$b_4$};
\vertx{V4}{-0.50,-0.87}
\node[anchor=420, scale=\bscl] (b4) at (240:1) {$b_5$};
\vertx{V5}{0.50,-0.87}
\node[anchor=480, scale=\bscl] (b5) at (300:1) {$b_6$};
\vertx{V6}{0.30,0.10}
\vertx{V7}{-0.20,-0.30}
\edge{V6}{V0}
\node[scale=\scsc] (E0) at (0.66,0.13) {$e_{1}$};
\edge{V6}{V1}
\node[scale=\scsc] (E1) at (0.31,0.51) {$e_{2}$};
\edge{V6}{V2}
\node[scale=\scsc] (E2) at (-0.19,0.39) {$e_{3}$};
\edge{V6}{V5}
\node[scale=\scsc] (E3) at (0.51,-0.36) {$e_{4}$};
\edge{V7}{V6}
\node[scale=\scsc] (E4) at (0.00,-0.04) {$e_{5}$};
\edge{V7}{V4}
\node[scale=\scsc] (E5) at (-0.29,-0.62) {$e_{6}$};
\edge{V7}{V3}
\node[scale=\scsc] (E6) at (-0.63,-0.24) {$e_{7}$};
\edge{V7}{V2}
\node[scale=\scsc] (E7) at (-0.48,0.25) {$e_{8}$};
\edge{V3}{V2}
\node[scale=\scsc] (E8) at (-0.85,0.49) {$e_{9}$};
\end{tikzpicture}}
&
\scalebox{1.0}{
\begin{tikzpicture}[scale=2.1]
\coordinate (zero) at (0:0.00);
\disk{0,0}
\coordinate (V0) at (1.00,0.00);
\coordinate (V1) at (0.50,0.87);
\coordinate (V2) at (-0.50,0.87);
\coordinate (V3) at (-1.00,0.00);
\coordinate (V4) at (-0.50,-0.87);
\coordinate (V5) at (0.50,-0.87);
\coordinate (V6) at (0.30,0.10);
\coordinate (V7) at (-0.20,-0.30);
\edgeop{V6}{V0}
\edgeop{V6}{V1}
\edgeop{V6}{V2}
\edgeop{V6}{V5}
\edgeop{V7}{V6}
\edgeop{V7}{V4}
\edgeop{V7}{V3}
\edgeop{V7}{V2}
\edgeop{V3}{V2}
\vertx{VMED6EE0}{0.65,0.05}
\vertx{VMED6EE1}{0.40,0.48}
\vertx{VMED6EE2}{-0.10,0.48}
\vertx{VMED6EE5}{0.40,-0.38}
\vertx{VMED7EE6}{0.05,-0.10}
\vertx{VMED7EE4}{-0.35,-0.58}
\vertx{VMED7EE3}{-0.60,-0.15}
\vertx{VMED7EE2}{-0.35,0.28}
\vertx{VMED3EE2}{-0.75,0.43}
\vertx{VMED0EE-1}{-10:1}
\vertx{VMED0EE0}{10:1}
\node[anchor=170,scale=\scsc] (D1) at (-10:1) {$\bbip_{1}$};
\node[anchor=190,scale=\scsc] (D2) at (10:1) {$\bbip_{2}$};
\vertx{VMED1EE-2}{50:1}
\vertx{VMED1EE1}{70:1}
\node[anchor=230,scale=\scsc] (D3) at (50:1) {$\bbip_{3}$};
\node[anchor=250,scale=\scsc] (D4) at (70:1) {$\bbip_{4}$};
\vertx{VMED2EE-3}{110:1}
\vertx{VMED2EE2}{130:1}
\node[anchor=290,scale=\scsc] (D5) at (110:1) {$\bbip_{5}$};
\node[anchor=310,scale=\scsc] (D6) at (130:1) {$\bbip_{6}$};
\vertx{VMED3EE-4}{170:1}
\vertx{VMED3EE3}{190:1}
\node[anchor=350,scale=\scsc] (D7) at (170:1) {$\bbip_{7}$};
\node[anchor=370,scale=\scsc] (D8) at (190:1) {$\bbip_{8}$};
\vertx{VMED4EE-5}{230:1}
\vertx{VMED4EE4}{250:1}
\node[anchor=410,scale=\scsc] (D9) at (230:1) {$\bbip_{9}$};
\node[anchor=430,scale=\scsc] (D10) at (250:1) {$\bbip_{10}$};
\vertx{VMED5EE-6}{290:1}
\vertx{VMED5EE5}{310:1}
\node[anchor=470,scale=\scsc] (D11) at (290:1) {$\bbip_{11}$};
\node[anchor=490,scale=\scsc] (D12) at (310:1) {$\bbip_{12}$};
\edgemed{VMED6EE0}{VMED0EE0}
\edgemed{VMED0EE-1}{VMED6EE0}
\edgemed{VMED1EE-2}{VMED6EE1}
\edgemed{VMED6EE1}{VMED1EE1}
\edgemed{VMED7EE2}{VMED3EE2}
\edgemed{VMED6EE2}{VMED7EE2}
\edgemed{VMED2EE-3}{VMED6EE2}
\edgemed{VMED3EE2}{VMED2EE2}
\edgemed{VMED3EE2}{VMED7EE3}
\edgemed{VMED3EE-4}{VMED3EE2}
\edgemed{VMED7EE3}{VMED3EE3}
\edgemed{VMED4EE-5}{VMED7EE4}
\edgemed{VMED7EE4}{VMED4EE4}
\edgemed{VMED5EE-6}{VMED6EE5}
\edgemed{VMED6EE5}{VMED5EE5}
\edgemed{VMED6EE5}{VMED7EE6}
\edgemed{VMED6EE0}{VMED6EE5}
\edgemed{VMED6EE1}{VMED6EE0}
\edgemed{VMED6EE2}{VMED6EE1}
\edgemed{VMED7EE6}{VMED6EE2}
\edgemed{VMED7EE6}{VMED7EE4}
\edgemed{VMED7EE2}{VMED7EE6}
\edgemed{VMED7EE3}{VMED7EE2}
\edgemed{VMED7EE4}{VMED7EE3}
\end{tikzpicture}}
&
\scalebox{1.0}{
\begin{tikzpicture}[scale=2.1]
\coordinate (zero) at (0:0.00);
\disk{0,0}
\coordinate (V0) at (1.00,0.00);
\coordinate (V1) at (0.50,0.87);
\coordinate (V2) at (-0.50,0.87);
\coordinate (V3) at (-1.00,0.00);
\coordinate (V4) at (-0.50,-0.87);
\coordinate (V5) at (0.50,-0.87);
\coordinate (V6) at (0.30,0.10);
\coordinate (V7) at (-0.20,-0.30);
\edgeop{V6}{V0}
\edgeop{V6}{V1}
\edgeop{V6}{V2}
\edgeop{V6}{V5}
\edgeop{V7}{V6}
\edgeop{V7}{V4}
\edgeop{V7}{V3}
\edgeop{V7}{V2}
\edgeop{V3}{V2}
\coordinate (VMED6EE0) at (0.65,0.05);
\coordinate (VMED6EE1) at (0.40,0.48);
\coordinate (VMED6EE2) at (-0.10,0.48);
\coordinate (VMED6EE5) at (0.40,-0.38);
\coordinate (VMED7EE6) at (0.05,-0.10);
\coordinate (VMED7EE4) at (-0.35,-0.58);
\coordinate (VMED7EE3) at (-0.60,-0.15);
\coordinate (VMED7EE2) at (-0.35,0.28);
\coordinate (VMED3EE2) at (-0.75,0.43);
\coordinate (VMED0EE-1) at (-10:1);
\coordinate (VMED0EE0) at (10:1);
\node[anchor=170,scale=\scsc] (D1) at (-10:1) {$\bbip_{1}$};
\node[anchor=190,scale=\scsc] (D2) at (10:1) {$\bbip_{2}$};
\coordinate (VMED1EE-2) at (50:1);
\coordinate (VMED1EE1) at (70:1);
\node[anchor=230,scale=\scsc] (D3) at (50:1) {$\bbip_{3}$};
\node[anchor=250,scale=\scsc] (D4) at (70:1) {$\bbip_{4}$};
\coordinate (VMED2EE-3) at (110:1);
\coordinate (VMED2EE2) at (130:1);
\node[anchor=290,scale=\scsc] (D5) at (110:1) {$\bbip_{5}$};
\node[anchor=310,scale=\scsc] (D6) at (130:1) {$\bbip_{6}$};
\coordinate (VMED3EE-4) at (170:1);
\coordinate (VMED3EE3) at (190:1);
\node[anchor=350,scale=\scsc] (D7) at (170:1) {$\bbip_{7}$};
\node[anchor=370,scale=\scsc] (D8) at (190:1) {$\bbip_{8}$};
\coordinate (VMED4EE-5) at (230:1);
\coordinate (VMED4EE4) at (250:1);
\node[anchor=410,scale=\scsc] (D9) at (230:1) {$\bbip_{9}$};
\node[anchor=430,scale=\scsc] (D10) at (250:1) {$\bbip_{10}$};
\coordinate (VMED5EE-6) at (290:1);
\coordinate (VMED5EE5) at (310:1);
\node[anchor=470,scale=\scsc] (D11) at (290:1) {$\bbip_{11}$};
\node[anchor=490,scale=\scsc] (D12) at (310:1) {$\bbip_{12}$};
\draw[color=red,line width=1pt] (VMED0EE0) to[in=20,out=180] (VMED6EE0);
\draw[color=red,line width=1pt] (VMED6EE0) to[in=59,out=200] (VMED6EE5);
\draw[color=red,line width=1pt] (VMED6EE5) to[in=480,out=239] (VMED5EE-6);
\draw[color=blue,line width=1pt] (VMED0EE-1) to[in=326,out=180] (VMED6EE0);
\draw[color=blue,line width=1pt] (VMED6EE0) to[in=300,out=506] (VMED6EE1);
\draw[color=blue,line width=1pt] (VMED6EE1) to[in=240,out=480] (VMED1EE1);
\draw[color=green!50!black,line width=1pt] (VMED1EE-2) to[in=49,out=240] (VMED6EE1);
\draw[color=green!50!black,line width=1pt] (VMED6EE1) to[in=360,out=229] (VMED6EE2);
\draw[color=green!50!black,line width=1pt] (VMED6EE2) to[in=38,out=540] (VMED7EE2);
\draw[color=green!50!black,line width=1pt] (VMED7EE2) to[in=59,out=218] (VMED7EE3);
\draw[color=green!50!black,line width=1pt] (VMED7EE3) to[in=360,out=239] (VMED3EE3);
\draw[color=brown,line width=1pt] (VMED2EE2) to[in=90,out=300] (VMED3EE2);
\draw[color=brown,line width=1pt] (VMED3EE2) to[in=104,out=270] (VMED7EE3);
\draw[color=brown,line width=1pt] (VMED7EE3) to[in=119,out=284] (VMED7EE4);
\draw[color=brown,line width=1pt] (VMED7EE4) to[in=420,out=299] (VMED4EE4);
\draw[color=black,line width=1pt] (VMED2EE-3) to[in=117,out=300] (VMED6EE2);
\draw[color=black,line width=1pt] (VMED6EE2) to[in=104,out=297] (VMED7EE6);
\draw[color=black,line width=1pt] (VMED7EE6) to[in=50,out=284] (VMED7EE4);
\draw[color=black,line width=1pt] (VMED7EE4) to[in=420,out=230] (VMED4EE-5);
\draw[color=violet,line width=1pt] (VMED3EE-4) to[in=180,out=360] (VMED3EE2);
\draw[color=violet,line width=1pt] (VMED3EE2) to[in=159,out=360] (VMED7EE2);
\draw[color=violet,line width=1pt] (VMED7EE2) to[in=136,out=339] (VMED7EE6);
\draw[color=violet,line width=1pt] (VMED7EE6) to[in=141,out=316] (VMED6EE5);
\draw[color=violet,line width=1pt] (VMED6EE5) to[in=480,out=321] (VMED5EE5);
\end{tikzpicture}}
\\

$G$
&
$\Gmed$
&
$\medpa_G$
\\

\end{tabular}

  \caption{\label{fig:Gmed} A planar graph $G$ embedded in a disk (left), the corresponding medial graph $\Gmed$ (middle) and its medial strands (right). }
\end{figure}
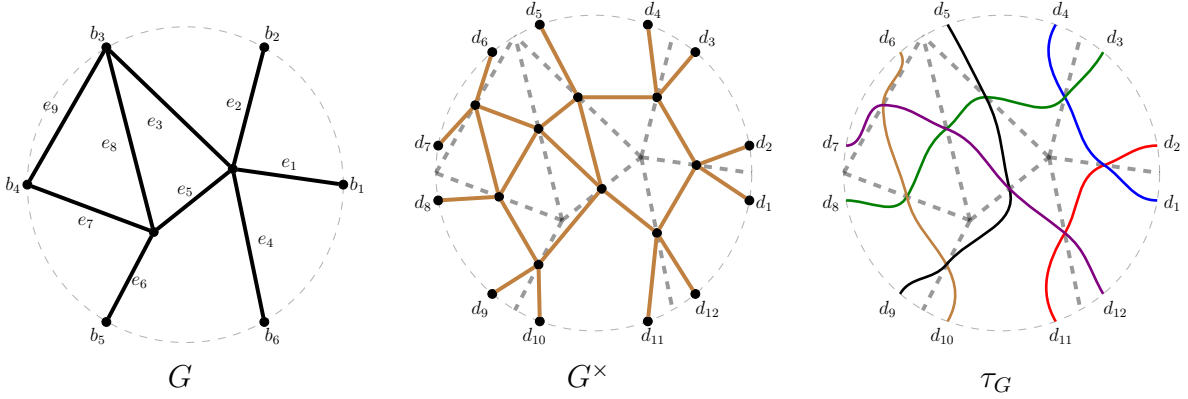

\subsection{Inverse problem}\label{sec:inverse_problem_intro}
In this section, we concentrate on answering the following question.
\begin{question}\label{question:inverse}
Given a planar Ising network $N=(G,J)$, is it possible to reconstruct $J$ from the matrix $M(G,J)$? In other words, is it true that the function $J:E\to \R_{>0}$ is uniquely determined by $G$ and $M(G,J)$?
\end{question} 
Of course, the answer to this question is negative if, for example, $G$ has more than $n\choose 2$ edges. In order to fix this, we introduce \emph{medial graphs}. Namely, given a planar graph $G$ embedded in a disk, the \emph{medial graph} $\Gmed$ associated with $G$ is a planar graph obtained from $G$ as in Figure~\ref{fig:Gbip_small} (right) and Figure~\ref{fig:Gmed} (middle). It has $2n$ boundary vertices $d_1,\dots,d_{2n}$, each of degree $1$, and $|E|$ interior vertices, each of degree $4$. See Section~\ref{sec:ising_to_OG} for a precise description.

Since each interior vertex of $\Gmed$ has degree $4$, we define a \emph{medial strand} in $\Gmed$ to be a path that starts at a boundary vertex $d_i$ of $\Gmed$, follows the only edge of $\Gmed$ incident to it, and then goes ``straight''  at each interior vertex of degree $4$, until it reaches another boundary vertex $d_j$ of $\Gmed$. (More precisely, a medial strand is determined by the condition that whenever two of its edges share a vertex, they do not share a face.) Clearly there are $n$ medial strands in $\Gmed$, and each of them connects $d_i$ to $d_j$ for some $i,j\in [2n]$, giving rise to a matching $\medpa_G$ on $[2n]$ called the  \emph{medial pairing associated with $G$}:
\[\medpa_G:=\left\{\{i,j\}\subset [2n]\mid \text{a medial strand in $\Gmed$ starts at $d_i$ and ends at $d_j$} \right\}.\]
Thus $\medpa_G$ is a partition of $[2n]$ into $n$ sets, each of size $2$. For example, for the graph $G$ in Figure~\ref{fig:Gmed} (left), there is a medial strand that starts at vertex $d_3$, then follows the midpoints of edges $e_2$, $e_3$, $e_8$, $e_7$, and then terminates at vertex $d_8$. Thus the medial pairing $\medpa_G$ of $G$ contains a pair $\{3,8\}$. Following the other five medial strands, we find
\begin{equation}\label{eq:medial_example}
\medpa_G=\left\{ \{1,4\},\{2,11\},\{3,8\},\{5,9\},\{6,10\},\{7,12\}  \right\},
\end{equation}
see Figure~\ref{fig:Gmed} (right).

\begin{definition}\label{dfn:xing}
For $i<j,i'<j'\in[2n]$, we say that pairs $\{i,j\}$ and $\{i',j'\}$ \emph{form a crossing} if either $i<i'<j<j'$ or $i'<i<j'<j$. For a matching $\medpa$ of $[2n]$, we let $\xing(\medpa)$ denote the number of pairs in $\medpa$ that form a crossing.
\end{definition}
For example, if $\medpa_G$ is given in~\eqref{eq:medial_example} then $\xing(\medpa_G)=9$. We are now ready to state an important definition, introduced in~\cite{CIM}.

\begin{definition}
We say that $G$ is \emph{reduced} if its number $|E|$ of edges equals $\xing(\medpa_G)$.
\end{definition}
For example, the graph $G$ in Figure~\ref{fig:Gmed} (left) is reduced since it has $|E|=9=\xing(\medpa_G)$ edges.

We will see later in Proposition~\ref{prop:Closure_cells} that if we fix $G$ and let $J$ vary, the space of matrices $M(G,J)$ obtained in such a way is an open ball of dimension $\xing(\medpa_G)$. Since $J$ varies over $\R_{>0}^E$, we see that if $G$ is not reduced (in which case clearly $|E|>\xing(\medpa_G)$) then the answer to Question~\ref{question:inverse} is negative. On the other hand, if $G$ is reduced then the answer turns out to be always positive, as we show in Section~\ref{sec:ball}.

\begin{theorem}\label{thm:reduced_injective}
  Let $N=(G,J)$ be a planar Ising network such that $G$ is reduced. Then the map $J\mapsto M(G,J)$ is injective, i.e., for each matrix $M\in\Matsym$, there is at most one function $J:E\to \R_{>0}$ such that $M=M(G,J)$.
\end{theorem}
Theorem~\ref{thm:main} combined with our results in Section~\ref{sec:tnn_OG_2}  gives a simple explicit way to test whether for a given reduced graph $G$ and a matrix $M\in\Matsym$ there exists a function $J:E\to \R_{>0}$ such that $M=M(G,J)$.

In order to describe a recursive way of reconstructing $J$ from $M(G,J)$, we introduce operations of \emph{adjoining a boundary spike} and \emph{adjoining a boundary edge} to $G$. An identical construction in the case of electrical networks has been considered in~\cite{CIM}.
\begin{definition}
Let $N'=(G',J')$ be a planar Ising network, where $G'=(V',E')$ has $n$ boundary vertices. Given $k\in[n]$, we say that another planar Ising network $N=(G,J)$ with $n$ boundary vertices is obtained from $N'$ by \emph{adjoining a boundary spike at $k$} if the vertex $b_k$ in $G$ is incident to a single edge $e$ and contracting this edge in $G$ yields $G'$. Similarly, we say that $N$ is obtained from $N'$ by \emph{adjoining a boundary edge} between $k$ and $k+1$ if $G$ contains an edge $e$ connecting $b_k$ and $b_{k+1}$, and removing this edge from $G$ yields $G'$. In both cases, we additionally require that the restriction of $J:E\to \R_{>0}$ to $E'=E\setminus\{e\}$ coincides with $J':E'\to\R_{>0}$.
\end{definition}
When adjoining boundary edges, we allow for $k=n$, in which case we set $k+1:=1$. We denote $t:=J_e$, and our first goal will be to reconstruct $t$ from the matrix $M(G,J)$.

\def\minimum{\operatorname{min}}
\def\maximum{\operatorname{max}}
\def\Imin{I^{\minimum}}
\def\Imax{I^{\maximum}}
\begin{definition}\label{dfn:prec_i}
  Let $i\in [2n]$. Consider a total order $\prec_{i}$ on $[2n]$ given by
  \[i\prec_i i+1\prec_i\dots\prec_i 2n\prec_i 1\prec_i\dots\prec_i i-1,\]
  where the indices are taken modulo $2n$. For a planar Ising network $N=(G,J)$ and $i\in[2n]$, define subsets $\Imin_i(G),\Imax_i(G)\in {[2n]\choose n}$ whose disjoint union is $[2n]$ as follows. For each unordered pair $\{a,b\}$ of $\medpa_G$, we may assume that $a\prec_i b$, and then we let $a\in \Imin_i(G)$ and $b\in \Imax_i(G)$. In particular, we always have $i\in \Imin_i(G)$ and $i-1\in\Imax_i(G)$. 
\end{definition}

For example, recall that if $G$ is the graph from Figure~\ref{fig:Gmed} (left) then $\medpa_G$ is given by~\eqref{eq:medial_example}. For $i=7$ and $i=12$, we have
\begin{equation}\label{eq:Imax_Imin_example}
  \begin{split}
    \Imin_7(G)&=\{7,8,9,10,11,1\},\quad \Imax_7(G)=\{6,5,4,3,2,12\},\\
    \Imin_{12}(G)&=\{12,1,2,3,5,6\},\quad \Imax_{12}(G)=\{11,10,9,8,7,5\}.
  \end{split}
\end{equation}

\def\kk{{\tilde{k}}}
We prove our next result in Section~\ref{sec:ball}.
\begin{theorem}\label{thm:inverse}\leavevmode
  Let $G$ be a reduced planar graph embedded in a disk.
  \begin{itemize}
  \item Suppose that $N=(G,J)$ is obtained from $N'=(G',J')$ by adjoining a boundary spike $e$ at $k\in[n]$. Let $M:=M(G,J)$, $\kk:=2k-1$, and $t:=J_e$. Then for $I:=\Imin_{\kk+1}(G)$, we have
    \[s_e=\sech(2t)=\frac{\Delta_{I}(\doublemap(M))}{\Delta_{I\cup\{\kk\}\setminus\{\kk+1\}}(\doublemap(M))}.\]
  \item Suppose that $N=(G,J)$ is obtained from $N'=(G',J')$ by adjoining a boundary edge $e$ between $k\in[n]$ and $k+1$. Let $M:=M(G,J)$, $\kk:=2k$, and $t:=J_e$. Then for $I:=\Imax_{\kk+1}(G)$, we have
    \[c_e=\tanh(2t)=\frac{\Delta_{I}(\doublemap(M))}{\Delta_{I\cup\{\kk+1\}\setminus\{\kk\}}(\doublemap(M))}.\]
  \end{itemize}
\end{theorem}
For example, the graph $G$ in Figure~\ref{fig:Gmed} (left) can be obtained from another reduced graph by adjoining a boundary spike $e_4$ at $k=6$, so we have $\kk=11$ in the first part of Theorem~\ref{thm:inverse}. Since $I=\Imin_{12}(G)=\{1,2,3,5,6,12\}$ by~\eqref{eq:Imax_Imin_example}, we have
\[s_{e_4}=\frac{\Delta_{\{1,2,3,5,6,12\}}(\doublemap(M))}{\Delta_{\{1,2,3,5,6,11\}}(\doublemap(M))}.\]
Similarly, $G$ can be obtained from another reduced graph by adjoining a boundary edge $e_9$ between $k=3$ and $k+1=4$, so we have $\kk=6$ in the second part of Theorem~\ref{thm:inverse}. Since $I=\Imax_{7}(G)=\{2,3,4,5,6,12\}$ by~\eqref{eq:Imax_Imin_example}, we have
\[c_{e_9}=\frac{\Delta_{\{2,3,4,5,6,12\}}(\doublemap(M))}{\Delta_{\{2,3,4,5,7,12\}}(\doublemap(M))}.\]

Since both functions $\sech,\tanh:(0,\infty)\to(0,1)$ are strictly monotone, it follows from Theorem~\ref{thm:inverse} that we can reconstruct $t=J_e$ uniquely from $M(G,J)$ whenever $e$ is either a boundary spike or a boundary edge of $G$. This constitutes the first step of our reconstruction algorithm, in view of the following result.

\begin{proposition}\label{prop:removable_spikes_edges}
Suppose that $G$ is a connected reduced planar graph embedded in a disk, having at least one edge. Then $G$ is obtained from another reduced graph $G'$ by adjoining either a boundary spike or a boundary edge.
\end{proposition}
We note that the graph $G'$ above need not be connected. Also, if $G$ itself is not connected then it is clearly enough to solve the inverse problem for each connected component of $G$ separately, and thus we may assume that $G$ is connected. See Lemma~\ref{lemma:exists_crossing} for a generalization of Proposition~\ref{prop:removable_spikes_edges}.

Proposition~\ref{prop:removable_spikes_edges} says that given a reduced graph $G$ and a matrix $M(G,J)$, we can reconstruct $J_e$ for at least one edge $e$ of $G$. A natural thing to do now would be to contract $e$ if it is a boundary spike and remove $e$ if it is a boundary edge, obtaining the reduced graph $G'$. Our next goal is to explain the relationship between the matrices $M(G,J)$ and $M(G',J')$ in the case when $N=(G,J)$ is obtained from $N'=(G',J')$ by adjoining either a boundary spike or a boundary edge.

We note that these two operations look like they have a very different effect on the boundary correlation matrix. For example, adjoining a boundary spike at $k$ only changes the correlation $\<\sigma_i\sigma_j\>$ when either $i$ or $j$ is equal to $k$, but adjoining a boundary edge between $k$ and $k+1$ in general changes all entries of the boundary correlation matrix. Surprisingly, these two operations have exactly the same form when written in terms of the matrix $\double M$, as we now explain. (In fact, it is clear that applying the duality from Section~\ref{sec:cyclic_intro} switches the roles of boundary spikes and boundary edges.)

Suppose that $N=(G,J)$ is obtained from $N'=(G',J')$ by adjoining a boundary spike $e$ at $k\in[n]$ (resp., a boundary edge $e$ between $k$ and $k+1$). Define $\kk:=2k-1$ (resp., $\kk:=2k$), $t:=J_e$,  $s_e:=\sech(2t)$, $c_e:=\tanh(2t)$, as in Theorem~\ref{thm:inverse}, and let $g=g_{\kk}(t)$ be a $2n\times 2n$ matrix which coincides with the identity matrix except that it contains a $2\times 2$ block $R_\kk$  in rows and columns indexed by $\kk$ and $\kk+1$. When $\kk$ is odd, we set $R_\kk:=\begin{pmatrix}
  1/c_e & s_e/c_e\\
  s_e/c_e & 1/c_e
\end{pmatrix}$ and when $\kk$ is even, we set $R_\kk:=\begin{pmatrix}
  1/s_e & c_e/s_e\\
  c_e/s_e & 1/s_e
\end{pmatrix}$. In the case where we have $\kk=2n$, i.e., when we are adding a boundary edge between $k=n$ and $k+1=1$, the relevant entries of $g$ are $g_{2n,2n}=g_{1,1}=1/s_e$ and $g_{1,2n}=g_{2n,1}=(-1)^{n-1}c_e/s_e$. This sign twist is related to the cyclic symmetry of $\Grtnn(n,2n)$ as we explained in Section~\ref{sec:cyclic_intro}. For example, for $n=2$ and $\kk=1,4$, we have
\[g_1(t)=\begin{pmatrix}
1/c_e & s_e/c_e & 0 & 0\\
s_e/c_e & 1/c_e & 0 & 0\\
0 & 0 & 1 & 0\\
0 & 0 & 0 & 1\\
\end{pmatrix},
\quad
g_4(t)=\begin{pmatrix}
1/s_e & 0 & 0 & -c_e/s_e\\
0 & 1 & 0 & 0\\
0 & 0 & 1 & 0\\
-c_e/s_e & 0 & 0 & 1/s_e
\end{pmatrix}.
\]

Recall that given an element $X\in\Gr(n,2n)$ and a $2n\times 2n$ invertible real matrix $g$, an element $X\cdot g\in\Gr(n,2n)$ is well defined as the row span of $A\cdot g$ where $A$ is any $n\times 2n$ matrix whose row span is $X$.

\begin{theorem}\label{thm:duality}
  Suppose that $N=(G,J)$ is obtained from $N'=(G',J')$ by adjoining a boundary spike $e$ at $k\in[n]$ (resp., a boundary edge $e$ between $k$ and $k+1$). Let $M=M(G,J)$, $M'=M(G',J')$, and $g_{\kk}(t)$ be as above. Then we have
  \[\doublemap(M)=\doublemap(M')\cdot g_{\kk}(t).\]
\end{theorem}

Theorems~\ref{thm:inverse} and~\ref{thm:duality} give the following inductive algorithm for reconstructing the function $J:E\to \R_{>0}$ for a given reduced graph $G=(V,E)$ from the matrix $M=M(G,J)$. The problem is trivial when $G$ has no edges. Otherwise by Proposition~\ref{prop:removable_spikes_edges}, there is either a boundary spike or a boundary edge $e$ in $G$. The matrix $M$ gives an element $\doublemap(M)\in\OGtnn(n,2n)$, from which we compute either $s_e$ or $c_e$ using Theorem~\ref{thm:inverse} and thus find $t=J_e$. After that, we contract $e$ in $G$ if it is a boundary spike and remove it if it is a boundary edge, and also modify the matrix $M$ accordingly: we let
\[X':=\doublemap(M)\cdot(g_{\kk}(t))^{-1}\in\OGtnn(n,2n),\]
where $(g_{\kk}(t))^{-1}$ can be found using
\[\begin{pmatrix}
  1/c_e & s_e/c_e\\
  s_e/c_e & 1/c_e
\end{pmatrix}^{-1}=\begin{pmatrix}
  1/c_e & -s_e/c_e\\
  -s_e/c_e & 1/c_e
\end{pmatrix},\quad \begin{pmatrix}
  1/s_e & c_e/s_e\\
  c_e/s_e & 1/s_e
\end{pmatrix}^{-1}=\begin{pmatrix}
  1/s_e & -c_e/s_e\\
  -c_e/s_e & 1/s_e
\end{pmatrix}.\]
By Lemma~\ref{lemma:from_minors_to_correlations}, we have $X'=\doublemap(M')\in\OGtnn(n,2n)$ for a unique matrix $M'\in\Matsym$. We then express the entries of $M'$ in terms of the Pl\"ucker coordinates of $X'$ using~\eqref{eq:from_minors_to_correlations}. It follows that this $n\times n$ matrix $M'$ is equal to $M(G',J')$, so we set $G:=G'$ and proceed recursively until $G$ has no edges left, splitting $G'$ into connected components if necessary. This finishes a constructive proof of Theorem~\ref{thm:inverse}. Alternatively, we deduce Theorem~\ref{thm:inverse} from Theorem~\ref{thm:main} at the end of Section~\ref{sec:ball}.

Another question similar to Question~\ref{question:inverse} is the following.
\begin{question}\label{question:tests}
Given an $n\times n$ matrix $M\in\Matsym$, does there exist a planar Ising network $N=(G,J)$ such that $M=M(G,J)$?
\end{question}
The answer to this question is provided by Theorem~\ref{thm:main}: the answer is ``yes'' if and only if all minors of the matrix $\double M$ are nonnegative. There are exponentially many minors to check, that is, $2n\choose n$, and in general one needs to check all of them to ensure that $\double M$ is totally nonnegative. However, checking whether $\doublemap(M)\in\OG(n,2n)$ belongs to $\OGtp(n,2n):=\OG(n,2n)\cap \Grtp(n,2n)$ (defined in~\eqref{eq:Grtp}), as opposed to $\OGtnn(n,2n)$, can be done in polynomial time. More precisely, one needs to check only $n^2+1$ minors of $\double M$, as it follows from the results of~\cite{Pos}. These minors are algebraically independent as functions on $\Gr(n,2n)$, but when restricted to $\OG(n,2n)$, this is no longer the case. Thus if all of them are positive then it follows that $\doublemap(M)\in\OGtp(n,2n)$, but in general one could check less minors and arrive at the same conclusion. See Section~\ref{sec:conjectures} for further discussion.

\section{Background on the totally nonnegative Grassmannian}\label{sec:tnn_OG}
In this section, we give a brief background on the totally nonnegative Grassmannian of Postnikov~\cite{Pos}. Most of the results in this section can be found in either~\cite{Pos} or~\cite{LamCDM}.

Recall that the totally nonnegative Grassmannian $\Grtnn(k,\n)$ is the subset of the real Grassmannian $\Gr(k,\n)$ where all Pl\"ucker coordinates are nonnegative. Given a point $X\in\Grtnn(k,\n)$, define the \emph{matroid} $\Mcal_X\subset{[\n]\choose k}$ of $X$ by
\[\Mcal_X:=\left\{J\in{[\n]\choose k}\mid \Delta_J(X)>0\right\}.\]
Given a collection $\Mcal\subset {[\n]\choose k}$, define the \emph{positroid cell} $\pc_\Mcal\subset\Grtnn(k,\n)$ by
\[\pc_\Mcal:=\{X\in\Grtnn(k,\n)\mid \Mcal_X=\Mcal\}.\]

For example, one can take $\Mcal={[\n]\choose k}$, in which case the positroid cell $\pc_\Mcal$ coincides with the \emph{totally positive Grassmannian} $\Grtp(k,n)$:
\begin{equation}\label{eq:Grtp}
\Grtp(k,\n):=\left\{X\in\Gr(k,\n)\mid \Delta_I(X)>0 \text{ for all $I\in{[\n]\choose k}$} \right\}.
\end{equation}

A collection $\Mcal\subset {[\n]\choose k}$ is called a \emph{positroid} if $\pc_\Mcal$ is nonempty. Positroids are special kinds of matroids, and have a very nice structure which we now explain.

Recall that for $i\in[\n]$, the total order $\prec_i$ on $[\n]$ is given by $i\prec_i i+1\prec_i\dots\prec_i \n\prec_i 1\prec_i \dots\prec_i i-1.$
\begin{definition}\label{dfn:Imin_Imax_Mcal}
For two sets $I,J\in {[\n]\choose k}$, we write $I\preceq_i J$ if $I=\{i_1\prec_i\dots\prec_i i_k\}$, $J=\{j_1\prec_i\dots\prec_i j_k\}$, and $i_s\preceq_i j_s$ for $1\leq s\leq k$. It turns out that if $\Mcal$ is a positroid then for each $i$ it has a unique $\preceq_i$-minimal element which we denote $\Imin_i(\Mcal)$. Thus $\Imin_i(\Mcal)$ satisfies $\Imin_i(\Mcal)\preceq_i J$ for all $J\in \Mcal$. Similarly, we let $\Imax_i(\Mcal)$ be the unique $\preceq_i$-maximal element of $\Mcal$.
\end{definition}

\begin{definition}\label{dfn:Gr_neck}
A sequence $\Ical:=(I_1,\dots,I_\n)$ of $k$-element subsets of $[\n]$ is called a \emph{Grassmann necklace} if for each $i\in [\n]$ there exists $j_i\in[\n]$ such that $I_{i+1}=I_i\setminus \{i\}\cup\{j_i\}$.
\end{definition}
Here (and everywhere in this section) the index $i+1$ is taken modulo $\n$.

\def\pidec{\pi}
\def\piinv{\pi^{-1}}
There is a simple bijection between positroids and Grassmann necklaces, which sends a positroid $\Mcal$ to the sequence $\Ical(\Mcal):=(\Imin_1(\Mcal), \Imin_2(\Mcal),\dots, \Imin_\n(\Mcal))$, which is a Grassmann necklace for each positroid $\Mcal$. Each Grassmann necklace $\Ical$ is encoded by an associated \emph{decorated permutation} $\pidec_\Ical:[\n]\to[\n]$ which sends $i\in[\n]$ to the index $j_i$ from Definition~\ref{dfn:Gr_neck}. (When $i$ is a fixed point of $\pidec_\Ical$, i.e., $\pidec_\Ical(i)=i$, there is an extra bit of data in $\pidec_\Ical$ recording whether $i\in I_i$ or $i\notin I_i$, but this will not be important for our exposition.)
The map $\Ical\mapsto \pidec_\Ical$ is a bijection between Grassmann necklaces and decorated permutations.

\begin{remark}\label{rmk:positroid_pidec}
Under the above correspondence, a positroid $\Mcal$ gives rise to a decorated permutation $\pidec_\Mcal$ such that for $i\in[\n]$,  $\pidec_\Mcal(i)$ is equal to the unique element of the set $\Imin_{i+1}(\Mcal)\setminus\Imin_i(\Mcal)$, if it is nonempty, and is equal to $i$ otherwise. It is not hard to see that $\piinv_\Mcal(i)$ is the unique element of the set $\Imax_i(\Mcal)\setminus \Imax_{i+1}(\Mcal)$.
\end{remark}

See~\cite[Section~16]{Pos} for a detailed description of these objects and bijections between them.

\def\Ebip{{E^\square}}
\def\Vbip{{V^\square}}
A \emph{plabic graph} is a planar bipartite graph $\Gbip=(\Vbip,\Ebip)$ embedded in a disk such that it has $\n$ boundary vertices $d_1,\dots,d_\n$, each of degree $1$. (Postnikov considers more general plabic graphs where vertices of the same color are allowed to be connected by an edge, but for our purposes it is sufficient to work with bipartite graphs.) Recall that the notion of an almost perfect matching is given in Definition~\ref{dfn:almost_perfect}. Given an almost perfect matching $\match$ of $G$, we define its \emph{boundary} $\partial(\match)\subset[\n]$ to be the set
\begin{equation*}
  \begin{split}
    \partial(\match)=&\{i\in[\n]\mid \text{$d_i$ is black and is not incident to an  edge of $\match$}\}\cup\\
                     &\{i\in[\n]\mid \text{$d_i$ is white and is incident to an  edge of $\match$}\}.
  \end{split}
\end{equation*}

It turns out that for every $\Gbip$ there exists an integer $0\leq k\leq \n$ such that every almost perfect matching $\match$ of $\Gbip$ satisfies $|\partial(\match)|=k$. The number $k$ is given explicitly in terms of the number of black and white vertices of $\Gbip$, see~\cite[Eq.~(9)]{LamCDM}.

\begin{definition}\label{dfn:pidec}
Each plabic graph $\Gbip$ gives rise to a decorated permutation $\pidec_\Gbip$, as follows. A \emph{strand} in $\Gbip$ is a path that turns maximally right (resp., maximally left) at each black (resp., white) vertex. Thus every edge of $\Gbip$ belongs to precisely two strands traversing it in opposite directions. If a strand that starts at $b_i$ ends at $b_j$ for some $i,j\in[\n]$ then we put $\pidec_\Gbip(i):=j$, which defines a decorated permutation $\pidec_\Gbip:[n]\to [n]$. (For each $i$ such that $\pidec_\Gbip(i)=i$, $\pidec_\Gbip$ also contains the information whether $i$ was black or white in $\Gbip$.)
\end{definition}
Since decorated permutations are in bijection with Grassmann necklaces and positroids, each plabic graph $\Gbip$ gives rise to a Grassmann necklace $\Ical_\Gbip$ and a positroid $\Mcal_\Gbip$.

\def\Meas{\operatorname{Meas}}
\def\Meast{{\overline \Meas}}
A \emph{weighted plabic graph} is a pair $(\Gbip,\wt)$ where $\Gbip$ is a plabic graph and $\wt:\Ebip\to \R_{>0}$ is a \emph{weight function} assigning positive real numbers to the edges of $\Gbip$. For an almost perfect matching $\match$ of $\Gbip$, recall that $\wt(\match)$ is the product of weights of edges of $\match$. We can consider a collection $\Meas(\Gbip,\wt):=(\Delta_I(\Gbip,\wt))_{I\in {[\n]\choose k}}\in\RP^{{\n\choose k}-1}$ of polynomials given for $I\in {[\n]\choose k}$ by
\begin{equation}\label{eq:Delta_I_matchings}
\Delta_I(\Gbip,\wt):=\sum_{\match:\partial(\match)=I} \wt(\match),
\end{equation}
where the sum is over all almost perfect matchings $\match$ of $\Gbip$ with boundary $I$. It turns out that $(\Delta_I(\Gbip,\wt))_{I\in {[\n]\choose k}}$ is the collection of Pl\"ucker coordinates of some point $X\in\Grtnn(k,\n)$. The following result is implicit in~\cite{PSW}.

\begin{theorem}[{\cite[Corollary~7.14]{LamCDM}}]\label{thm:Meas}
  Given a weighted plabic graph $(\Gbip,\wt)$, there exists a unique point $X\in\Gr(k,\n)$ such that
  \[\Meas(\Gbip,\wt)=(\Delta_I(X))_{I\in {[\n]\choose k}}\]
  as elements of the projective space $\RP^{{\n\choose k}-1}$. The point $X$ belongs to $\Grtnn(k,\n)$ and in fact to the positroid cell $\pc_{\Mcal_\Gbip}$, where $\Mcal_\Gbip$ is the positroid whose decorated permutation is $\pidec_\Gbip$. Every point $X\in \pc_{\Mcal_\Gbip}$ arises in this way from some weight function $\wt:\Ebip\to\R_{>0}$.
\end{theorem}
\def\Gauge{\operatorname{Gauge}}
The map $\Meas(\Gbip,\cdot):\R_{>0}^\Ebip\to \Grtnn(k,\n)$ sending $\wt \mapsto X$ is not usually injective. To see this, observe that each interior vertex of $\Gbip$ is incident to precisely one edge of each almost perfect matching $\match$. Thus rescaling the weights of all edges incident to a single interior vertex (i.e. applying a \emph{gauge transformation}) does not change the value of $\Meas$. We denote by $\R_{>0}^\Ebip/\Gauge$ the space of gauge-equivalence classes of functions $\wt:\Ebip\to\R_{>0}$, so that $\wt$ and $\wt'$ are the same in $\R_{>0}^\Ebip/\Gauge$ if and only if $\wt'$ can be obtained from $\wt$ by a sequence of gauge transformations. It is not hard to see that $\R_{>0}^\Ebip/\Gauge$ is homeomorphic to an open ball of dimension $F(\Gbip)-1$, where $F(\Gbip)$ denotes the number of faces of $\Gbip$.

Thus by Theorem~\ref{thm:Meas}, $\Meas$ gives rise to a map
\[\Meast: \R_{>0}^\Ebip/\Gauge\to \pc_{\Mcal_\Gbip}\subset \Grtnn(k,\n)\]
which turns out to be injective for some plabic graphs $\Gbip$. More precisely, let us say that $\Gbip$ is \emph{reduced} if all of the following conditions are satisfied:
\begin{itemize}
\item no strand in $\Gbip$ intersects itself;
\item there are no closed strands in $\Gbip$;
\item no two strands in $\Gbip$ have a \emph{bad double crossing}.
\end{itemize}
Here two strands are said to form a \emph{bad double crossing} if there are two vertices $u,v\in \Vbip$ such that both strands first pass through $u$ and then through $v$. The following result can be found in~\cite{Pos,LamCDM}.

\begin{theorem}\label{thm:Meast}
  For each positroid $\Mcal$, there exists a reduced plabic graph $\Gbip$ such that $\Mcal=\Mcal_\Gbip$. Given a reduced plabic graph $\Gbip$, the map $\Meast:\R_{>0}^\Ebip/\Gauge\to \pc_{\Mcal_\Gbip}$ is a homeomorphism. Thus the positroid cell $\pc_{\Mcal_\Gbip}$ is homeomorphic to $\R^{F(\Gbip)-1}$. In addition, we have
  \begin{equation}\label{eq:disjoint_Gr}
\Grtnn(k,\n)=\bigsqcup_{\Mcal}\pc_\Mcal,
  \end{equation}
  where the union is over all positroids $\Mcal\subset {[\n]\choose k}$.
\end{theorem}

The last ingredient from the theory of plabic graphs that we will need is \emph{BCFW bridges}, introduced in~\cite{abcgpt,BCFW}. Our exposition will follow~\cite[Section~7]{LamCDM}.

Recall that each boundary vertex of a plabic graph is incident to a unique edge.
\begin{definition}
  Given $i\in[\n]$, we say that a plabic graph $\Gbip$ has a \emph{removable bridge between $i$ and $i+1$} if there exists a path of length $3$ between $\bbip_i$ and $\bbip_{i+1}$ in $\Gbip$. (In particular, these vertices have to be of different color).
\end{definition}
Here we again allow $i=\n$ and $i+1=1$. We refer to the middle edge of this path of length $3$ as a \emph{bridge between $i$ and $i+1$}. There are two types of bridges, since $i$ can be incident either to a white or to a black interior vertex. It turns out that the weight of the bridge can always be recovered from the minors of the corresponding element of the Grassmannian. The following result can be found in~\cite[Proposition~7.10]{LamCDM} and~\cite[Proposition~3.10]{Lam}, and is the main ingredient of the proof of Theorem~\ref{thm:inverse}.
\begin{theorem}\label{thm:bridge_removal}
  Let $(\Gbip,\wt)$ be a weighted reduced plabic graph, and suppose that it has a removable bridge between $i$ and $i+1$. Assume that the weights of the edges incident to $\bbip_i$ and $\bbip_{i+1}$ are both equal to $1$ (which can always be achieved using gauge transformations). Let $e\in\Ebip$ be the bridge between $i$ and $i+1$, and denote $X:=\Meast(\Gbip,\wt)\in\Grtnn(k,\n)$.
  \begin{itemize}
  \item If $i$ is white then for $I:=\Imin_{i+1}(\Mcal_\Gbip)$, we have
    \[\wt(e)=\frac{\Delta_I(X)}{\Delta_{I\cup\{i\}\setminus \{i+1\}}(X)}.\]
  \item If $i$ is black then for $I:=\Imax_{i+1}(\Mcal_\Gbip)$, we have
    \[\wt(e)=\frac{\Delta_I(X)}{\Delta_{I\cup\{i+1\}\setminus \{i\}}(X)}.\]
  \end{itemize}
\end{theorem}

We will also need to explain how removing a bridge changes the corresponding element of the Grassmannian. For $i\in[\n-1]$ and $t\in\R$, define $x_i(t)\in\Mat_{\n}(\R)$ to be a $\n\times \n$ matrix with ones on the diagonal and a single nonzero off-diagonal entry in row $i$ and column $i+1$ equal to $t$. We also define $x_{\n}(t)$ to be the matrix with ones on the diagonal and the entry in row $\n$, column $1$ equal to $(-1)^{k-1}t$. We define $y_{i}(t)$ to be the matrix transpose of $x_i(t)$ for $i\in[\n]$.
\begin{lemma}[{\cite[Lemma~7.6]{LamCDM}}]\label{lemma:bridge_removal}
  Let $(\Gbip,\wt)$ be a weighted plabic graph, and suppose that it has a removable bridge between $i$ and $i+1$. Assume that the weights of the edges incident to $\bbip_i$ and $\bbip_{i+1}$ are both equal to $1$. Let $e\in\Ebip$ be the bridge between $i$ and $i+1$ with weight $\wt(e)=t$, and denote $X:=\Meast(\Gbip,\wt)\in\Grtnn(k,\n)$. Let $(\Gbip',\wt')$ be obtained from $(\Gbip,\wt)$ by removing $e$, and define $X':=\Meast(\Gbip',\wt)$. Then for all $I\in{[\n]\choose k}$ we have the following.
  \begin{itemize}
  \item If $i$ is white then $X'=X\cdot x_i(-t)$, and
    \[\Delta_I(X')=
      \begin{cases}
        \Delta_I(X)-t\Delta_{I\setminus\{i+1\}\cup \{i\}}(X), &\text{if $i+1\in I$ but $i\notin I$;}\\
        \Delta_I(X), &\text{otherwise.}\\
      \end{cases}\]
  \item If $i$ is black then $X'=X\cdot y_i(-t)$, and
    \[\Delta_I(X')=
      \begin{cases}
        \Delta_I(X)-t\Delta_{I\setminus\{i\}\cup \{i+1\}}(X), &\text{if $i\in I$ but $i+1\notin I$;}\\
        \Delta_I(X), &\text{otherwise.}\\
      \end{cases}\]
  \end{itemize}
\end{lemma}

\section{The totally nonnegative orthogonal Grassmannian}\label{sec:tnn_OG_2}
In this section, we discuss how the stratification of $\Gr(n,2n)$ induces a stratification of the totally nonnegative orthogonal Grassmannian $\OGtnn(n,2n)$. We remark that some of the statements below have appeared in~\cite{HW,HWX}, but mostly without proofs.

Recall from Definition~\ref{dfn:OG} that the orthogonal Grassmannian $\OG(n,2n)\subset\Gr(n,2n)$ is the set of $X\in\Gr(n,2n)$ such that $\Delta_I(X)=\Delta_{[2n]\setminus I}(X)$ for all $n$-element sets $I\subset[2n]$. In the literature, the term ``orthogonal Grassmannian'' usually refers to the set of subspaces where a certain non-degenerate symmetric bilinear form vanishes. Over the complex numbers, there is only one such bilinear form up to isomorphism, but over the real numbers, one needs to choose a signature. Following~\cite{HW}, define a non-degenerate symmetric bilinear form $\eta:\R^{2n}\times \R^{2n}\to \R$ by $\eta(u,v):=u_1v_1-u_2v_2+\dots+u_{2n-1}v_{2n-1}-u_{2n}v_{2n}.$ Let us also introduce another subset $\OG_-(n,2n)\subset\Gr(n,2n)$ consisting of all $X\in\Gr(n,2n)$ such that $\Delta_I(X)=-\Delta_{[2n]\setminus I}(X)$ for all $n$-element sets $I\subset[2n]$.\footnote{We thank David Speyer for suggesting to consider both $\OG(n,2n)$ and $\OG_-(n,2n)$.} We justify our terminology as follows.
\begin{proposition}\label{prop:eta}
  For a subspace $X\in\Gr(n,2n)$, the following are equivalent:
  \begin{itemize}
  \item $X\in \OG(n,2n)\sqcup \OG_-(n,2n)$;
  \item for any two vectors $u,v\in X\subset\R^{2n}$, we have $\eta(u,v)=0$.
  \end{itemize}
\end{proposition}
\begin{proof}
  Given an $k\times \n$ matrix $A=(a_{i,j})$, define another $k\times \n$ matrix $\alt(A):=((-1)^{j}a_{i,j})$. Taking row spans and setting $k:=n$, $\n:=2n$, we get a map $\alt: \Gr(n,2n)\to\Gr(n,2n)$. It is a classical result (see e.g.~\cite[Section~7]{Hochster} or~\cite[Lemma~1.11]{Karp}) that for $X\in\Gr(n,2n)$ and $I\in {[2n]\choose n}$, we have $\Delta_{[2n]\setminus I}(\alt(X^\perp))=c\Delta_I(X),$ where $\perp$ denotes the orthogonal complement of $X\subset\R^{2n}$ with respect to the standard scalar product $\<\cdot,\cdot\>$ on $\R^{2n}$ and $c\in\R$ is some nonzero constant. Note that $\eta(u,v)=\<\alt(u), v\>$ for $u,v\in\R^{2n}$, which shows that $\eta$ vanishes on $X$ if and only if $\Delta_I(X)=c\Delta_{[2n]\setminus I}(X)$ for all $I\in{[2n]\choose n}$. Applying this equality twice, we get $\Delta_I(X)=c\Delta_{[2n]\setminus I}(X)=c^2\Delta_I(X)$, and thus $c=\pm1$. We are done with the proof.
\end{proof}

\begin{remark}\label{rmk:lusztig}
Lusztig~\cite{Lus2} has defined the totally nonnegative part $(G/P)_{\geq0}$ of any partial flag variety $G/P$ inside a split reductive algebraic group $G$ over $\R$. Rietsch showed that the space $\Grtnn(k,n)$ is a special case of $(G/P)_{\geq0}$, see~\cite[Remark~3.8]{LamCDM}. For a specific choice of $G=O(n,n)$ (i.e. the \emph{split orthogonal group}, which corresponds to the Dynkin diagram of type $D_n$) and a maximal parabolic subgroup $P\cong \SL_n(\R)$ (corresponding to the Dynkin diagram of type $A_{n-1}$, obtained from $D_n$ by removing a leaf adjacent to a degree $3$ vertex), $G/P$ becomes equal to $\OG(n,2n)$. If we had $(G/P)_{\geq0}=\OGtnn(n,2n)$ then the fact that $\OGtnn(n,2n)$ is a closed ball would follow from the results of~\cite{GKL2}. However, the relationship between Lusztig's $(G/P)_{\geq0}$ and $\OGtnn(n,2n)$ remains unclear to us. For instance, the cell decomposition of $(G/P)_{\geq0}$ conjectured by Lusztig and proved by Rietsch~\cite{Rietsch_thesis,Rietsch2} appears to have a different number of cells than the cell decomposition of $\OGtnn(n,2n)$ indexed by matchings on $[2n]$ that we consider in this paper.
\end{remark}

\begin{remark}
A different relation between the Ising model and spin representations of the orthogonal group can be found in~\cite{Kaufman, SMJ, Palmer}.
\end{remark}

\begin{remark}
The generators $g_i(t)$ from Section~\ref{sec:inverse_problem_intro} belong to $O(n,n)$, and moreover, they are \emph{hyperbolic rotation matrices}, since for $t\in\R_{>0}$ and $c:=\tanh(2t)$, $s:=\sech(2t)$, there exists a unique $r(t)\in\R$ such that $\begin{pmatrix}
1/c & s/c\\
s/c & 1/c
\end{pmatrix}=\begin{pmatrix}
\cosh(r(t)) & \sinh(r(t))\\
\sinh(r(t)) & \cosh(r(t))
\end{pmatrix}.$ It would thus be interesting to find an analog of the theory of~\cite{LPLie} for the orthogonal group rather than the symplectic group.
\end{remark}

\begin{remark}
A more standard choice of coordinates for $\OG(n,2n)$ is to consider the set $\OG'(n,2n)$ of all $X\in \Gr(n,2n)$ where another symmetric bilinear form, $\eta'(u,v):=u_1v_{n+1}+u_2v_{n+2}+\dots+u_nv_{2n}$, vanishes. Consider a $2n\times 2n$ matrix $J$ with the only nonzero entries given by $J_{2j-1,j}=J_{2j,j}=J_{2j-1,j+n}=1/2,$ $J_{2j,j+n}=-1/2,$ for all $j\in[n]$. Then the map $X\mapsto X\cdot J$ gives a bijection between $\OG(n,2n)\sqcup \OG_-(n,2n)$ and $\OG'(n,2n)$. Moreover, for $M\in\Matsym$, the matrix $M\cdot J$ has the form $[I_n|M']$ for a skew-symmetric matrix $M'$ given by $m'_{i,j}=(-1)^{i+j+1}m_{i,j}$ for $i\neq j$ and $m'_{i,j}=0$ for $i=j$. A standard way to work with  $\OG'(n,2n)$ is to consider \emph{spinor coordinates}, which are essentially \emph{Pfaffians} of the matrix $M'$ above, see e.g.~\cite[Section~5]{HScube}. It was shown in~\cite{GBK} that these Pfaffians are multi-point boundary correlation functions for the Ising model, as we explain in Proposition~\ref{prop:Pf}. We thank David Speyer for this remark.
\end{remark}

Proposition~\ref{prop:eta} allows one to deduce that the image of the map $\doublemap$ is contained inside the orthogonal Grassmannian.
\begin{corollary}\label{cor:doublemap_subset_OG}
We have $\doublemap(\Matsym)\subset \OG(n,2n)$.
\end{corollary}
\begin{proof}
  Let $M\in\Matsym$. It is obvious from the definition of $\double M$ that if $u,v\in\R^{2n}$ are any two rows of $\double M$ then we have $\eta(u,v)=0$, and thus by Proposition~\ref{prop:eta} we get $\doublemap(\Matsym)\subset\OG(n,2n)\sqcup \OG_-(n,2n)$. But note that $\OG(n,2n)$ and $\OG_-(n,2n)$ are not connected to each other inside $\OG(n,2n)\sqcup \OG_-(n,2n)$, however, $\Matsym\cong\R^{n\choose2}$ is connected. Thus $\doublemap(\Matsym)$ is connected, and clearly the image of the identity matrix $I_n\in\Matsym$ belongs to $\OG(n,2n)$ and not to $\OG_-(n,2n)$. The result follows.
\end{proof}

\begin{proposition}\label{prop:involution}
Let $X\in \OGtnn(n,2n)$, and let $\Mcal:=\Mcal_X$ be the positroid of $X$ with decorated permutation $\pidec_\Mcal$. Then $\pidec_\Mcal$ is a \emph{fixed-point free involution}: if $\pidec_\Mcal(i)=j$ then $i\neq j$ and $\pidec_\Mcal(j)=i$.
\end{proposition}
\begin{proof}
  It is clear from Definition~\ref{dfn:Imin_Imax_Mcal} that we have $\Imin_i(\Mcal)=[2n]\setminus \Imax_i(\Mcal)$, because $X\in\OG(n,2n)$. Suppose now that $\pidec_\Mcal(i)=j$ and that $i\neq j$. By Remark~\ref{rmk:positroid_pidec}, $\piinv_\Mcal(i)$ is the unique element of the set
  \[\Imax_i(\Mcal)\setminus \Imax_{i+1}(\Mcal)=([2n]\setminus \Imin_i(\Mcal))\setminus ([2n]\setminus \Imin_{i+1}(\Mcal))=\Imin_{i+1}(\Mcal)\setminus \Imin_i(\Mcal),\]
  which is equal to $\{\pidec_\Mcal(i)\}=\{j\}$. 
  Thus $\piinv_\Mcal(i)=j$, equivalently, $\pidec_\Mcal(j)=i$, so $\pidec_\Mcal$ is an involution. It remains to show that it is \emph{fixed-point free}, i.e., that $\pidec_\Mcal(i)\neq i$ for all $i\in[2n]$. We can have $\pidec_\Mcal(i)=i$ if either $i$ is a \emph{loop} or a \emph{coloop} of $\Mcal$, that is, if either $i\in I$ for all $I\in\Mcal$ or $i\notin I$ for all $I\in \Mcal$, respectively. Choose some $I\in\Mcal$. Then $[2n]\setminus I$ also belongs to $\Mcal$, which shows that $i$ is neither a loop nor a coloop of $\Mcal$. We are done with the proof.
\end{proof}

\begin{remark}\label{rmk:Imin_Imax}
Recall that given a matching $\medpa$ on $[2n]$, Definition~\ref{dfn:prec_i} gives two disjoint sets $\Imin_i(\medpa)$ and $\Imax_i(\medpa)$ for each $i\in[2n]$. It is easy to check that if $\pidec$ is the fixed-point free involution corresponding to $\medpa$ then $\Imin_i(\Mcal_\pidec)=\Imin_i(\medpa)$ and $\Imax_i(\Mcal_\pidec)=\Imax_i(\medpa)$.
\end{remark}

In Section~\ref{sec:inverse_problem_intro}, we described how to transform a planar graph $G$ embedded in a disk into a \emph{medial graph} $\Gmed$, and then how to obtain a medial pairing $\medpa_G$ from $\Gmed$. Not all matchings can be obtained in this way, for example, when $n=2$, the matching $\{\{1,4\}, \{2,3\}\}$ is not a medial pairing of any graph $G$. It will thus be more convenient for us to work with medial graphs rather than matchings. In Section~\ref{sec:ising_to_OG}, we introduce \emph{generalized planar Ising networks} which correspond to \emph{all} matchings on $[2n]$.

\def\int{\operatorname{int}}
\def\Emed{{E^\times}}
\def\Vmed{{V^\times}}
\def\Vint{{V^\times_{\int}}}
\def\Nmed{{N^\times}}
\def\Jmed{J^\times}

\begin{definition}
A \emph{medial graph} is a planar graph $\Gmed=(\Vmed,\Emed)$ embedded in a disk, such that it has $2n$ boundary vertices $\bbip_1,\bbip_2,\dots,\bbip_{2n}\in\Vmed$ in counterclockwise order, each of degree $1$, and such that every other vertex of $\Gmed$ has degree $4$.
\end{definition}

The non-boundary vertices (the ones that have degree $4$) are called \emph{interior vertices} of $\Gmed$, and we let $\Vint:=\Vmed\setminus\{\bbip_1,\dots,\bbip_{2n}\}$ denote the set of such vertices. Each medial graph $\Gmed$ gives rise to a \emph{medial pairing} $\medpa_\Gmed$, as in Section~\ref{sec:inverse_problem_intro}. We say that a medial graph $\Gmed$ is \emph{reduced} if the number of its interior vertices equals $\xing(\medpa_\Gmed)$. Equivalently, $\Gmed$ is reduced if every edge of $\Gmed$ belongs to some medial strand connecting two boundary vertices, no medial strand intersects itself, and no two medial strands intersect more than once.

\begin{lemma}\label{lemma:Gmed_exists}
For every matching $\medpa$ on $[2n]$, there exists a reduced medial graph $\Gmed$ satisfying $\medpa_\Gmed=\medpa$.
\end{lemma}
\begin{proof}
For each pair $\{i,j\}\in\medpa$, connect $\bbip_i$ with $\bbip_j$ by a straight line segment. Then perturb each line segment slightly so that every point inside the disk would belong to at most two segments, obtaining a \emph{pseudoline arrangement}. Let $\Gmed$ be obtained from this pseudoline arrangement by putting an interior vertex at each intersection point. It is clear that $\Gmed$ is a reduced medial graph whose medial pairing is $\medpa$. Alternatively, $\Gmed$ can be constructed by induction on $\xing(\medpa)$ in an obvious way using the poset $\P_n$ from Definition~\ref{dfn:P_n}.
\end{proof}

Let us say that a \emph{medial network} $\Nmed=(\Gmed,\Jmed)$ is a medial graph $\Gmed$ together with a function $\Jmed:\Vint\to \R_{>0}$. Thus if $N=(G,J)$ is a planar Ising network then the edges of $G$ correspond to the interior vertices of the corresponding medial graph $\Gmed$ and thus the Ising network $N$ gives rise to a medial network $\Nmed=(\Gmed,\Jmed)$, as described in Sections~\ref{sec:inverse_problem_intro} and~\ref{sec:ising_to_OG}. In the remainder of this section, we will work with medial networks rather than with planar Ising networks.

Every medial graph gives rise to a plabic graph. In order to describe this correspondence, we first introduce a canonical way to orient each medial graph, as described in~\cite{HWX}.

\begin{proposition}\label{prop:orient}
  Let $\Gmed$ be a medial graph. Then there exists a unique orientation of the edges of $\Gmed$ such that:
  \begin{enumerate}
  \item for $i\in[2n]$, $\bbip_i$ is a source if and only if $i$ is odd;
  \item each interior vertex $v\in\Vint$ of $\Gmed$ is incident to two incoming and two outgoing arrows so that their directions alternate around $v$.
  \end{enumerate}
\end{proposition}
\begin{proof}
If $\Gmed$ is connected then it is easy to see that there are just two orientations satisfying the second condition, since we can color the faces of $\Gmed$ in a bipartite way and then orient all black faces clockwise and all white faces counterclockwise, or vice versa. One of these two orientations will satisfy the first condition. If $\Gmed$ has $C$ connected components then there are $2^C$ orientations of $G$ satisfying the second condition, but there will still be one of them that satisfies the first condition, because the number of vertices of each connected component is even.
\end{proof}

\newcommand\ewh[1]{#1^\circ}
\newcommand\ebl[1]{#1^\bullet}

\begin{definition}\label{dfn:medial_to_plabic}
Given a medial network $\Nmed=(\Gmed,\Jmed)$, the associated weighted plabic graph $(\Gbip,\wt)$ is constructed as follows. First, orient the edges of $\Gmed$ as in Proposition~\ref{prop:orient}, and then for each oriented edge $e$ of $\Gmed$, put a white vertex $\ewh e$ of $\Gbip$ close to the source of $e$ and a black vertex $\ebl e$ of $\Gbip$ close to the target of $e$. If the source (resp., the target) of $e$ is a boundary vertex $\bbip_i$ then we set $\ewh e:=\bbip_i$ (resp., $\ebl e:=\bbip_i$). Now, for each edge $e\in\Ebip$ of $\Gbip$, connect $\ebl e$ and $\ewh e$ by an edge of $\Gbip$, and set its weight to $1$: $\wt(\{\ebl e,\ewh e\}):=1$. Additionally for every interior vertex $v\in\Vint$ of $\Gbip$ incident to edges $e_1,e_2,e_3,e_4\in\Ebip$ in counterclockwise order so that $v$ is the target of $e_1$ and $e_3$ and the source of $e_2$ and $e_4$, add four edges $\{\ebl{e_1},\ewh{e_2}\}$, $\{\ewh{e_2},\ebl{e_3}\}$, $\{\ebl{e_3},\ewh{e_4}\}$, $\{\ewh{e_4},\ebl{e_1}\}$ to $\Gbip$. The weights of these edges are given by
\begin{equation}\label{eq:wt_ebl_ewh}
\wt(\{\ebl{e_1},\ewh{e_2}\})=\wt(\{\ebl{e_3},\ewh{e_4}\}):=s_v,\quad \wt(\{\ewh{e_2},\ebl{e_3}\})=\wt(\{\ewh{e_4},\ebl{e_1}\}):=c_v,
\end{equation}
where $s_v$ and $c_v$ are given by~\eqref{eq:sinh_tanh}, that is,
\[
  s_v:=\sech(2\Jmed_v)=\frac{2}{\exp(2\Jmed_v)+\exp(-2\Jmed_v)};\quad c_v:=\tanh(2\Jmed_v)=\frac{\exp(2\Jmed_v)-\exp(-2\Jmed_v)}{\exp(2\Jmed_v)+\exp(-2\Jmed_v)}.\]
This defines a weighted plabic graph $(\Gbip,\wt)$ associated to the medial network $\Nmed=(\Gmed,\Jmed)$.
\end{definition}

Recall that for a medial graph $\Gmed$, the corresponding medial pairing $\medpa_\Gmed=\{\{i_1,j_1\},\dots,\{i_n,j_n\}\}$ is a matching on $[2n]$. We define a permutation $\pidec_\Gmed:[2n]\to[2n]$ by setting $\pidec_\Gmed(i_k):=j_k$ and $\pidec_\Gmed(j_k):=i_k$ for all $k\in[2n]$. Thus $\pidec_\Gmed$ is a fixed-point free involution.

\begin{lemma}\label{lemma:reduced_reduced}
A medial graph $\Gmed$ is reduced if and only if the corresponding plabic graph $\Gbip$ is reduced. We have $\pidec_\Gmed=\pidec_\Gbip$.
\end{lemma}
\begin{proof}
This is straightforward to check from the definitions, since the medial strands correspond to the strands in $\Gbip$ from Definition~\ref{dfn:pidec}.
\end{proof}

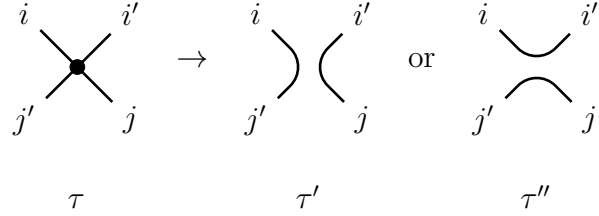
\begin{figure}
  \def\lw{1pt}
  \def\rc{10}
  \begin{tabular}{ccccc}
  \begin{tikzpicture}[baseline=(zero.base)]
    \node[draw,circle,fill=black,scale=0.5] (zero) at (0,0) {};
  \node  (A) at (45:1) {$i'$};
  \node  (B) at (135:1) {$i$};
  \node  (C) at (-135:1) {$j'$};
  \node  (D) at (-45:1) {$j$};
  \draw[line width=\lw] (A)--(C);
  \draw[line width=\lw] (B)--(D);
\end{tikzpicture}
    & $\to$ & 
  \begin{tikzpicture}[baseline=(zero.base)]
    \coordinate (zero) at (0,0);
  \node  (A) at (45:1) {$i'$};
  \node  (B) at (135:1) {$i$};
  \node  (C) at (-135:1) {$j'$};
  \node  (D) at (-45:1) {$j$};
  \draw[line width=\lw,rounded corners=\rc] (A) -- (zero) -- (D);
  \draw[line width=\lw,rounded corners=\rc] (C) -- (zero) -- (B);
\end{tikzpicture}
    & or &
  \begin{tikzpicture}[baseline=(zero.base)]
    \coordinate (zero) at (0,0);
  \node  (A) at (45:1) {$i'$};
  \node  (B) at (135:1) {$i$};
  \node  (C) at (-135:1) {$j'$};
  \node  (D) at (-45:1) {$j$};
  \draw[line width=\lw,rounded corners=\rc] (A) -- (zero) -- (B);
  \draw[line width=\lw,rounded corners=\rc] (C) -- (zero) -- (D);
\end{tikzpicture}\\
    &&&&\\
    $\medpa$ &  & $\medpa'$ & & $\medpa''$\\
  \end{tabular}
  \caption{\label{fig:uncrossing}Uncrossing the pairs $\{i,j\}$ and $\{i',j'\}$.}
\end{figure}
\def\proj{{S^\times}}
\begin{proposition}\label{prop:matchings_complements}
Given a medial network $\Nmed=(\Gmed,\Jmed)$, let $(\Gbip,\wt)$ be the corresponding weighted plabic graph. Then $\Meast(\Gbip,\wt)$ yields an element of $\pc_\Mcal\cap \OGtnn(n,2n)$, where $\Mcal=\Mcal_{\pidec_\Gmed}$ is the positroid corresponding to the fixed-point free involution $\pidec_\Gmed$.
\end{proposition}
\begin{proof}

By Lemma~\ref{lemma:reduced_reduced}, we have $\pidec_\Gmed=\pidec_\Gbip$. Let $X\in\Grtnn(n,2n)$ be the point given by $\Meast(\Gbip,\wt)$. By Theorem~\ref{thm:Meas}, $X$ belongs to $\pc_\Mcal$, and it remains to show that $X$ belongs to $\OGtnn(n,2n)$. Given an almost perfect matching $\match$ of $\Gbip$, let us define $\proj(\match)\subset \Emed$ to be the set of edges $e$ of $\Gmed$ such that the edge $\{\ebl e,\ewh e\}$ of $\Gbip$ belongs to $\match$. We claim that for all sets $R\subset\Emed$ of edges of $\Gmed$, we have
  \begin{equation}\label{eq:sum_over_matchings_Emed}
\sum_{\proj(\match)=R}\wt(\match)= \sum_{\proj(\match)=\Emed\setminus R}\wt(\match),
  \end{equation}
  where the sums are over almost perfect matchings of $\Gbip$. The left hand side of~\eqref{eq:sum_over_matchings_Emed} is equal to the product over all interior vertices $v\in\Vint$ of $\Gmed$ of $\insideprod(v)$, where $\insideprod(v)$ is equal to either $c_e$, $s_e$, $c_e^2+s_e^2$, $1$, or $0$, depending on which of the four edges of $\Gmed$ adjacent to $v$ belong to $R$. It is clear from Figure~\ref{fig:Gbip_matchings} that replacing $R$ with its complement does not affect this product. (The only non-trivial change is replacing $c_e^2+s_e^2$ with $1$, but recall that we have $c_e^2+s_e^2=1$ by construction.) This proves~\eqref{eq:sum_over_matchings_Emed}, and clearly if two almost perfect matchings $\match, \match'$ of $\Gbip$ satisfy $\proj(\match')=\Emed\setminus \proj(\match)$ then they also satisfy $\partial(\match')=[2n]\setminus\partial(\match)$, finishing the proof.
\end{proof}

For a matching $\medpa$ on $[2n]$, let $\Mcal_\medpa$ be the positroid corresponding to the fixed-point free involution $\pidec:[2n]\to [2n]$ associated with $\medpa$, and denote by $\pc_\medpa:=\pc_{\Mcal_\medpa}$. Following~\cite{Lam}, denote by $\P_n$ the partially ordered set (\emph{poset}) of all matchings $\medpa$ on $[2n]$. (It is easy to see that $\P_n$ has $\frac{(2n)!}{n!2^n}$ elements.) The covering relations of $\P_n$ are described as follows. Given a matching $\medpa$ on $[2n]$, suppose that the pairs $\{i,j\},\{i',j'\}\in\medpa$ form a crossing, as in Definition~\ref{dfn:xing}. Introduce two matchings
\begin{equation}\label{eq:two_uncrossings}
\begin{split}
\medpa'&:=\medpa\setminus\{\{i,j\},\{i',j'\}\}\cup\{\{i,j'\},\{i',j\}\};\\
\medpa''&:=\medpa\setminus\{\{i,j\},\{i',j'\}\}\cup\{\{i,i'\},\{j,j'\}\}.
\end{split}
\end{equation}

\begin{definition}\label{dfn:P_n}
  We say that $\medpa'$ and $\medpa''$ are obtained from $\medpa$ by \emph{uncrossing the pairs $\{i,j\}$ and $\{i',j'\}$} (see Figure~\ref{fig:uncrossing}). In addition, if $\xing(\medpa')+1=\xing(\medpa)$ (resp., $\xing(\medpa'')+1=\xing(\medpa)$), we write $\medpa'\lessdot \medpa$ (resp., $\medpa''\lessdot\medpa$), and let $\P_n$ be the poset whose order relation $\leq$ is the transitive closure of $\lessdot$.
\end{definition}
\begin{remark}\label{rmk:uncrossing}
  Equivalently, as explained in~\cite[Section~4.5]{Lam}, given a medial graph $\Gmed$ with medial pairing $\medpa$, we have $\medpa'\lessdot \medpa$ if and only if  ``uncrossing'' the unique vertex $v\in\Vint$ of $\Gmed$ that belongs to the intersection of medial strands connecting $\bbip_i$ to $\bbip_j$ and $\bbip_{i'}$ to $\bbip_{j'}$ yields a \emph{reduced} medial graph with medial pairing $\medpa'$. Here \emph{uncrossing} an interior vertex of a medial graph means replacing its neighborhood in one of the two ways shown in Figure~\ref{fig:uncrossing}.
\end{remark}

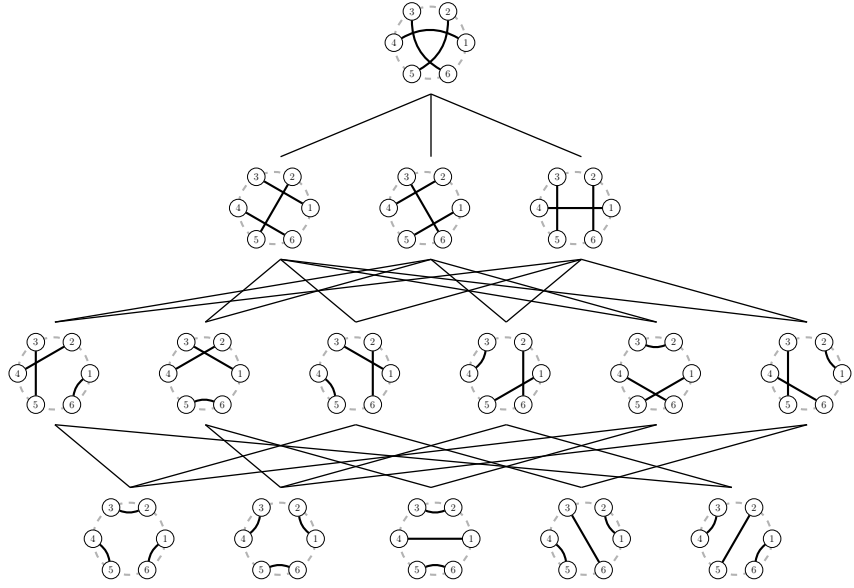
\begin{figure}

\def\nodescl{0.4}
\def\sclbx{0.8}
\def\tikzscl{0.6}
\newcommand\matchinggg[9]{
\scalebox{\sclbx}{
\begin{tikzpicture}[scale=\tikzscl]
\draw[line width=1.0,opacity=0.3,dashed] (0,0) circle (1);
\node[draw,circle,fill=white,scale=\nodescl] (1) at (0:1) {$1$};
\node[draw,circle,fill=white,scale=\nodescl] (2) at (60:1) {$2$};
\node[draw,circle,fill=white,scale=\nodescl] (3) at (120:1) {$3$};
\node[draw,circle,fill=white,scale=\nodescl] (4) at (180:1) {$4$};
\node[draw,circle,fill=white,scale=\nodescl] (5) at (-120:1) {$5$};
\node[draw,circle,fill=white,scale=\nodescl] (6) at (-60:1) {$6$};
\draw[line width=1.0] (#1) to[bend right=#7] (#2);
\draw[line width=1.0] (#3) to[bend right=#8] (#4);
\draw[line width=1.0] (#5) to[bend right=#9] (#6);
\end{tikzpicture}}
}

\newcommand\matchingg[7]{
\matchinggg{#1}{#2}{#3}{#4}{#5}{#6}{#7}{#7}{#7}
}
\newcommand\matchingb[6]{
\matchinggg{#1}{#2}{#3}{#4}{#5}{#6}{\bnd}00
}
\newcommand\matching[6]{
\matchingg{#1}{#2}{#3}{#4}{#5}{#6}{0}
}
\renewcommand\edge[4]{
\draw[line width=0.5] (#1#3.north)--(#2#4.south);
}

\def\levl{2.2}
\def\levll{4.4}
\def\levlll{6.6}
\def\bnd{20}
\begin{tikzpicture}
\node (1A) at (-4,0) {\matchingg 165432{\bnd}};
\node (1B) at (-2,0) {\matchingg 123456{-\bnd}};
\node (1C) at (0,0) {\matchinggg 143265{0}{\bnd}{\bnd}};
\node (1D) at (2,0) {\matchinggg 213654{\bnd}0{\bnd}};
\node (1E) at (4,0) {\matchinggg 162543{\bnd}0{\bnd}};
\node (2A) at (-5,\levl) {\matchingb 162435};
\node (2B) at (-3,\levl) {\matchingb 651324};
\node (2C) at (-1,\levl) {\matchingb 541326};
\node (2D) at (1,\levl) {\matchingb 431526};
\node (2E) at (3,\levl) {\matchingb 321546};
\node (2F) at (5,\levl) {\matchingb 213546};
\node (3A) at (-2,\levll) {\matching 132546};
\node (3B) at (0,\levll) {\matching 152436};
\node (3C) at (2,\levll) {\matching 142635};
\node (4A) at (0,\levlll) {\matchingg 145236{30}};
\edge 12AA
\edge 12AC
\edge 12AE
\edge 12BB
\edge 12BD
\edge 12BF
\edge 12CB
\edge 12CE
\edge 12DC
\edge 12DF
\edge 12EA
\edge 12ED
\edge 23AB
\edge 23AC
\edge 23BA
\edge 23BB
\edge 23CA
\edge 23CC
\edge 23DB
\edge 23DC
\edge 23EA
\edge 23EB
\edge 23FA
\edge 23FC
\edge 34AA
\edge 34BA
\edge 34CA
\end{tikzpicture}

  \caption{\label{fig:P_3}The Hasse diagram of the poset $\P_3$.}
\end{figure}

By~\cite[Lemma~4.13]{Lam}, the poset $\P_n$ is graded with grading given by $\xing(\medpa)$, and by~\cite{HK,LamEuler}, $\P_n$ is a shellable Eulerian poset. See Figure~\ref{fig:P_3} for the case $n=3$.

We are now ready to state the main result of this section.

\begin{theorem}\label{thm:Gmed_parametrization}\leavevmode
\begin{enumerate}[(i)]
\item\label{item:Gmed_param:homeo}     Given a reduced medial graph $\Gmed$, let $\Gbip$ be the corresponding plabic graph with positroid $\Mcal:=\Mcal_\Gbip$. Then the map $\Jmed\mapsto \Meast(\Gbip,\wt)$ is a homeomorphism between $\R_{>0}^{\Emed}$ and $\pc_{\Mcal}\cap \OGtnn(n,2n)$.
\item\label{item:Gmed_param:disjoint} The set $\OGtnn(n,2n)$ is a disjoint union of cells
  \begin{equation}\label{eq:disjoint_OG}
\OGtnn(n,2n)=\bigsqcup_{\medpa\in\P_n}\left(\pc_\medpa\cap \OGtnn(n,2n)\right),
  \end{equation}
  and each cell $\pc_\medpa\cap \OGtnn(n,2n)$ is homeomorphic to  $\R^{\xing(\medpa)}$.
\item\label{item:Gmed_param:closures} For $\medpa\in\P_n$, the closure of the cell $\pc_\medpa\cap \OGtnn(n,2n)$ in $\Gr(n,2n)$ equals
  \begin{equation}\label{eq:closures}
\overline{\left(\pc_\medpa\cap \OGtnn(n,2n)\right)}= \bigsqcup_{\sigma\in\P_n:\sigma\leq\medpa}\left(\pc_\sigma\cap \OGtnn(n,2n)\right).
  \end{equation}
\end{enumerate}
\end{theorem}
\begin{proof}
 As we have shown in Proposition~\ref{prop:involution}, for every $X\in\OGtnn(n,2n)$, the decorated permutation $\pidec=\pidec_{\Mcal_X}$ associated with the positroid $\Mcal_X$ of $X$ is a fixed-point free involution, and thus~\eqref{eq:disjoint_OG} follows from~\eqref{eq:disjoint_Gr}. The remainder of part~\eqref{item:Gmed_param:disjoint} (that each cell is an open ball) follows from part~\eqref{item:Gmed_param:homeo}, which we prove now. Thus we fix a reduced medial graph $\Gmed$ and the corresponding plabic graph $\Gbip$, reduced by Lemma~\ref{lemma:reduced_reduced}. Let $\psi:\R_{>0}^{\Emed}\to \pc_{\Mcal}\cap \OGtnn(n,2n)$ be the map that sends $\Jmed\mapsto \Meast(\Gbip,\wt)$. We first show that $\psi$ is injective. By Theorem~\ref{thm:Meast}, it suffices to show that the map $\Jmed:\Vint\to\R_{>0}$ can be reconstructed from the corresponding weight function $\wt\in\R_{>0}^\Ebip/\Gauge$. Fix an interior vertex $v\in\Vint$ of $\Gmed$ and consider the corresponding four vertices $v_1,v_2,v_3,v_4\in\Vbip$ of $\Gbip$ on the four edges of $\Gmed$ incident to $v$. Let $e_{12}$, $e_{23}$, $e_{34}$, $e_{14}$ be the four edges of $\Gbip$ forming a square around $v$. We have $\wt(e_{12})=\wt(e_{34})=s_v$ and $\wt(e_{23})=\wt(e_{14})=c_v$, as in~\eqref{eq:wt_ebl_ewh}. Suppose that we have applied a gauge transformation to $\wt$ obtaining another weight function $\wt'$. Thus we have rescaled all edges adjacent to the vertex $v_k$ by some number $t_k\in\R_{>0}$ for $1\leq k\leq 4$. Therefore
  \[\wt'(e_{12})=t_1t_2s_v,\quad \wt'(e_{34})=t_3t_4s_v,\quad \wt'(e_{23})=t_2t_3c_v,\quad \wt'(e_{14})=t_1t_4c_v.\]
  \def\Jmedprime{(\Jmed)'}
  In order for $\wt'$ to come from some other map $\Jmedprime:\Vint\to\R_{>0}$, we must have
  \[t_1t_2s_v=t_3t_4s_e=s_v',\quad t_2t_3c_v=t_1t_4c_e=c_v', \quad (s_v')^2+(c_v')^2=1,\]
  where $s_v'=\sech(2\Jmedprime_v)$ and $c_v'=\tanh(2\Jmedprime_v)$. But the above equations imply that $t_1=t_2=t_3=t_4=1$, and it follows that $\psi$ is injective.

Clearly $\psi$ is continuous, and we now prove that it is surjective, and that its inverse is also continuous. We need the following simple observation, whose proof we leave as an exercise to the reader.
\begin{lemma}\label{lemma:exists_crossing}
 Suppose that $\Gmed$ is a connected medial graph having at least one interior vertex. Then there exists an interior vertex $v\in\Vint$ and an index $i\in[2n]$ such that $v$ is connected in $\Gmed$ to both $\bbip_i$ and $\bbip_{i+1}$ (modulo $2n$).
\end{lemma}
Note that Proposition~\ref{prop:removable_spikes_edges} follows from Lemma~\ref{lemma:exists_crossing} as an immediate corollary.

We now return to the proof of Theorem~\ref{thm:Gmed_parametrization}, part~\eqref{item:Gmed_param:homeo}. Let $\medpa$ be a matching on $[2n]$, $\pidec$ be the corresponding fixed-point free involution, $\Mcal=\Mcal_\pidec$ be the corresponding positroid. Choose a reduced medial graph $\Gmed$ with medial pairing $\medpa_\Gmed=\medpa$ (which exists by Lemma~\ref{lemma:Gmed_exists}), and let $\Gbip$ be the associated plabic graph. If $\Gmed$ is not connected then each of its connected components contains an even number of vertices. Moreover, in this case $\Gbip$ induces the same partition of boundary vertices into connected components, and it is clear from the definition of the map $\Meast$ and Theorem~\ref{thm:Meas} that each minor of $\Meast(\Gbip,\wt)$ is a product of the individual minors for each of the connected components. Thus the problem naturally separates into several independent problems, one for each connected component of $\Gmed$, and in what follows, we assume that $\Gmed$ is connected.

Let $X\in\pc_\Mcal\cap\OGtnn(n,2n)$. By Theorem~\ref{thm:Meast}, there exists a weight function $\wt:\Ebip\to\R_{>0}$ such that $\Meast(\Gbip,\wt)=X$, and our goal is to show that there exists a unique function $\Jmed:\Vint\to\R_{>0}$ such that $\psi(\Jmed):\Ebip\to\R_{>0}$ is obtained from $\wt$ using gauge transformations. Choose $v\in\Vint$ and $i\in[2n]$ as in Lemma~\ref{lemma:exists_crossing}. Thus $\Gbip$ contains a removable bridge between $i$ and $i+1$ (modulo $2n$). Denote by $v_1$ and $v_2$ the vertices of $\Gbip$ adjacent to $\bbip_i$ and $\bbip_{i+1}$ respectively, and denote by $v_3$ and $v_4$ the other two vertices of $\Gbip$ so that $v_1,v_2,v_3,v_4$ surround $v$ in counterclockwise order. Applying gauge transformations to $v_1$ and $v_2$, we may assume that $\wt(\{v_1,\bbip_i\})=\wt(\{v_2,\bbip_{i+1}\})=1$. Let $s:=\wt(\{v_1,v_2\})>0$. Applying gauge transformations to $v_3$ and $v_4$, we may assume that $\wt(\{v_3,v_4\})=s$, and $\wt(\{v_1,v_4\})=\wt(\{v_2,v_3\})=c$ for some $c\in\R_{>0}$. Now, let $I:=\Imin_{i+1}(\Mcal)$ and $J:=[2n]\setminus I=\Imax_{i+1}(\Mcal)$. Thus $i+1\in I$ and $i\in J$. Choose some $n\times 2n$ matrix  representing $X$ and denote by $X_k\in\R^n$ its $k$-th column.  For $u\in\R^n$ and $k\in[2n]$, let $X(k\to u)$ denote the matrix obtained from $X$ by replacing its $k$-th column with $u$. We introduce linear functions $h_I,h_J:\R^n\to \R$ as follows:
\[h_I(u):=\Delta_I(X({i+1}\to u)),\quad h_J(u):=\Delta_J(X(i\to u)).\]
Denote $u:=X_i$ and $w:=X_{i+1}$. Since $X\in\OG(n,2n)$, we get $h_I(w)=h_J(u)$ and $h_I(u)=h_J(w)$. Let $(\Gbip',\wt')$ be obtained from $(\Gbip,\wt)$ by removing the bridge $\{v_1,v_2\}$, and let $X'\in\Grtnn(n,2n):=\Meast(\Gbip',\wt')$. By the first part of Theorem~\ref{thm:bridge_removal}, we have $s=h_I(w)/h_I(u)$. By Lemma~\ref{lemma:bridge_removal}, we have $(X')_i=u$ while  $(X')_{i+1}=w-su$. Now, after removing degree $2$ vertices\footnote{It is well known that adding/removing vertices of degree $2$ does not affect the result of $\Meast$, see~\cite[Section~4.5~(M2)]{LamCDM}.} $v_1$ and $v_2$ from $\Gbip'$ and denoting the resulting graph $\Gbip''$, each of the vertices $\bbip_i$ and $\bbip_{i+1}$ changes color and becomes adjacent to an edge of weight $\wt'(\{\bbip_i,v_4\})=\wt'(\{\bbip_{i+1},v_3\})=c$. Let us define $\wt''$ to be the same as $\wt'$ except that $\wt''(\{\bbip_i,v_4\})=\wt''(\{\bbip_{i+1},v_3\}):=1$, and let $X'':=\Meast(\Gbip'',\wt'')$. It is clear from the definition of $\Meast$ that $(X'')_i=cu$ and $(X'')_{i+1}=\frac1c(w-su)$. Finally, $\Gbip''$ has a removable bridge between $i$ and $i+1$ so by the second part of Theorem~\ref{thm:bridge_removal}, the weight $\wt''(\{v_3,v_4\})$ of this bridge must be equal to
\[\wt''(\{v_3,v_4\})=\frac{\Delta_{J}(X'')}{\Delta_{J\cup\{i+1\}\setminus\{i\}}(X'')}=\frac{h_J(cu)}{h_J \left(\frac1c(w-su)\right)},\]
  which after substituting $s:=h_I(w)/h_I(u)$, $h_J(w):=h_I(u)$, $h_J(u):=h_I(w)$,  and using the linearity of $h_I$, transforms into
  \[\wt''(\{v_3,v_4\})=\frac{c^2 h_I(u)h_I(w)}{h_I(u)^2-h_I(w)^2}.\]
By construction, $\wt''(\{v_3,v_4\})$ is equal to $s$, and thus after substituting $h_I(w):=sh_I(u)$ the above equation becomes
  \[s=\frac{c^2 s h_I(u)^2}{(1-s^2)h_I(u)^2}=\frac{c^2 s}{1-s^2},\]
  which is equivalent to $c^2+s^2=1$. Since we have $s,c>0$, it follows that $0<s,c<1$ and thus there exists a unique $t\in\R_{>0}$ satisfying $s=\sech(2t)$ and $c=\tanh(2t)$. Moreover, it is clear that $t$ depends continuously on the minors of $X$, since the denominators $\Delta_I(X)=\Delta_J(X)$ must be positive. Setting $\Jmed_v:=t$, we uncross (in the sense of Remark~\ref{rmk:uncrossing}) the interior vertex $v$ in $\Gmed$ so that the corresponding graph $\Gbip$ would be obtained by removing the bridges $\{v_1,v_2\}$ and $\{v_3,v_4\}$,  and proceed by induction, finishing the proof of  part~\eqref{item:Gmed_param:homeo} (and therefore of part~\eqref{item:Gmed_param:disjoint} as well).

  It remains to prove~\eqref{eq:closures}. There is a certain partial order (called the \emph{affine Bruhat order}) on the set of decorated permutations such that two decorated permutations $\pi,\sigma$ satisfy $\sigma\leq \pi$ if and only if the closure of the positroid cell $\pc_{\Mcal_{\pi}}$ contains $\pc_{\Mcal_{\sigma}}$. We have
  \[\Pi_{\Mcal_\pi}=\bigsqcup_{\sigma\leq \pi} \pc_{\Mcal_{\sigma}},\]
  see e.g.~\cite[Theorem~8.1]{LamCDM}. Moreover, the restriction of the affine Bruhat order to the set of fixed-point free involutions coincides with the poset $\P_n$ from Definition~\ref{dfn:P_n}. Thus we have
  \[\overline{\left(\pc_\medpa\cap \OGtnn(n,2n)\right)}\subset \bigsqcup_{\sigma\in\P_n:\sigma\leq\medpa}\left(\pc_\sigma\cap \OGtnn(n,2n)\right),\]
  and it remains to prove that the left hand side of~\eqref{eq:closures} contains the right hand side, i.e., that for all pairs $\sigma\leq\medpa$ in $\P_n$, the cell $\pc_\sigma\cap \OGtnn(n,2n)$ is contained inside the closure of $\pc_\medpa\cap \OGtnn(n,2n)$. Clearly it is enough to consider the case $\sigma\lessdot \medpa$. Then $\sigma$ is obtained from $\medpa$ by uncrossing some pairs $\{i,j\}$ and $\{i',j'\}$. Moreover, since $\Gmed$ is reduced, it contains a unique vertex $v\in\Vint$ which belongs to the medial strands connecting $i$ to $j$ and $i'$ to $j'$, and one of the two ways of uncrossing $v$ yields a reduced medial graph with medial pairing $\sigma$, see Remark~\ref{rmk:uncrossing}. But the two ways of uncrossing $v$ correspond to sending $\Jmed_v$ to either $0$ or $\infty$, or equivalently, sending either $s_v\to 1, c_v\to 0$ or $s_v\to 0,c_v\to 1$. By part~\eqref{item:Gmed_param:homeo}, we indeed see that $\pc_\sigma\cap \OGtnn(n,2n)$ is a subset of the closure of $\pc_\medpa\cap \OGtnn(n,2n)$, finishing the proof of Theorem~\ref{thm:Gmed_parametrization}.
\end{proof}

\def\fin{{\operatorname{fin}}}
\def\Efin{E_{\fin}}
\def\Gfin{G_{\fin}}
\def\Jfin{J_{\fin}}
\section{From the Ising model to the orthogonal Grassmannian}\label{sec:ising_to_OG}\label{sec:generalized_planar_Ising}
In this section, we study the relationship between the space $\Closure_n$ and the space $\OGtnn(n,2n)$. 

We start by slightly extending the notion of a planar Ising network so that contracting an edge in such a network would yield another such network. Throughout, we assume that a planar graph embedded in a disk has no \emph{loops} (i.e. edges connecting a vertex to itself) or interior vertices of degree $1$.
\begin{definition}\label{dfn:GPIN}
A \emph{generalized planar Ising network} is a pair $N=(G,J)$ where $G=(V,E)$ is a planar graph embedded in a disk and $J:E\to \R_{>0}\cup\{\infty\}$. We denote $E=\Efin\sqcup E_\infty$, where $E_\infty:=\{e\in E\mid J_e=\infty\}$. The \emph{Ising model} associated to $N$ is a probability measure on the space 
\[\{-1,1\}^{V/E_\infty}:=\{\sigma:V\to \{-1,1\}\mid \sigma_u=\sigma_v\text{ for all $\{u,v\}\in E_\infty$}\}.\]
The definitions of the probability $\Prob(\sigma)$ of a spin configuration $\sigma\in \{-1,1\}^{V/E_\infty}$, the partition function $\Zpart$, and a two-point boundary correlation $\<\sigma_i\sigma_j\>$ are  obtained from the corresponding definitions~\eqref{eq:dfn:Prob}, \eqref{eq:dfn:Zpart}, and~\eqref{eq:dfn:Corr} by replacing $\{-1,1\}^V$ with $\{-1,1\}^{V/E_\infty}$ and $E$ with $\Efin$. As before, we let $M(G,J)=(\<\sigma_i\sigma_j\>)\in\Matsym$ denote the boundary correlation matrix. Thus we have $m_{i,j}=1$ whenever there exists a path connecting $b_i$ to $b_j$ by edges in $E_\infty$.
\end{definition}

\def\vmed{v}
\def\Nmedfin{N^\times_\fin}
\def\Gmedfin{G^\times_\fin}
\def\Jmedfin{J^\times_\fin}
\begin{definition}\label{dfn:Ising_to_medial}
To each generalized planar Ising network $N=(G,J)$ we associate a medial network $\Nmed=(\Gmed,\Jmed)$. First suppose that $E=\Efin$, i.e., that $J$ only takes values in $\R_{>0}$. Then the medial graph $\Gmed=(\Vmed,\Emed)$ is obtained from $G$ as in Figures~\ref{fig:Gbip_small} and~\ref{fig:Gbip_big}. More precisely, the vertex set $\Vmed$ is given by
\[\Vmed=\{\bbip_1,\dots,\bbip_{2n}\}\sqcup\{\vmed_e\mid e\in E\},\]
where the $\bbip_1,\dots,\bbip_{2n}$ are boundary vertices placed counterclockwise on the boundary of the disk so that $b_i$ is between $\bbip_{2i-1}$ and $\bbip_{2i}$, while $\vmed_e$ is the midpoint of the edge $e\in E$ of $G$. The edges of $\Gmed$ are described as follows. If $e,e'\in E$ share both a vertex and a face then we connect $\vmed_e$ to $\vmed_{e'}$ in $\Gmed$. In addition, for each $i\in[n]$, we connect $\bbip_{2i-1}$ (resp., $\bbip_{2i}$) with $\vmed_e$ where $e\in E$ is the first in the clockwise (resp., counterclockwise) order edge of $G$ incident to $b_i$. If $b_i$ is isolated in $G$ then we connect $\bbip_{2i-1}$ to $\bbip_{2i}$ in $\Gmed$. Thus each vertex $\vmed_e\in\Vmed$ has degree $4$, and each boundary vertex $\bbip_i$, $i\in[2n]$, has degree $1$ in $\Gmed$.  Finally, we set $\Jmed_{\vmed_e}:=J_e$.

Suppose now that $E\neq \Efin$, and thus $J$ takes the value of $\infty$ on some edges of $G$. Let $N_\fin=(G,\Jfin)$ be obtained from $N$ by setting $(\Jfin)_e:=1$ for all $e\in E_\infty$ and $(\Jfin)_e:=J_e$ for $e\in \Efin$. Let $\Nmedfin:=(\Gmedfin,\Jmedfin)$ be the medial network associated to $N_\fin$. Then the  medial network $\Nmed=(\Gmed,\Jmed)$ associated to $N$ is obtained from $\Nmedfin$ by ``uncrossing'' (see Remark~\ref{rmk:uncrossing}) the vertices $v_e$ of $\Gmedfin$ for all $e\in E_\infty$. There are two ways to uncross the vertex $v_e$ as in Figure~\ref{fig:uncrossing}, and we choose the one where no edge of the resulting graph $\Gmed$ intersects the corresponding edge $e$ of $G$. This uniquely defines the medial graph $\Gmed$, and we set $\Jmed_{v_e}:=J_e$ for all $e\in \Efin$.
\end{definition}

The notion of a generalized planar Ising network is equivalent to the notion of a \emph{cactus network} introduced in~\cite[Section~4.1]{Lam}, where he also assigns a medial graph to it in the same way as in Definition~\ref{dfn:Ising_to_medial}.

\begin{remark}
Given a planar Ising network $N=(G,J)$, the above procedure assigns a medial network $\Nmed=(\Gmed,\Jmed)$ to it. In Section~\ref{sec:tnn_OG_2}, we assign a weighted plabic graph $(\Gbip,\wt)$ to $\Nmed$. It is trivial to check that the same weighted plabic graph $(\Gbip,\wt)$ gets assigned to $N=(G,J)$ in the construction described in Section~\ref{sec:dimer_model_intro}.
\end{remark}

To each medial graph $\Gmed$ (and thus to each generalized planar Ising network) we have associated a medial pairing $\medpa$ in Section~\ref{sec:inverse_problem_intro}. Let us denote 
\[\Space_\medpa:=\{M(G,J)\mid \text{$N=(G,J)$ is a generalized planar Ising network with medial pairing $\medpa$}\}.\]
The following stratification of $\Closure_n$ will be deduced from Theorem~\ref{thm:main} at the end of Section~\ref{sec:ball}.
\begin{proposition}\label{prop:Closure_cells}
The space $\Closure_n$ decomposes as
\begin{equation}\label{eq:Closure_cells}
  \Closure_n=\bigsqcup_{\medpa\in\P_n} \Space_\medpa,
\end{equation}
and for each $\medpa\in\P_n$, $\Space_\medpa$ is homeomorphic to $\R^{\xing(\medpa)}$, with closure relations given by the poset $\P_n$.
\end{proposition}

Given a (generalized) planar Ising network $N=(G,J)$, we have described two ways to assign an element of $X\in\OGtnn(n,2n)$ to $N$. First, one can take the boundary correlation matrix $M=M(G,J)$, and let $X:=\doublemap(M)$, as we did in Section~\ref{sec:embedding}. Second, one can construct a medial network $\Nmed=(\Gmed,\Jmed)$ as above, transform it into a weighted plabic graph $(\Gbip,\wt)$, and then put $X':=\Meast(\Gbip,\wt)$, as we did in Section~\ref{sec:dimer_model_intro}. Theorem~\ref{thm:Gmed_parametrization}, part~\eqref{item:Gmed_param:homeo} shows that the second map $J\mapsto \Meast(\Gbip,\wt)$ gives a homeomorphism between $\R_{>0}^{E}$ and $\pc_\medpa\cap \OGtnn(n,2n)$, where $\medpa$ is the medial pairing of $\Gmed$. The goal of the rest of this section is to show that the outputs $X=\doublemap(M)$ and $X'=\Meast(\Gbip,\wt)$ of these two maps coincide.

\begin{theorem}\label{thm:X=X'}
  Let $N=(G,J)$ be a (generalized) planar Ising network with boundary correlation matrix $M=M(G,J)$. Define $X:=\doublemap(M)$. Let $\Nmed=(\Gmed,\Jmed)$ and $(\Gbip,\wt)$ be the medial network and the weighted plabic graph corresponding to $N$, and put $X':=\Meast(\Gbip,\wt)$. Then $X=X'$ in $\Gr(n,2n)$.
\end{theorem}

\newcommand\Ge{G^{\emptyset}}
\newcommand\Ga[1]{G^{#1}}
\newcommand\GAB[2]{G^{#1\sqcup #2}}
\newcommand\VAB[2]{V^{#1\sqcup #2}}
\newcommand\EAB[2]{E^{#1\sqcup #2}}
\newcommand\Gab{G^{a,b}}
\newcommand\Vab{V^{a,b}}
\newcommand\Eab{E^{a,b}}
\def\w{\operatorname{w}}

We give two proofs of Theorem~\ref{thm:X=X'}, one using Dub\'edat's results~\cite{Dubedat}, and one using a formula of Lis~\cite{Lis} for boundary correlations in terms of the \emph{random alternating flow model}. Note that it suffices to prove Theorem~\ref{thm:X=X'} only for planar Ising networks, since the corresponding statement for generalized planar Ising networks is obtained by taking the limit $J_e\to\infty$ for all $e\in E_\infty$.

 Before we proceed with the proofs, we need several preliminary results.
\begin{proof}[Proof of Lemma~\ref{lemma:from_minors_to_correlations}]
  By~\eqref{eq:sum_of_minors_2^n}, it is enough to show that $m_{i,j}=2^{-n}\sum_{I\in \OddEven{\{i,j\}}}\Delta_I(\double M).$
  For $k\in[n]$, denote by $e_k\in \R^n$ the $k$-th standard basis vector, and for $k\in[2n]$, denote by $(\double M)_k$ the $k$-th column of $\double M$. Consider an $n\times n$ matrix $A$ with columns
  \[e_1,e_2,\dots,e_{i-1},(\double M)_{2i-1},(\double M)_{2i},e_{i+1},\dots,e_{j-1},e_{j+1},\dots,e_n.\]
  Since by Remark~\ref{rmk:full_rank}, $2e_k=(\double M)_{2k-1}+(\double M)_{2k}$ for all $k\in[n]$, we can expand $\det A$ in terms of minors of $\double M$:  
  \[\det A=2^{-n+2}\sum_{I\in \OddEven{\{i,j\}}: 2i-1,2i\in I} \Delta_I(\double M).\]
  On the other hand, since most of the columns of $A$ are basis vectors, we can compute its determinant directly: $\det A=2m_{i,j},$ where the sign in~\eqref{eq:double_signs} is chosen so that we would have $\det A=2m_{i,j}$ and not $\det A=-2m_{i,j}$. Similarly, we can define an $n\times n$ matrix $B$ with columns
  \[e_1,e_2,\dots,e_{i-1},e_{i+1},\dots,e_{j-1},(\double M)_{2j-1},(\double M)_{2j},e_{j+1},\dots,e_n.\]
  We have $\det B=2m_{i,j}$ as well, and
  \[\det B=2^{-n+2}\sum_{I\in \OddEven{\{i,j\}}: 2j-1,2j\in I} \Delta_I(\double M).\]
  It remains to note that for all $I\in\OddEven{\{i,j\}}$, we have either $2i-1,2i\in I$ or $2j-1,2j\in I$, but not both. Thus
  \[4m_{i,j}=\det A+\det B=2^{-n+2}\sum_{I\in \OddEven{\{i,j\}}}\Delta_I(\double M),\]
  finishing the proof.
\end{proof}

\begin{lemma}\label{lemma:Skandera}
  Let $J:=\{1,3,\dots,2n-1\}$ and $X'\in\OGtnn(n,2n)$. Then for all $I\in{[2n]\choose n}$, we have
  \[\Delta_I(X')\leq \Delta_J(X').\]
\end{lemma}
\begin{proof}
This follows from Skandera's inequalities~\cite{Skandera} for $\Grtnn(n,2n)$. Namely, by~\cite[Theorem~6.1]{FP}, we have $\Delta_I(X')\Delta_{[2n]\setminus I}(X')\leq \Delta_J(X')\Delta_{[2n]\setminus J}(X')$ for all $X'\in\Grtnn(n,2n)$. In particular, if $X'\in\OGtnn(n,2n)$, this becomes $(\Delta_I(X'))^2\leq (\Delta_J(X'))^2$, which finishes the proof.
\end{proof}
An important consequence of the above lemma is that for $J:=\{1,3,\dots,2n-1\}$ and all $X'\in\OGtnn(n,2n)$, we have $\Delta_J(X')>0$, since we must have $\Delta_I(X')>0$ for some $I\in{[2n]\choose n}$. 

  \begin{lemma}\label{lemma:doulbemap_image}
The image $\doublemap(\Matsym)$ contains $\OGtnn(n,2n)$. Equivalently, for any $X'\in\OGtnn(n,2n)$, there exists a  matrix $M'\in\Matsym$ such that $X'=\doublemap(M')$ as elements of $\Gr(n,2n)$.
\end{lemma}
\begin{proof}
  We are going to use Lemma~\ref{lemma:from_minors_to_correlations}. Choose some $n\times 2n$ matrix $A$ representing $X'$ in $\Gr(n,2n)$, and let $\double{I_n}$ be the $n\times 2n$ matrix given by $(\double{I_n})_{i,2i-1}=(\double{I_n})_{i,2i}=1$ and the remaining entries being zero. Remark~\ref{rmk:full_rank} says that for a matrix $M'\in\Matsym$, we have $\double{M'} \cdot (\double{I_n})^T=2I_n,$ where $I_n$ is the $n\times n$ identity matrix, and $ ^T$ denotes matrix transpose. Let $B:=A\cdot (\double{I_n})^T$. We claim that if $X'\in\OGtnn(n,2n)$ then $B$ is an invertible matrix. Indeed, by the multilinearity of the determinant, we have $\det B=\sum_{I\in \OddEven{\emptyset}}\Delta_I(A).$ This sum contains only nonnegative terms, and the term $\Delta_{\{1,3,\dots,2n-1\}}(A)$ is positive by Lemma~\ref{lemma:Skandera}. Thus $B\in\GL_n(\R)$ is invertible, and we can consider the matrix $C:=2\cdot B^{-1}\cdot A$, which represents the same element $X'$ in $\Gr(n,2n)$. The matrix $C$ satisfies $C\cdot (\double{I_n})^T=2 I_n$, in particular, $C_{i,2j-1}=-C_{i,2j}$ for $i\neq j\in[n]$. We define the $n\times n$ matrix $M'=(m'_{i,j})$ by $m'_{i,i}:=1$ and $m'_{i,j}:=(-1)^{i+j+\one(i<j)}C_{i,2j-1}$
  for $i\neq j\in[n]$, in agreement with~\eqref{eq:double_signs}. It turns out  that $M'$ is a symmetric matrix, since its entries can be recovered from the minors of $C$ as follows. As we have mentioned in the proof of Lemma~\ref{lemma:from_minors_to_correlations}, for each $I\in\OddEven{\{i,j\}}$, we have either $2i-1,2i\in I$ or  $2j-1,2j\in I$, but not both. Thus we can write $m'_{i,j}=2^{-n+2}\sum_{I} \Delta_I(C)$, where the sum is over all $I\in\OddEven{\{i,j\}}$ such that $2i-1,2i\in I$. Similarly, we have $m'_{j,i}=2^{-n+2}\sum_{I} \Delta_I(C)$, where the sum is over all $I\in\OddEven{\{i,j\}}$ such that $2j-1,2j\in I$. Since $\Delta_I(C)=\Delta_{[2n]\setminus I}(C)$ (because $C$ represents $X'\in\OGtnn(n,2n)$), we see that $m'_{i,j}=m'_{j,i}$, and thus $M'$ belongs to $\Matsym$. Similarly, using
\[2^{n-1}C_{i,2i-1}=\sum_{I\in\OddEven{\emptyset}:2i-1\in I} \Delta_I(C)=\sum_{I\in\OddEven{\emptyset}:2i\in I} \Delta_I(C)=2^{n-1}C_{i,2i},\] 
and $C_{i,2i-1}+C_{2i}=2$, we get $C_{i,2i-1}=C_{i,2i}=1$, and thus  $\doublemap(M')=X'$. We are done with the proof of Lemma~\ref{lemma:doulbemap_image}.
\end{proof}

In order to prove Theorem~\ref{thm:X=X'}, we need to show that $X:=\doublemap(M)$ equals to $X':=\Meast(\Gbip,\wt)$ as an element of $\Gr(n,2n)$. By Theorem~\ref{thm:Gmed_parametrization}, we know that $X'\in\OGtnn(n,2n)$.

By Lemma~\ref{lemma:doulbemap_image}, we get a matrix $M'\in\Matsym$ such that $\doublemap(M')=X'$. Since $X=\doublemap(M)$, we have $X=X'$ if and only if $M=M'$. By Lemma~\ref{lemma:from_minors_to_correlations}, the entries $m'_{i,j}$ of $M'$ can be written as ratios of sums of minors of $X'$. By Theorem~\ref{thm:Meas}, each such minor is a sum over almost perfect matchings of $\Gbip$ with prescribed boundary. Putting it all together, we get the following: for $i\neq j\in[n]$,
\begin{equation}\label{eq:m'_i,j:matchings}
m'_{i,j}=\frac{\sum_{\match:\partial(\match)\in\OddEven{\{i,j\}}}\wt(\match)}
  {\sum_{\match:\partial(\match)\in\OddEven{\emptyset}}\wt(\match)},
\end{equation}
where the sums are over almost perfect matchings in $\Gbip$. Our goal is to show that $m'_{i,j}$ equals to  $m_{i,j}:=\<\sigma_i\sigma_j\>$.

\def\PM{\operatorname{Match}}
\subsection{Dub\'edat's bosonization identity}\label{sec:Dubedat}
In this section, we show that $m'_{i,j}$ equals to  $m_{i,j}$ using the results of~\cite{Dubedat}. We thank the anonymous referee for explaining the following argument to us. 

Introduce a planar bipartite graph $\Gdub=(\Vdub,\Edub)$ which is obtained from $\Gbip$ by simply adding an extra edge of weight $1$ connecting $\bbip_{2i-1}$ to $\bbip_{2i}$ for all $i\in[n]$. We let $\PM(\Gdub)$ denote the set of \emph{perfect matchings} of $\Gdub$. Given such a perfect matching $\match\in\PM(\Gdub)$, its \emph{weight} is the product of weights of its edges. Given a subset $V'\subset \Vdub$ of vertices of $\Gdub$, denote by $\Gdub\setminus V'$ the graph obtained from $\Gdub$ by removing the vertices in $V'$.

\def\Zij{\Zpart_{i,j}}

\def\dg{\ast}
\def\dimer{{\operatorname{dimer}}}
\def\Gammad{{G^\dg}}
\def\Gdub{\widehat{G}^\square}
\def\match{\mathcal{A}}
\def\Gammaa{G}
\def\ndub{k}
\def\Zdub{\widehat{\Zpart}}
\def\Zdubij{\widehat{\Zpart}_{i,j}}
\def\wtdub{\widehat{\wt}}
\def\O{\mathcal{O}}

Recall from~\eqref{eq:dfn:Zpart} that the partition function $\Zpart$ of the Ising model is given by $\Zpart:=\sum_{\sigma\in\{-1,1\}^V} \wt(\sigma)$, where $\wt(\sigma):=\prod_{\{u,v\}\in E}\exp \left( J_{\{u,v\}} \sigma_u\sigma_v\right)$. For $i,j\in[n]$, let $\Zij:=\sum_{\sigma:\sigma_i=\sigma_j} \wt(\sigma)-\sum_{\sigma:\sigma_i\neq\sigma_j} \wt(\sigma)$, so that $\<\sigma_i\sigma_j\>=\frac{\Zij}{\Zpart}$.

\def\cnst{h}
\begin{proposition}[{\cite[Lemma~1]{Dubedat}}]
Let $N=(G,J)$ be a planar Ising network with $n$ boundary vertices. There exists a constant $\cnst$ (depending only on $N$) such that for all $i\neq j\in[n]$, we have
\begin{align}\label{eq:Dub1}
  \Zpart\cdot \cnst&=  \sum_{\match\in \PM(\Gdub)} \wt(\match);\\ 
\label{eq:Dub2}
\Zij\cdot \cnst&=\sum_{\match\in \PM(\Gdub\setminus \{\bbip_{2i-1},\bbip_{2j}\})} \wt(\match) + \sum_{\match\in \PM(\Gdub\setminus \{\bbip_{2i},\bbip_{2j-1}\})} \wt(\match).
\end{align}
\end{proposition}
\begin{proof}
 Let us restate the first part of~\cite[Lemma~1]{Dubedat}:
\begin{equation}\label{eq:Dub_full}
\left\langle \prod_{s=1}^\ndub X(v_sf_s) \right\rangle_{G} \left\langle \prod_{s=1}^\ndub X^\dg(f_sv_s) \right\rangle_\Gammad=2c \left\langle \prod_{s=1}^\ndub Y(b_sw_s) \right\rangle_\dimer.
\end{equation}
Before explaining this notation, we state two specializations of~\eqref{eq:Dub_full}, one for $\ndub=0$ (where each product is empty) and one for $\ndub=2$ with $X,X^\ast$ given in line~5 of~\cite[Table~1]{Dubedat} (thus $X(vf)=\psi(vf)$ is the \emph{fermion} of~\cite{KC} and $X^\dg(fv)=1$):
\begin{equation}\label{eq:Dub}
\<1\>_G\cdot \<1\>_\Gammad=2c\cdot \<1\>_\dimer,\quad \<\psi(b_if)\psi(b_jf)\>_G\cdot \<1\>_\Gammad=2c\cdot \<Y(\bbip_{2i-1}\bbip_{2i})Y(\bbip_{2j-1}\bbip_{2j})\>_\dimer.
\end{equation}
Let us now briefly describe the ingredients that go into~\eqref{eq:Dub}. First, denote $w_e:=\exp(-2J_e)$ and $w'_e:=\exp(-2J^\ast_{e^\ast})$, thus the statement $\sinh(2J_e)\sinh(2J^\ast_{e^\ast})=1$ of~\eqref{eq:duality} is equivalent to the statement $w_e+w'_e+w_ew'_e=1$ of~\cite[Eq.~(2.1)]{Dubedat}. Let $\lambda:=\prod_{e\in E} \exp(-J_e)$ and  $\Zdub:=\sum_{\sigma\in\{-1,1\}^V}  \wtdub(\sigma)$, where $\wtdub(\sigma):=\prod_{\{u,v\}\in E: \sigma_u\neq \sigma_v}w_e$, thus $\Zdub=\lambda\Zpart$. Similarly, let $\Zdubij:=\sum_{\sigma:\sigma_i=\sigma_j} \wtdub(\sigma)-\sum_{\sigma:\sigma_i\neq\sigma_j} \wtdub(\sigma)=\lambda\Zij$. Finally, we choose the vertices $v_1:=b_i$ and $v_2:=b_j$ (where $b_1,\dots,b_n$ are the boundary vertices of $G$) while the faces $f_1$ and $f_2$ are both taken to be the outer face $f$ of $G$, where $G$ is viewed as embedded in the plane. We refer the reader to~\cite{Dubedat} for the following statements.
\begin{itemize}
\item $\<1\>_\Gammaa=\Zdub$ and $\<\psi(b_if)\psi(b_jf)\>_G=\<\sigma_i\mu_f\sigma_j\mu_f\>_G=\Zdubij$ (since $\mu_f^2=1$);
\item $\<1\>_\Gammad$ is the partition function of the Ising model on $(G^\ast,J^\ast)$ with $+$ boundary conditions;
\item $c=\prod_{e\in E} \frac{w_e+w'_e}2=\prod_{e\in E} \frac1{1+c_e+s_e}$;
\item $\<1\>_\dimer=\sum_{\match\in \PM(\Gdub)} \wt(\match)$.
\end{itemize}
Next, in the notation of~\cite{Dubedat}, we have
\begin{align*} \<Y(\bbip_{2i-1}\bbip_{2i})Y(\bbip_{2j-1}\bbip_{2j})\>_\dimer&=\<\O_1(\bbip_{2j})\O_{-1}(\bbip_{2i-1})\>_\dimer+\<\O_1(\bbip_{2i})\O_{-1}(\bbip_{2j-1})\>_\dimer,\quad\text{where}\\
\<\O_1(\bbip_{2j})\O_{-1}(\bbip_{2i-1})\>_\dimer&=\sum_{\match\in \PM(\Gdub\setminus \{\bbip_{2i-1},\bbip_{2j}\})} \wt(\match), \quad\text{and}\\
\<\O_1(\bbip_{2i})\O_{-1}(\bbip_{2j-1})\>_\dimer&=\sum_{\match\in \PM(\Gdub\setminus \{\bbip_{2i},\bbip_{2j-1}\})} \wt(\match).
\end{align*} 
Substituting $h:=\frac{\lambda\<1\>_\Gammad}{2c}$ finishes the proof.
\end{proof}

\begin{example}
If $G$ has two vertices and one edge as in Figure~\ref{fig:Gbip_small}, then $\Gdub$ has the following form: 
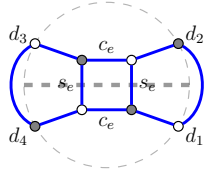
\begin{center}
\scalebox{1.1}{
  \begin{tikzpicture}[baseline=(zero.base),scale=\tikzscl]
    \drawgraph
    \edgepl{A}{AA}
    \edgepl{B}{BB}
    \edgepl{C}{CC}
    \edgepl{D}{DD}
    \edgepl{A}{B}
    \edgepl{C}{B}
    \edgepl{C}{D}
    \edgepl{A}{D}
    \draw[line width=1pt,blue] (AA) to[bend right=-60] (BB);
    \draw[line width=1pt,blue] (CC) to[bend right=-60] (DD);
  \end{tikzpicture}}
\end{center}
We find $\<1\>_G=\Zdub=2(1+w_e)$, $\<1\>_\Gammad=1$ (because of the $+$ boundary conditions), $c=\frac1{1+c_e+s_e}$, $\Zdubij=2(1-w_e)$. There are five perfect matchings of $\Gdub$ with weights $s_e^2$, $s_e$, $s_e$, $1$, and $c_e^2$, respectively. Each of the graphs $\Gdub\setminus \{\bbip_{1},\bbip_{4}\}$ and $\Gdub\setminus \{\bbip_{2},\bbip_{3}\}$ admits a single perfect matching of weight $c_e$. Thus~\eqref{eq:Dub} reads
\[2(1+w_e)\cdot 1=\frac2{1+c_e+s_e}\cdot (1+2s_e+s_e^2+c_e^2),\quad 2(1-w_e)\cdot 1=\frac2{1+c_e+s_e}\cdot 2c_e.\]
Both of these identities are easily checked using $s_e=\frac2{w_e^{-1}+w_e}$ and $c_e=\frac{w_e^{-1}-w_e}{w_e^{-1}+w_e}$. Dividing the second identity by the first one, we obtain
\[\<\sigma_i\sigma_j\>=\frac{1-w_e}{1+w_e}=\frac{c_e}{1+s_e},\]
in agreement with~\eqref{eq:sigma_12_s_c}.
\end{example}

\begin{proof}[First proof of Theorem~\ref{thm:X=X'}]
 Recall that our goal is to show that $m'_{i,j}=\<\sigma_i\sigma_j\>$, where $m'_{i,j}$ is given by~\eqref{eq:m'_i,j:matchings}. A trivial correspondence between perfect matchings of $\Gdub$ and almost perfect matchings of $\Gbip$ yields $\sum_{\match:\partial(\match)\in\OddEven{\emptyset}}\wt(\match)= \sum_{\match\in \PM(\Gdub)} \wt(\match)$ and 
 \[ \sum_{\match:\partial(\match)\in\OddEven{\{i,j\}}}\wt(\match)=\sum_{\match\in \PM(\Gdub\setminus \{\bbip_{2i-1},\bbip_{2j}\})} \wt(\match) + \sum_{\match\in \PM(\Gdub\setminus \{\bbip_{2i},\bbip_{2j-1}\})} \wt(\match).\]
Dividing~\eqref{eq:Dub2} by~\eqref{eq:Dub1} yields the desired result.
\end{proof}

\subsection{Random alternating flows of Lis}\label{sec:alternating_flows}
For our second proof, we use a formula due to Lis~\cite{Lis}, which he proved using the \emph{random currents model} of~\cite{GHS}, see also~\cite{DC,LW}. Let us say that a \emph{clockwise bidirected edge} (resp., a \emph{counterclockwise bidirected edge}) is a directed cycle of length two in the plane which is oriented clockwise (resp., counterclockwise).

Suppose we are given a planar Ising network $N=(G,J)$ with $n$ boundary vertices and two disjoint subsets $A,B\subset[n]$ of the same size. We define $\GAB AB=(\VAB AB,\EAB AB)$ to be the graph obtained from $G$ by adding a boundary spike at $b_i$ for all $i\in A\sqcup B$.

  An \emph{$(A,B)$-alternating flow} $F$ on $G$ is a graph obtained from $\GAB A B$ by replacing each edge $\{u,v\}\in\EAB AB$ of $\GAB AB$ by either
  \begin{enumerate}[(a)]
  \item an undirected edge, or
  \item a directed edge $u\to v$ or $v\to u$, or
  \item a clockwise or a counterclockwise bidirected edge,
  \end{enumerate}
  so that the vertex $b_i$ is incident to an outgoing (resp., incoming) edge if $i\in A$ (resp., if $i\in B$), and so that every other vertex $v\in V$ of $G$ is incident to an even number of directed edges of $F$, and their directions alternate around $v$. The set of all $(A,B)$-alternating  flows on $G$ is denoted $\Fcal_{A,B}(G)$.

  For $e\in E$, we put
  \[x_e:=\tanh(J_e)=\frac{\exp(J_e)-\exp(-J_e)}{\exp(J_e)+\exp(-J_e)},\quad y_e:=\sech(J_e)=\frac{2}{\exp(J_e)+\exp(-J_e)}.\]
  (Recall that $x_e$ and $y_e$ are not the same as $c_e=\tanh(2J_e)$ and $s_e=\sech(2J_e)$.)
  Given an edge $e\in \EAB AB$ and an  $(A,B)$-alternating flow $F\in\Fcal_{A,B}(G)$, we set
  \[\w(F,e)=
    \begin{cases}
      {2x_e}/{y_e^2}, &\text{if $e$ is a directed edge in $F$;}\\
      {2x_e^2}/{y_e^2}, &\text{if $e$ is a bidirected edge in $F$;}\\
      1, &\text{otherwise.}
    \end{cases}\]
  Following~\cite[Eq.~(4.2)]{Lis}, the \emph{weight} of an $(A,B)$-alternating flow $F$ is given by
  \begin{equation}\label{eq:Lis_weight}
\w(F):=2^{|A|-|V(F)|}\prod_{e\in\EAB AB} \w(F,e),
  \end{equation}
  where $V(F)$ denotes the set of vertices $v\in\VAB AB\setminus \{b_i\mid i\in A\sqcup B\}$ incident to a directed or a bidirected edge in $F$ (note that $b_i$ is always incident to a directed edge in $F$ when $i\in A\sqcup B$).
  \begin{remark}
The equivalence of~\eqref{eq:Lis_weight} and~\cite[Eq.~(4.2)]{Lis} is explained in the proof of~\cite[Lemma~4.2]{Lis}.
  \end{remark}

  We will be interested in the two special cases $A=B=\emptyset$ and $A=\{a\}$, $B=\{b\}$ for $a\neq b\in [n]$. We denote the corresponding graphs by $\Ge$ and $\Gab$, respectively. Denote also $\Fcal_{\emptyset}(G):=\Fcal_{\emptyset,\emptyset}(G)$ and $\Fcal_{a,b}(G):=\Fcal_{\{a\},\{b\}}(G)$.
\begin{lemma}[{\cite[Lemma~5.2]{Lis}}]
  Let $N=(G,J)$ be a planar Ising network with $n$ boundary vertices, and let $i\neq j\in[n]$. Then the boundary correlation $\<\sigma_i\sigma_j\>$ equals
  \begin{equation}\label{eq:Lis_formula}
\<\sigma_i\sigma_j\>=\frac{\sum_{F\in\Fcal_{a,b}(G)} \w(F)}{\sum_{F\in\Fcal_{\emptyset}(G)} \w(F)}.
  \end{equation}
\end{lemma}

\begin{proof}[Second proof of Theorem~\ref{thm:X=X'}.]

\def\U{U}
\def\ww{\tilde{\w}}
\def\red{{\operatorname{min}}}
\def\Fcalmin{\Fcal^\red}
\def\www{\bar\w}
\def\spins{\alpha}
\def\UF{{\U(F)}}

For a flow $F\in\Fcal_{A,B}(G)$, let $\U(F)$ denote the set of vertices $v\in \VAB AB$ that are \emph{not} incident to a directed or a bidirected edge of $F$. Thus $|\U(F)|=|V|-|V(F)|$, and we set
\[\ww(F):=2^{|V|}\w(F)=2^{|A|+|\U(F)|}\prod_{e\in\EAB AB} \w(F,e).\]
Suppose that we are given a flow $F\in\Fcal_{A,B}(G)$ together with a map $\spins:\UF\to\{-1,1\}$. We say that the pair $(F,\spins)$ is a \emph{spinned flow}. The weight of a spinned flow is defined to be $\ww(F,\spins)=2^{|A|}\prod_{e\in\EAB AB} \w(F,e)$, so that $\ww(F)=\sum_{\spins\in\{-1,1\}^\UF} \ww(F,\spins)$. We then define an order relation $\leq$ on spinned flows by writing $(F,\spins)\leq (F',\spins')$ if all of the following conditions are satisfied:
\begin{itemize}
\item $F'$ is obtained from $F$ by making some undirected edges bidirected (thus $\U(F')\subset\U(F)$),
\item the restriction of $\spins$ to $\U(F')$ equals $\spins'$, and
\item for every vertex $v\in \U(F)\setminus \U(F')$ such that $\spins(v)=1$ (resp., $\spins(v)=-1$), all bidirected edges of $F'$ incident to $v$ are clockwise (resp., counterclockwise) bidirected edges.
\end{itemize}
Even though $\spins'$ can be obtained from $\spins$ by restricting it to $\U(F')\subset\U(F)$, we can also reconstruct $\spins$ from $(F',\spins')$, since every vertex $v\in \U(F)\setminus \U(F')$ is incident to at least one bidirected edge of $F'$, and either all such edges are clockwise (in which case we must have $\spins(v)=1$) or counterclockwise (in which case we must have $\spins(v)=-1$).

Given a spinned flow $(F,\spins)$, we say that an undirected edge $e$ of $F$ is \emph{active} if there exists a spinned flow $(F',\spins')>(F,\spins)$ such that $e$ is bidirected in $F'$. Thus any $(F',\spins')>(F,\spins)$ is obtained from $(F,\spins)$ by making some active edges bidirected. (An active edge can become either a clockwise or a counterclockwise bidirected edge but not both.) Equivalently, for every undirected edge $e$ of $F$ and a vertex $v$ incident to $e$, we set $\spins(v,e):=\spins(v)$ if $v\in\UF$, and $\spins(v,e):=1$ (resp., $\spins(v,e):=-1$) if $v\in V\setminus \UF$ and after replacing $e$ by a clockwise (resp., counterclockwise) bidirected edge, the directions of edges still alternate around $v$. Then an undirected edge $e=\{u,v\}$ of $F$ is active if and only if we have $\spins(v,e)=\spins(u,e)$.

We say that a spinned flow $(F,\spins)$ is \emph{minimal} if it is minimal with respect to our order relation $\leq$. Equivalently, $(F,\spins)$ is minimal if $F$ has no bidirected edges. We denote $\Fcalmin_{A,B}(G)$ the set of all minimal spinned flows $(F,\spins)$ where $F\in\Fcal_{A,B}(G)$. For $(F,\spins)\in\Fcalmin_{A,B}(G)$, we define its weight
\[\www(F,\spins):=\sum_{(F',\spins')\geq (F,\spins)}\ww(F',\spins')=2^{|A|}\prod_{e\in\EAB AB}\www(F,\spins,e),\]
where
\begin{equation}\label{eq:www}
\www(F,\spins,e)=
    \begin{cases}
      {2x_e}/{y_e^2}, &\text{if $e$ is a directed edge in $F$;}\\
      1+{2x_e^2}/{y_e^2}=(1+x_e^2)/{y_e^2}, &\text{if $e$ is an active edge of $(F,\spins)$;}\\
      1, &\text{otherwise.}
    \end{cases}
\end{equation}
  Here we used the fact that $x_e^2+y_e^2=1$.

  \begin{figure}
    \renewcommand{\drawbbip}{}
  \def\tikzscl{1}
  \def\sclbx{1.1}
  \def\sclbxx{0.7}
  \def\pmscl{2}
  \def\pluss{\scalebox{\pmscl}{\textcolor{blue}{\bf$+$}}}
  \def\minuss{\scalebox{\pmscl}{\textcolor{black!70}{\bf$-$}}}

\scalebox{\sclbxx}{
  \begin{tabular}{ccccccc}
 \scalebox{\sclbx}{
   \begin{tikzpicture}[baseline=(zero.base),scale=\tikzscl]
    \drawgraph
    \edgeplmatch{A}{AA}
    \edgeplmatch{D}{DD}
    \edgeplmatch{B}{C}
  \end{tikzpicture}
}    & 
 \scalebox{\sclbx}{
  \begin{tikzpicture}[baseline=(zero.base),scale=\tikzscl]
    \drawgraph
    \edgeplmatch{A}{D}
    \edgeplmatch{B}{BB}
    \edgeplmatch{C}{CC}
  \end{tikzpicture}
}    &  \scalebox{\sclbx}{
   \begin{tikzpicture}[baseline=(zero.base),scale=\tikzscl]
    \drawgraph
    \edgeplmatch{A}{AA}
    \edgeplmatch{B}{BB}
    \edgeplmatch{D}{C}
  \end{tikzpicture}
}    & 
 \scalebox{\sclbx}{
  \begin{tikzpicture}[baseline=(zero.base),scale=\tikzscl]
    \drawgraph
    \edgeplmatch{A}{B}
    \edgeplmatch{D}{DD}
    \edgeplmatch{C}{CC}
  \end{tikzpicture}
}   & 
 \scalebox{\sclbx}{
  \begin{tikzpicture}[baseline=(zero.base),scale=\tikzscl]
    \drawgraph
    \edgeplmatch{C}{CC}
    \edgeplmatch{D}{DD}
    \edgeplmatch{A}{AA}
    \edgeplmatch{B}{BB}
  \end{tikzpicture}
}    & 
 \scalebox{\sclbx}{
  \begin{tikzpicture}[baseline=(zero.base),scale=\tikzscl]
    \drawgraph
    \edgeplmatch{A}{B}
    \edgeplmatch{C}{D}
  \end{tikzpicture}
}     & 
 \scalebox{\sclbx}{
  \begin{tikzpicture}[baseline=(zero.base),scale=\tikzscl]
    \drawgraph
    \edgeplmatch{A}{D}
    \edgeplmatch{C}{B}
  \end{tikzpicture}
}      \\
    $\wt(\match)=c_e$ &  $\wt(\match)=c_e$ &    $\wt(\match)=s_e$ &  $\wt(\match)=s_e$ &   $\wt(\match)=1$ &   $\wt(\match)=s_e^2$ &      $\wt(\match)=c_e^2$\\

    \scalebox{\sclbx}{
  \begin{tikzpicture}[baseline=(zero.base),scale=\tikzscl]
    \coordinate (zero) at (0,0);
    \disk{0,0}
    \edgedir{-1,0}{1,0}
    \node[anchor=south,scale=\scsc] (etop) at (0,0) {$e$};
    \vertx{V1}{-1,0}
    \vertx{V2}{1,0}
  \end{tikzpicture}
  } & \scalebox{\sclbx}{
  \begin{tikzpicture}[baseline=(zero.base),scale=\tikzscl]
    \coordinate (zero) at (0,0);
    \disk{0,0}
    \edgedir{1,0}{-1,0}
    \node[anchor=south,scale=\scsc] (etop) at (0,0) {$e$};
    \vertx{V1}{-1,0}
    \vertx{V2}{1,0}
  \end{tikzpicture}
  } &  \scalebox{\sclbx}{
  \begin{tikzpicture}[baseline=(zero.base),scale=\tikzscl]
    \coordinate (zero) at (0,0);
    \disk{0,0}
    \edge{-1,0}{1,0}
    \node[anchor=south,scale=\scsc] (etop) at (0,0) {$e$};
    \vertx{V1}{-1,0}
    \vertx{V2}{1,0}
    \node[anchor=north,scale=\bscl] (b1) at (-1,0) {\minuss};
    \node[anchor=north,scale=\bscl] (b2) at (1,0) {\pluss};
  \end{tikzpicture}
  } &   \scalebox{\sclbx}{
  \begin{tikzpicture}[baseline=(zero.base),scale=\tikzscl]
    \coordinate (zero) at (0,0);
    \disk{0,0}
    \edge{-1,0}{1,0}
    \node[anchor=south,scale=\scsc] (etop) at (0,0) {$e$};
    \vertx{V1}{-1,0}
    \vertx{V2}{1,0}
    \node[anchor=north,scale=\bscl] (b1) at (-1,0) {\pluss};
    \node[anchor=north,scale=\bscl] (b2) at (1,0) {\minuss};
  \end{tikzpicture}
  } &   \scalebox{\sclbx}{
  \begin{tikzpicture}[baseline=(zero.base),scale=\tikzscl]
    \coordinate (zero) at (0,0);
    \disk{0,0}
    \edge{-1,0}{1,0}
    \node[anchor=south,scale=\scsc] (etop) at (0,0) {$e$};
    \vertx{V1}{-1,0}
    \vertx{V2}{1,0}
    \node[anchor=north,scale=\bscl] (b1) at (-1,0) {\pluss};
    \node[anchor=north,scale=\bscl] (b2) at (1,0) {\pluss};
  \end{tikzpicture}
  } &   \scalebox{\sclbx}{
  \begin{tikzpicture}[baseline=(zero.base),scale=\tikzscl]
    \coordinate (zero) at (0,0);
    \disk{0,0}
    \edge{-1,0}{1,0}
    \node[anchor=south,scale=\scsc] (etop) at (0,0) {$e$};
    \vertx{V1}{-1,0}
    \vertx{V2}{1,0}
    \node[anchor=north,scale=\bscl] (b1) at (-1,0) {\minuss};
    \node[anchor=north,scale=\bscl] (b2) at (1,0) {\minuss};
  \end{tikzpicture}
  } &   \scalebox{\sclbx}{
  \begin{tikzpicture}[baseline=(zero.base),scale=\tikzscl]
    \coordinate (zero) at (0,0);
    \disk{0,0}
    \edge{-1,0}{1,0}
    \node[anchor=south,scale=\scsc] (etop) at (0,0) {$e$};
    \vertx{V1}{-1,0}
    \vertx{V2}{1,0}
    \node[anchor=north,scale=\bscl] (b1) at (-1,0) {\minuss};
    \node[anchor=north,scale=\bscl] (b2) at (1,0) {\minuss};
  \end{tikzpicture}
  } 
  \end{tabular}
  }
  
  \def\pmscl{1}
  \caption{\label{fig:theta}The correspondence $\theta$ between almost perfect matchings of $\Gbip$ (top) and minimal spinned flows (bottom),
    where \textcolor{blue}{\bf$+$} (resp., \textcolor{black!70}{\bf$-$})
    next to a vertex $v$ denotes $\spins(v,e)=1$ (resp., $\spins(v,e)=-1$).}
\end{figure}
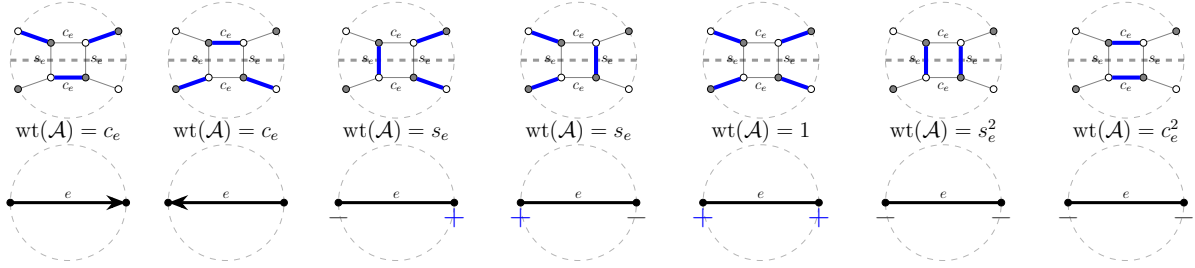
  
  It thus follows that
  \[\sum_{F\in\Fcal_{A,B}(G)}\ww(F)=\sum_{(F,\spins)\in\Fcalmin_{A,B}(G)}\www(F,\spins).\]
  Our goal is to give a map $\theta$ from almost perfect matchings of $\Gbip$ to minimal spinned flows, which locally is defined in Figure~\ref{fig:theta}. Namely, each edge $e=\{u,v\}$ of $G$  corresponds to  four interior vertices of $\Gbip$, as in Figure~\ref{fig:Gbip_small}. Every almost perfect matching $\Acal$ of $\Gbip$ assigns a single edge to each of those four vertices, and there are seven ways to do so, as in Figure~\ref{fig:theta} (top). The product of the weights of edges of $\Acal$ incident to one of the four vertices of $\Gbip$ equals, respectively, to $c_e,c_e,s_e,s_e,1,s_e^2,c_e^2$, see Figure~\ref{fig:theta} (top).

  Similarly, for every minimal spinned flow $(F,\spins)$, $e$ may be directed from $u$ to $v$, or directed from $v$ to $u$, or undirected, in which case the functions $\spins(u,e),\spins(v,e)\in\{-1,1\}$ are well defined. As shown in Figure~\ref{fig:theta}, the two matchings of weight $c_e$ correspond to the case where $e$ is directed in $F$, and the remaining five matchings correspond to $e$ being undirected in $F$. Specifically, the two matchings of weight $s_e$ correspond to the two cases where $\spins(u,e)\neq\spins(v,e)$, the matching of weight $1$ corresponds to the case $\spins(u,e)=\spins(v,e)=1$, and the two matchings of weights $s_e^2,c_e^2$ correspond to a single case $\spins(u,e)=\spins(v,e)=-1$.

  It is straightforward to check that these rules give a well defined map $\theta$ from the set of almost perfect matchings of $\Gbip$ to the set of minimal spinned flows on $G$.
  Moreover, it is easy to check that the set $J:=\partial(\Acal)\subset [2n]$ determines uniquely two disjoint sets $A,B\subset [n]$ such that $\theta(\Acal)\in\Fcalmin_{A,B}(G)$.
  Namely, we have $A=\{i\in[n]\mid 2i-1,2i\notin J\}$ and $B=\{i\in[n]\mid 2i-1,2i\in J\}$.
  Finally, let $(F,\spins)\in\Fcalmin_{A,B}(G)$ be a minimal spinned flow, then we claim that
  \begin{equation}\label{eq:www_leningrad}
\www(F,\spins)= \frac1{\prod_{e\in E}s_e}\sum_{\Acal:\theta(\Acal)=(F,\spins)}\wt(\Acal),
  \end{equation}
  where the sum is over almost perfect matchings $\Acal$ of $\Gbip$. To see why this is the case, note that the multiplicative contribution of an edge $e\in E$ to $\www(F,\spins)$ is given by~\eqref{eq:www}. On the other hand, it is clear from Figure~\ref{fig:theta} that for any two almost perfect matchings $\match,\match'$ such that $\theta(\match)=\theta(\match')$, we have $\proj(\match)=\proj(\match')$, where $\proj(\match)$ is defined in the proof of Proposition~\ref{prop:matchings_complements}. Thus the total weight of almost perfect matchings in the preimage of $(F,\spins)$ under $\theta$ equals to the product over all edges $e\in E$ of $\insideprod(e)$, defined in the proof of Proposition~\ref{prop:matchings_complements} as
  \begin{equation}\label{eq:leningrad}
\insideprod(e)=  \begin{cases}
      c_e, &\text{if $e$ is a directed edge in $F$;}\\
      1=s_e^2+c_e^2, &\text{if $e$ is an active edge of $(F,\spins)$;}\\
      s_e, &\text{otherwise.}
    \end{cases}
  \end{equation}
  Indeed, if $e=\{u,v\}$ is an active edge of $(F,\spins)$ then we either have $\spins(u,e)=\spins(v,e)=1$ in which case $e$ corresponds locally to a single matching of weight $1$, or we have $\spins(u,e)=\spins(v,e)=-1$ in which case $e$ corresponds locally to two matchings of weights $s_e^2$ and $c_e^2$, which can be interchanged in every almost perfect matching in the preimage of $(F,\spins)$ under $\theta$. It remains to note that the right hand side of~\eqref{eq:leningrad} can be obtained from the right hand side of~\eqref{eq:www} by multiplying by $s_e$:
\[s_e\www(F,\spins,e)=
  \begin{cases}
      {2x_es_e}/{y_e^2}=c_e, &\text{if $e$ is a directed edge in $F$;}\\
      (1+x_e^2)s_e/{y_e^2}=1, &\text{if $e$ is an active edge of $(F,\spins)$;}\\
      s_e, &\text{otherwise.}
    \end{cases}\]
Thus $s_e\www(F,\spins,e)=\insideprod(e)$, which proves~\eqref{eq:www_leningrad}. This implies that the right hand sides of~\eqref{eq:Lis_formula} and~\eqref{eq:m'_i,j:matchings} are equal, finishing the second proof of Theorem~\ref{thm:X=X'}.
\end{proof}

\section{Cyclic symmetry and a homeomorphism with a ball}\label{sec:ball}
By Theorems~\ref{thm:X=X'} and~\ref{thm:Gmed_parametrization}, the map $\doublemap$ is a stratification-preserving homeomorphism from $\Closure_n$ to $\OGtnn(n,2n)$, which is the first part of Theorem~\ref{thm:main}. In this section, we follow the strategy of~\cite{GKL} to prove the second part of Theorem~\ref{thm:main}, which states that $\Closure_n$ is homeomorphic to a closed ball of dimension $n\choose 2$.

\def\expttau{\exp(t(S+S^T))}
Recall that the cyclic shift $2n\times 2n$ matrix $S$ was defined in Section~\ref{sec:cyclic_intro}. We let $S^T$ denote the matrix transpose of $S$. We recall the following result from~\cite{GKL}.
\begin{lemma}[{\cite[Corollary~3.8]{GKL}}]\label{lemma:expts_Grtp}
For all $X\in\Grtnn(k,\n)$ and all $t>0$, we have $X\cdot \expttau\in \Grtp(k,\n)$.
\end{lemma}
Recall that the totally positive Grassmannian $\Grtp(k,\n)$ is defined in~\eqref{eq:Grtp}. Let us define the \emph{totally positive orthogonal Grassmannian} to be the intersection
\begin{equation*}\label{eq:OGtp}
\OGtp(n,2n):=\Grtp(n,2n)\cap \OG(n,2n).
\end{equation*}

\begin{lemma}\label{lemma:expttau}
For all $X\in\OGtnn(n,2n)$ and all $t>0$, we have $X\cdot \expttau\in \OGtp(n,2n)$.
\end{lemma}
\begin{proof}
  In view of Lemma~\ref{lemma:expts_Grtp}, it suffices to show that $X\cdot \expttau\in \OG(n,2n)$. By Proposition~\ref{prop:eta}, it is enough to prove that $\expttau$ belongs to the Lie group $O(n,n)$ consisting of all $2n\times 2n$ matrices $g$ preserving the bilinear form $\eta$, i.e., satisfying $\eta(gu,gv)=\eta(u,v)$ for all $u,v\in\R^{2n}$. It is a standard fact from Lie theory that $\expttau$ is such a matrix if and only if $S+S^T$ belongs to the Lie algebra of $O(n,n)$. Let $D:=\diag(1,-1,1,-1,\dots,1,-1)$ be a $2n\times 2n$ diagonal matrix with $D_{i,i}=(-1)^{i-1}$ for $1\leq i\leq 2n$. Then the Lie algebra of $O(n,n)$ consists of all $2n\times 2n$ matrices $B$ such that $B\cdot D=-D\cdot B^T$. It is easy to check that $S+S^T$ belongs to this Lie algebra. We are done with the proof.
\end{proof}

\begin{example}
For $n=2$, the computation we need to check that $S+S^T$ belongs to the Lie algebra of $O(n,n)$ goes as follows.
  \[(S+S^T)\cdot D=\begin{pmatrix}
      0 & 1 & 0 & -1\\
      1 & 0 & 1 & 0\\
      0 & 1 & 0 & 1\\
      -1 & 0 & 1 & 0
    \end{pmatrix}\cdot \begin{pmatrix}
      1 & 0 & 0 & 0\\
      0 & -1 & 0 & 0\\
      0 & 0 & 1 & 0\\
      0 & 0 & 0 & -1
    \end{pmatrix}=\begin{pmatrix}
      0 & -1 & 0 & 1\\
      1 & 0 & 1 & 0\\
      0 & -1 & 0 & -1\\
      -1 & 0 & 1 & 0
    \end{pmatrix};\]
  \[D\cdot (S+S^T)^T =\begin{pmatrix}
      1 & 0 & 0 & 0\\
      0 & -1 & 0 & 0\\
      0 & 0 & 1 & 0\\
      0 & 0 & 0 & -1
    \end{pmatrix}\cdot\begin{pmatrix}
      0 & 1 & 0 & -1\\
      1 & 0 & 1 & 0\\
      0 & 1 & 0 & 1\\
      -1 & 0 & 1 & 0
    \end{pmatrix}=\begin{pmatrix}
      0 & 1 & 0 & -1\\
      -1 & 0 & -1 & 0\\
      0 & 1 & 0 & 1\\
      1 & 0 & -1 & 0
    \end{pmatrix}.\]
  This shows $(S+S^T)\cdot D=-D\cdot (S+S^T)^T$ for $n=2$.
\end{example}

\begin{remark}\label{rmk:critical}
  For all $X\in\Grtnn(n,2n)$, it was shown in~\cite{GKL} that the limit of $X~\cdot~\exp(t(S+S^T))$ as $t\to\infty$ is the unique cyclically symmetric element $X_0\in\Grtnn(n,2n)$ from Proposition~\ref{prop:X_0}.  It follows from Lemma~\ref{lemma:expttau} that this point $X_0$ belongs to $\OGtnn(n,2n)$.
\end{remark}

\def\flow{f}
\newcommand\act[2]{\flow(#1,#2)}
\def\openball{Q}
\def\closedball{\overline{\openball}}

\begin{proof}[Proof of Theorem~\ref{thm:main}]
  As we have already discussed, the first part is a direct consequence of Theorems~\ref{thm:X=X'} and~\ref{thm:Gmed_parametrization}. The second part follows from Lemma~\ref{lemma:expttau} together with an argument completely identical to the one in~\cite{GKL}, which we briefly outline here. We refer the interested reader to~\cite{GKL} for more details.
\begin{definition}[{\cite[Definition~2.1]{GKL}}]
  A map $\flow:\R\times \R^N$ is called a \emph{contractive flow} if the following conditions hold for $f$:
  \begin{enumerate}
\item\label{item:continuous} the map $\flow$ is continuous;
\item\label{item:action} for all $p\in\RR^N$ and $t_1,t_2\in\RR$, we have $\act 0 p=p$ and $\act{t_1+t_2}{p}=\act{t_1}{\act{t_2}{p}}$;
\item\label{eq:contract} for all $p\neq 0$ and $t > 0$, we have $\|\act{t}{p}\| < \|p\|$.  \end{enumerate}
\end{definition}
Here $\|\cdot\|$ denotes the Euclidean norm on $\R^N$. A useful feature of having a contractive flow is the following result.
\begin{lemma}[{\cite[Lemma~2.3]{GKL}}]
Let $\openball \subset \RR^N$ be a smooth embedded submanifold of dimension $d \leq N$, and $\flow:\RR\times \RR^N\to\RR^N$ a contractive flow. Suppose that $\openball$ is bounded and satisfies the condition
\begin{align}\label{eq:invariant}
\act{t}{\closedball} \subset \openball \quad \text{ for $t > 0$}.
\end{align}
Then the closure $\closedball$ is homeomorphic to a closed ball of dimension $d$.
\end{lemma}
It was shown in the proof of~\cite[Theorem~1.1]{GKL} that the space $\Grtnn(n,2n)$ can be explicitly realized as a subset of $\R^N$  so that the image of $\Grtp(n,2n)$ would be an embedded submanifold of $\RR^N$, and that the action of $\expttau$ on $\Grtnn(n,2n)$ extends to a contractive flow on $\R^N$. Since $\OG(n,2n)$ is an embedded submanifold of $\Gr(n,2n)$, we see that $\openball:=\OGtp(n,2n)$ becomes an embedded submanifold of $\R^N$ whose closure is $\closedball:=\OGtnn(n,2n)$ in $\R^N$. By Lemma~\ref{lemma:expttau}, the contractive flow $\expttau$ restricts to $\OGtnn(n,2n)$ and satisfies~\eqref{eq:invariant}. The result follows.
\end{proof}

Theorem~\ref{thm:main} establishes the correspondence between the planar Ising model and the totally nonnegative orthogonal Grassmannian. Having finished its proof, we are in a position to deduce several other results stated in Section~\ref{sec:applications}.

\def\wtdual{\wt^\ast}
\def\bbipdual{\bbip^\ast}
\def\Gbipdual{(G^\ast)^\square}

\begin{proof}[Proof of Theorem~\ref{thm:planar_dual}.]
  This follows easily from studying the relationship of the map
  \[(G,J)\mapsto (\Gbip,\wt)\mapsto \Meast(\Gbip,\wt)\in\OGtnn(n,2n)\]
  with the duality map $(G,J)\mapsto (\Gdual,\Jdual)$. Namely, a planar Ising network $N=(G,J)$ corresponds to a weighted plabic graph $(\Gbip,\wt)$ then the dual planar Ising network $\Ndual=(\Gdual,\Jdual)$ corresponds to a weighted plabic graph $(\Gbipdual,\wtdual)$ so that $\Gbipdual$ is obtained from $\Gbip$ by switching the colors of all vertices and cyclically relabeling  boundary vertices (i.e., $\bbipdual_{i}:=\bbip_{i+1}$), and $\wtdual$ is obtained from $\wt$ by swapping $s_e$ and $c_e$ for all $e\in E$. More precisely, for each $e\in E$ we have $\sinh(2J_e)\sinh(2\Jdual_{\edual})=1$ by~\eqref{eq:duality}. On the other hand, by~\eqref{eq:sinh_tanh}, we have $\sinh(2J_e)=\frac{c_e}{s_e}$ and $\sinh(2\Jdual_\edual)=\frac{c_\edual}{s_\edual}$, so $s_\edual=c_e$ and $c_\edual=s_e$.  It thus follows from the definition of $\Meast$ given in~\eqref{eq:Delta_I_matchings} that the minor $\Delta_I$ of $\Meast(\Gbip,\wt)$ equals the minor $\Delta_{I'}$ of $\Meast(\Gbipdual,\wtdual)$ where $I'=\{i+1\mid i\in I\}$ (modulo $2n$) for all $I\in{[2n]\choose n}$. This is equivalent to having $\Meast(\Gbip,\wt)\cdot S=\Meast(\Gbipdual,\wtdual)$, which finishes the proof.
\end{proof}

\begin{proof}[Proof of Proposition~\ref{prop:cyclic_fixed_pt}.]
We know by Proposition~\ref{prop:X_0} that there exists a unique cyclically symmetric element $X_0\in\Grtnn(n,2n)$, and by Remark~\ref{rmk:critical}, we have $X_0\in\OGtnn(n,2n)$. By Theorem~\ref{thm:main}, $X_0$ corresponds to a unique boundary correlation matrix $M_0\in\Closure$ of a planar Ising network (i.e., $\doublemap(M_0)=X_0$). Since the operation $N=(G,J)\mapsto \Ndual=(\Gdual,\Jdual)$ amounts to applying the cyclic shift on $\OGtnn(n,2n)$ by Theorem~\ref{thm:planar_dual}, we see that $M_0=M(G,J)$ if and only if $M(\Gdual,\Jdual)=M(G,J)$.
\end{proof}

\begin{proof}[Proof of Theorem~\ref{thm:bound_measurements}.]
This also follows easily from Theorem~\ref{thm:main} combined with~\eqref{eq:Delta_I_matchings}.
\end{proof}

\begin{proof}[Proof of Theorem~\ref{thm:reduced_injective}.]
Follows from Theorem~\ref{thm:main} and part~\eqref{item:Gmed_param:homeo} of Theorem~\ref{thm:Gmed_parametrization}.
\end{proof}

\begin{proof}[Proof of Theorem~\ref{thm:inverse}.]
Follows from Theorems~\ref{thm:bridge_removal} and~\ref{thm:main}.
\end{proof}

\def\Gbipprime{(G')^\square}
\begin{proof}[Proof of Theorem~\ref{thm:duality}.]
  Indeed, adjoining a boundary spike $e$ to $G'$ amounts to adding a pair of bridges to $\Gbipprime$. Adding bridges to $\Gbipprime$ translates into acting by $x_\kk(s_e)$ and $y_{\kk+1}(s_e)$ on $\Meast(\Gbipprime,\wt')$ by Lemma~\ref{lemma:bridge_removal}. However, we also rescale the edges incident to $\kk$ and $\kk+1$ by $c_e$ between adding the two bridges, which amounts to multiplying $\Meast(\Gbipprime,\wt)$ by a diagonal matrix $D_\kk(c_e)$ whose $(\kk,\kk)$-th and $(\kk+1,\kk+1)$-th entries are equal to $c_e$ and  $1/c_e$, respectively. Thus if $N=(G,J)$ is obtained from $N'=(G',J')$ by adjoining a boundary spike, then the matrices $M=M(G,J)$ and $M'=M(G',J')$ are related by
  \[\doublemap(M)=\Meast(\Gbip,\wt)=\Meast(\Gbipprime,\wt')\cdot x_\kk(s_e)\cdot D_\kk(c_e)\cdot y_{\kk+1}(s_e)=\Meast(\Gbipprime,\wt')\cdot g_{\kk}(t),\]
  which is equal to $\doublemap(M')\cdot g_{\kk}(t)$. Here $x_\kk(s_e)\cdot D_\kk(c_e)\cdot y_{\kk+1}(s_e)= g_{\kk}(t)$ reduces to the following $2\times 2$ matrix computation, which relies on $s_e^2+c_e^2=1$:
  \[\begin{pmatrix}
      1 & s_e\\
      0 & 1
    \end{pmatrix}\cdot \begin{pmatrix}
      c_e & 0\\
      0 & 1/c_e
    \end{pmatrix}\cdot \begin{pmatrix}
      1 & 0\\
      s_e & 1
    \end{pmatrix}=\begin{pmatrix}
  1/c_e & s_e/c_e\\
  s_e/c_e & 1/c_e
\end{pmatrix}.\]
We are done with the case of adjoining a boundary spike. The case of adjoining a boundary edge is completely similar, and also follows by applying the duality from Section~\ref{sec:cyclic_intro}, which switches between $s_e$ and $c_e$ due to~\eqref{eq:duality}. We are done with the proof.
\end{proof}

\begin{proof}[Proof of Proposition~\ref{prop:Closure_cells}.]
It follows from Theorem~\ref{thm:X=X'} that $\doublemap$ sends $\Space_\medpa$ homeomorphically onto the cell $\pc_\medpa\cap\OGtnn(n,2n)$, and thus the result follows from Theorem~\ref{thm:main} combined with Theorem~\ref{thm:Gmed_parametrization}~\eqref{item:Gmed_param:disjoint}.
\end{proof}

\section{Generalized Griffiths' inequalities}\label{sec:griffiths}
\def\sd{C}
\def\oesd{\OddEven{\sd}}
\def\oesdgood{\OddEven{\sd}\cap\Sumparity{B}{\epsilon}}
\def\oesdbad{\OddEven{\sd}\cap\Sumparity{B}{1-\epsilon}}
\def\odd{\operatorname{odd}}
\def\even{\operatorname{even}}
\def\codd{\beta}
\def\ceven{\alpha}
\newcommand\mon[1]{m_{\codd,#1}}
In this section, our goal is to prove Theorem~\ref{thm:generalized_Griffiths}. Note that~\eqref{eq:easy_Griffiths} follows from~\eqref{eq:Griffiths} by taking disjoint $A$ and $B$ such that $|B|=1$. We thus focus on proving~\eqref{eq:Griffiths}. Let us fix two subsets $A,B\subset [n]$, and let $C:=\symdiff{A}{B}$ be their symmetric difference. If $\sd$ has odd size then both sides of~\eqref{eq:Griffiths} become zero. Thus we assume that the size of $\sd$ is even. Recall that $\oesd\subset{[2n]\choose n}$ consists of all $n$-element subsets $I$ of $[2n]$ such that $I\cap\{2i-1,2i\}$ has even size if and only if $i\in \sd$. (In particular, $\oesd$ is empty when $\sd$ has odd size.)

Throughout, we also fix a matrix $M=(m_{i,j})\in\Matsym$, and we treat the entries $m_{i,j}=m_{j,i}$ as indeterminates for $i\neq j$.

Our first goal is to give a formula for the minors $\Delta_I(\double M)$ for $I\in \oesd$.
\begin{definition}
Denote $n':=n-|\sd|/2$. Let $\ceven:[2n]\to [2n']$ be the unique order-preserving map such that $\ceven(2i-1)=\ceven(2i)$ if and only if $i\in \sd$. Let $\codd:[2n']\to[n]$ be the unique order-preserving map such that the composition $\codd\circ\ceven:[2n]\to[n]$ sends both $2i-1$ and $2i$ to $i$ for all $i\in[n]$.
\end{definition}

\def\twolinescl{0.8}
\def\threelinescl{0.7}
\begin{example}\label{ex:Griffiths_codd_ceven}
  Suppose that $n=4$ and $\sd=\{1,3\}$. Then $n'=3$, and the map $\ceven:[8]\to[6]$ sends the top row entries of the $2$-line array \scalebox{\twolinescl}{$\begin{tabular}{|cc|c|c|cc|c|c|}
  1&2&3&4&5&6&7&8\\
  1&1&2&3&4&4&5&6
\end{tabular}$} to the corresponding bottom row entries (i.e., $\ceven(1)=\ceven(2)=1,\ceven(3)=2$, etc.). Similarly, $\codd:[6]\to[4]$ sends the top row entries of \scalebox{\twolinescl}{$\begin{tabular}{|c|cc|c|cc|}
  1&2&3&4&5&6\\
  1&2&2&3&4&4
\end{tabular}$} to its bottom row entries, giving rise to a composite map $\codd\circ\ceven$ represented by a $3$-line array \scalebox{\threelinescl}{$\begin{tabular}{|cc|cc|cc|cc|}
  1&2&3&4&5&6&7&8\\
  1&1&2&3&4&4&5&6\\  1&1&2&2&3&3&4&4\\                                                                                                   \end{tabular}$}.
\end{example}

\newcommand\MIJ[2]{\operatorname{Match}(#1,#2)}
\newcommand\MK[1]{\operatorname{Match}(#1)}
\newcommand\monn[1]{m_{#1}}

For disjoint subsets $I,J\subset[2\n]$ of the same size, we say that $\pi$ is a matching \emph{between $I$ and $J$} if $\pi$ contains $|I|=|J|$ pairs, and for each pair $\{i,j\}\in\pi$, we have either $i\in I, j\in J$ or $i\in J, j\in I$. The set of matchings between $I$ and $J$ is denoted by $\MIJ IJ$. For a subset $K\subset[2\n]$ of even size, a \emph{matching on $K$} is a partition of $K$ into $|K|/2$ disjoint subsets of size $2$, and we let $\MK K$ denote the set of matchings on $K$. Thus $\MIJ IJ\subset\MK{I\sqcup J}$, and $\MK{[2n]}$ is as a set equal to $\P_n$. The function $\xing$ naturally extends to $\MIJ IJ$ and $\MK K$.

 For each $I\in\oesd$, we denote $I':=\ceven(I)$, and it is easy to check that we have $I'\in{[2n']\choose n'}$ for $I\in\oesd$.  Given a matching $\pi$ on $[2n']$, we define a monomial $\mon\pi:=\prod_{\{i,j\}\in\pi}m_{\codd(i),\codd(j)}$. Similarly, given a subset $K\subset[n]$ of even size and a matching $\pi\in\MK K$, we set $\monn\pi:=\prod_{\{i,j\}\in\pi}m_{i,j}$.

\begin{proposition}\label{prop:minors_pf}
  For $I\in \oesd$, we have
  \[\Delta_I(\double M)=2^{|\sd|/2} \sum_{\pi\in\MIJ{I'}{[2n']\setminus I'}}(-1)^{\xing(\pi)} \mon\pi.\]
\end{proposition}
\begin{proof}
This is essentially~\cite[Proposition~5.2]{Pos}, see also~\cite[Eq.~(2.2)]{Lis}.
\end{proof}

\begin{remark}
For any $I\in{[2n]\choose n}$, there exists a unique $\sd\subset [n]$ such that $I\in\oesd$. Thus Proposition~\ref{prop:minors_pf} actually  gives a formula for all maximal minors of $\double M$ in terms of the entries of $M$.
\end{remark}
  
\begin{example}
  Let $n=4$ and $C=\{1,3\}$ as in Example~\ref{ex:Griffiths_codd_ceven}, so $n'=3$. The matrices $M$ and $\double M$ are given in Figure~\ref{fig:double}. Let $I:=\{1,2,4,7\}$. We have $I\in\oesd$ since $|I\cap\{1,2\}|=2$ and $|I\cap \{5,6\}|=0$ are both even, while $|I\cap\{3,4\}|=|I\cap \{7,8\}|=1$ are both odd. Next, $I'=\ceven(I)=\{1,3,5\}\in{[2n']\choose n'}$. Computing the maximal minor $\Delta_I(\double M)$, we find
  \[\Delta_I(\double M)=
    2(m_{14} m_{23} m_{24} - m_{13} m_{24}^{2} + m_{12} m_{24} m_{34} + m_{12} m_{23} + m_{14} m_{34} + m_{13}).\]
These six terms correspond to the six elements of $\MIJ{I'}{[2n']\setminus I'}=\MIJ{\{1,3,5\}}{\{2,4,6\}}$. For instance, the term $- m_{13} m_{24}^{2}$ comes from the matching $\pi=\{\{1,4\},\{3,6\},\{5,2\}\}$ with $\xing(\pi)=3$, while the term $m_{13}$ comes from the matching $\pi=\{\{1,4\},\{3,2\},\{5,6\}\}$ with $\xing(\pi)=0$.
\end{example}

\begin{definition}
  We introduce two disjoint subsets $A',B'\subset[2n']$ by:
  \begin{equation*}
    \begin{split}
      A'&:=\{i\in[2n']\mid \codd(i)\in A\setminus B\}\cup \{i\in [2n']\mid \codd(i)=\codd(i+1)\in A\cap B\};\\
      B'&:=\{i\in[2n']\mid \codd(i)\in B\setminus A\}\cup \{i+1\in [2n']\mid \codd(i)=\codd(i+1)\in A\cap B\}.
    \end{split}
  \end{equation*}
Define the number $\epsilon\in\{0,1\}$ mentioned in Theorem~\ref{thm:generalized_Griffiths} by
\begin{equation}\label{eq:epsilon_Griffiths}
  \epsilon \equiv 1+\sum_{i\in B'} i\pmod 2.
\end{equation}
\end{definition}

\def\Pf{\operatorname{Pf}}
Next, we state a classical result expressing correlations of the Ising model in terms of \emph{Pfaffians}. Given a set $K\subset[n]$ of even size, we define 
\[\Pf_K(M):=\sum_{\pi\in\MK K} (-1)^{\xing(\pi)} m_\pi.\]
If the size of $K$ is odd, we set $\Pf_K(M):=0$. The following classical result expresses multi-point correlations in terms of two-point correlations.
\begin{proposition}[{\cite[Theorem~A]{GBK}}]\label{prop:Pf}
Given a planar Ising network $N=(G,J)$, let $M=M(G,J)$ be its boundary correlation matrix. Then  for every set $K\subset[n]$, we have 
\[\<\sigma_K\>=\Pf_K(M)=\sum_{\pi\in\MK K} (-1)^{\xing(\pi)} \prod_{\{i,j\}\in\pi} \<\sigma_i\sigma_j\>.\]
\end{proposition}

Thus Theorem~\ref{thm:generalized_Griffiths} becomes a consequence of the following result.

\begin{theorem}\label{thm:pfaffians_Griffiths}
  We have
\begin{equation}\label{eq:pfaffians_Griffiths}
\Pf_{\sd}(M)-\Pf_A(M)\Pf_B(M)=\frac1{2^{n-1}}\sum_{I\in \OddEven{\symdiff{A}{B}}\cap \Sumparity{B}{\epsilon}}\Delta_I(\double M).
\end{equation}
\end{theorem}
Both sides of~\eqref{eq:pfaffians_Griffiths} are polynomials in the entries of $M$ by Propositions~\ref{prop:minors_pf} and~\ref{prop:Pf}.
\begin{remark}
It may look like the right hand side of~\eqref{eq:pfaffians_Griffiths} is not symmetric with respect to $A$ and $B$, but in fact it is easy to see that 
\[\OddEven{\symdiff{A}{B}}\cap \Sumparity{B}{\epsilon}=\OddEven{\symdiff{A}{B}}\cap \Sumparity{A}{\epsilon'},\]
where $\epsilon'\equiv 1+n+\sum_{i\in A'} i\pmod 2$.
\end{remark}
Before we prove Theorem~\ref{thm:pfaffians_Griffiths}, we state a lemma which will be used repeatedly later.
\begin{lemma}\label{lemma:xing_parity}
Let $[2\n]=K_1\sqcup K_2$ for two sets $K_1,K_2$ of even size. Let $\pi_1\in\MK{K_1}$, $\pi_2\in\MK{K_2}$, and let $\pi_1\sqcup \pi_2\in\MK{[2\n]}$ be obtained by superimposing $\pi_1$ and $\pi_2$. Then
\begin{equation}\label{eq:xing_parity}
\xing(\pi_1\sqcup \pi_2)-\xing(\pi_1)-\xing(\pi_2)\equiv |K_1|/2+\sum_{i\in K_1}i\equiv |K_2|/2+\sum_{i\in K_2}i \pmod2.
\end{equation}
\end{lemma}
\begin{proof}
Suppose that there is $i\in K_1$ such that $i>1$ and $i-1\notin K_1$. Then replacing $K_1$ with $K_1\setminus \{i\}\cup\{i-1\}$ and modifying $\pi_1,\pi_2$ accordingly changes the parity of each side of~\eqref{eq:xing_parity}. Applying this operation repeatedly until $K_1=[|K_1|]$, the result follows.
\end{proof}

\begin{proof}[Proof of Theorem~\ref{thm:pfaffians_Griffiths}.]
First, it is straightforward to check that if $i\in [n]\setminus (A\cup B)$ then removing $i$ from $[n]$ does not affect the left and right hand sides of~\eqref{eq:pfaffians_Griffiths}. Thus \emph{from now on we assume that $A\cup B=[n]$.}

\def\cleft{c_{\operatorname{left}}}
\def\cright{c_{\operatorname{right}}}

{\bf {Assume first that $A \cap B = \emptyset$}}. This implies $\sd = A \sqcup B = [n]$, $n$ is even, $n'=n/2$, $A'=A$, and $B'=B$.
For a matching $\pi\in\MK{[n]}$, we are going to compare the coefficients of $\monn{\pi}$ on both sides of~\eqref{eq:pfaffians_Griffiths}, and show that in all cases they are equal. 

\newcommand\Mres[2]{\operatorname{Match}|^{\operatorname{res}}_{#1,#2}}
Here for two disjoint subsets $I$ and $J$, we say that a matching $\pi\in\MK{I\sqcup J}$ {\emph {restricts to $I$ and $J$}} and write $\pi\in\Mres IJ$ if for all $\{i,j\}\in\pi$ we have either $\{i,j\}\subset I$ or $\{i,j\}\subset J$. We denote by $\pi|_{I}\in\MK I$ and $\pi|_{J}\in\MK J$ the corresponding restricted matchings. Thus the set $\Mres IJ\subset \MK{I\sqcup J}$ is in bijection with $\MK I\times \MK J$.

For $\pi\in\MK{[n]}$, the coefficient of $\monn{\pi}$ in $\Pf_{\sd}(M)-\Pf_A(M)\Pf_B(M)$ is equal to 
\[
\cleft(\pi) = 
\begin{cases}
(-1)^{\xing(\pi)} - (-1)^{\xing(\pi|_{A})} (-1)^{\xing(\pi|_{B})} & \text{if $\pi\in\Mres AB$;}\\
(-1)^{\xing(\pi)} & \text{otherwise.}
\end{cases}
\]

For the right hand side of~\eqref{eq:pfaffians_Griffiths}, observe that by Definition~\ref{dfn:double_sumparity}, a given set $I\in \oesd$ belongs to $\Sumparity{B}{\epsilon}$ if and only if
\[\sum_{i \in I\cap \double B} i  \equiv 1 +\sum_{i \in B} i \pmod 2,\]
because $B=B'$. Since $\sd=[n]$, we have $I\in \oesd$ if and only if $I\cap \{2i-1,2i\}$ has even size for all $i\in[n]$.  Let us say that a set $I\in\oesd$ is \emph{compatible with $\pi$} if  $\pi\in\MIJ{I'}{[2n']\setminus I'}$. It is clear that the coefficient of $\monn{\pi}$ in the right hand side of~\eqref{eq:pfaffians_Griffiths} is equal to 
\[\cright(\pi):=\frac{2^{n/2}}{2^{n-1}}  (-1)^{\xing(\pi)} N(\pi),\]
where $N(\pi)$ is the number of $I\in \oesdgood$ compatible with $\pi$. We claim that $N(\pi)$ is given by
\begin{equation}\label{eq:c_right_easy}
N(\pi) = 
\begin{cases}
2^{n/2} & \text{if $\pi\in\Mres AB$ and $|B|/2 \equiv 1 +\sum_{i \in B} i \pmod 2$;}\\
0 & \text{if  $\pi\in\Mres AB$ and $|B|/2 \not \equiv 1 +\sum_{i \in B} i \pmod 2$;}\\
2^{n/2-1} & \text{if  $\pi\notin\Mres AB$.}
\end{cases}
\end{equation}
 
Indeed, assume first $\pi\notin\Mres AB$. Then there exists a pair  $\{i,j\} \in \pi$ such that $i \in A$ and $j \in B$. Note that there are a total of $2^{n/2}$ sets $I\in\oesd$ compatible with $\pi$. Each such set satisfies either  $2i-1,2i\in I, 2j-1,2j\notin I$ or  $2j-1,2j\in I, 2i-1,2i\notin I$, so they naturally split into pairs $\{I,\symdiff{I}{\{2i-1,2i,2j-1,2j\}}\}$. Exactly one set $I$ in each pair satisfies $\sum_{i \in I \cap \double B} i\equiv \epsilon \pmod 2$. Thus the total number $N(\pi)$ of sets $I\in\oesdgood$ compatible with $\pi$ equals $2^{n/2-1}$ in this case.
 
Assume now that $\pi\in\Mres AB$. Then for any $I\in\oesd$ compatible with $\pi$, we have $\sum_{i \in I\cap \double B} i \equiv |B|/2 \pmod 2$. Thus, either all $I$ compatible with $\pi$ belong to $\oesdgood$, in which case we get $2^{n/2}$ of them, or they all belong to $\oesdbad$, in which case we get $N(\pi)=0$. It is easy to check that the former case happens  exactly when $|B|/2 \equiv 1 +\sum_{i \in B} i \pmod 2$. This shows~\eqref{eq:c_right_easy}, which, combined with~\eqref{eq:xing_parity}, clearly implies $\cleft(\pi)=\cright(\pi)$. We are done with the case $A\cap B=\emptyset$.

{\bf {Assume now that $A \cap B \not = \emptyset$}}. Since we are assuming $A\cup B=[n]$, we have $2n'=n+|A\cap B|$, and $[2n']=A'\sqcup B'$.

For $k,k+1\in[2n']$ such that $\codd(k) = \codd(k+1) = j$, let the {\emph {flipping}} of a matching $\pi\in\MK{[2n']}$ {\emph {at $j$}} be a matching $\pi'$ obtained from $\pi'$ by ``swapping'' the elements $k, k+1$, i.e., $\pi'=\pi\setminus\{\{a,k\},\{b,k+1\}\}\cup\{\{a,k+1\},\{b,k\}\}$ for some $a,b\in[2n']$. (If $\{k,k+1\}\in\pi$ then we set $\pi':=\pi$.)

Note that two different matchings $\pi,\pi'\in\MK{[2n']}$ can yield the same monomial $\mon{\pi}$ if they differ by a flipping at some $j \in A \cap B$. We write in this case $\pi \sim \pi'$, and denote 
$\Pi = [\pi]$ the equivalence class of matchings $\pi\in\MK{[2n']}$ with respect to this equivalence relation. Thus we have $\mon{\pi}=\mon{\pi'}$ if and only if $\pi\sim \pi'$, and we denote $\mon{[\pi]}  = \mon{\pi}$.

We say that $\pi$ is {\emph {trivial on $j \in A \cap B$}}, denoted $\pi \perp j$, if the pair $\{k, k+1\}= \codd^{-1}(j)$ belongs to $\pi$. We say that $\pi$ is {\emph {trivial on $A \cap B$}}, denoted $\pi \perp A \cap B$, if $\pi$ is trivial on all elements of $A \cap B$. It is easy to see that triviality depends only on the equivalence class of $\pi$, justifying the notation $\Pi \perp j$ and $\Pi \perp A \cap B$. In the case $\Pi \perp A \cap B$, $\Pi$ consists of just a single element $\pi$, so we define $\xing(\Pi):=\xing(\pi)$ in this case.

\newcommand\Graph[1]{\Gamma_{#1}}
\def\Cyc{\operatorname{Cyc}}
\def\cyc{\operatorname{cyc}}
\def\Conn{\operatorname{Conn}}
\def\MresAB{\Mres{A'}{B'}}
Let $\pi\in\MK{[2n']}$ be a matching. Consider a graph $\Graph\pi=([2n'],E(\pi))$ with vertex set $[2n']$ and edge set 
\[E(\pi)=\pi\cup \{\{k,k+1\}\mid k\in[2n'] \text{ is such that } \codd(k)=\codd(k+1)\}=\pi\cup\{\codd^{-1}(j)\mid j\in A\cap B\}.\]
(Here if $\pi$ is trivial on $j$ then the corresponding pair $\{k,k+1\}=\codd^{-1}(j)$ belongs to both $\pi$ and $\{\codd^{-1}(j)\mid j\in A\cap B\}$, so $\Graph\pi$ contains \emph{two} edges connecting $k$ to $k+1$.)

Each connected component of $\Graph\pi$ contains an even number of vertices and is either a cycle or a path. We denote by $\Conn(\Graph\pi)$ the set of connected components of $\Graph\pi$ and by $\Cyc(\Graph\pi)\subset\Conn(\Graph\pi)$ the set of cycles of $\Graph\pi$. Clearly, flipping $\pi$ at $j\in A\cap B$ preserves the set of vertices of each connected component of $\Graph\pi$. In particular, we have $\cyc(\pi):=|\Cyc(\pi)|=|\Cyc(\pi')|$ for all $\pi\sim\pi'$, and thus we set $\cyc([\pi]):=\cyc(\pi)$.

For each equivalence class $\Pi$ of matchings we are going to compare the coefficients of $\mon{\Pi}$ on both sides of~\eqref{eq:pfaffians_Griffiths}, and show that they are equal. (Recall that we have $\mon{\pi}=\mon{\pi'}$ if and only if $\pi\sim\pi'$, and in particular we have $\mon{\Pi}\neq\mon{\Pi'}$ for $\Pi\neq\Pi'$.)

The coefficient of $\mon{\Pi}$ in the left hand side of~\eqref{eq:pfaffians_Griffiths} equals
\[
\cleft(\Pi)=\begin{cases}
(-1)^{\xing(\Pi)} & \text{if $\Pi\cap\MresAB=\emptyset$ and  $\Pi \perp A \cap B$;}\\
- (-1)^{\xing(\pi|_{A'})+ \xing(\pi|_{B'})} 2^{\cyc(\Pi)} & \text{if $\pi\in\Pi\cap\MresAB$ and $\Pi \not \perp A \cap B$;}\\
0 & \text{if  $\Pi\cap\MresAB=\emptyset$ and $\Pi \not \perp A \cap B$.}
\end{cases}
\]
Note that the case  $\Pi\cap\MresAB\neq \emptyset$, $\Pi \perp A \cap B$ is impossible because $A \cap B \not = \emptyset$. For the second case  $\pi\in\Pi\cap\MresAB$, $\Pi \not \perp A \cap B$, the parity of $\xing(\pi|_{A'}) + \xing(\pi|_{B'})$ is uniquely determined, even if $\pi$ itself may not be  uniquely determined. Indeed, any two $\pi,\pi'\in\Pi\cap\MresAB$ can be obtained from each other by flipping all $j\in S$ for some $S\subset A\cap B$ such that  $\codd^{-1}(S)$ is a union of cycles of $\Graph\pi$ (and thus a union of cycles of $\Graph{\pi'}$). Clearly in this case we have $\xing(\pi|_{A'}) + \xing(\pi|_{B'})=\xing(\pi'|_{A'}) + \xing(\pi'|_{B'})$.

Recall that $I\in\oesd$ is {\emph {compatible}} with $\pi$ if  $\pi\in\MIJ{I'}{[2n']\setminus I'}$. In this case we also say that $I'$ is {\emph {compatible}} with $\pi$. Note that the map $I\mapsto I'=\ceven(I)$ is injective on $\oesd$, and we denote by $\OddEvenPrime{\sd}:=\{I'\mid I\in\oesd\}\subset {[2n']\choose n'}$ the image of this map. Thus $I'\in\OddEvenPrime{\sd}$ if and only if $|I'|=n'$ and $|I'\cap\{k,k+1\}|=1$ for all $k\in[2n']$ such that $\codd(k)=\codd(k+1)$.
 
 It is clear that the coefficient of $\mon{\Pi}$ in the right hand side of~\eqref{eq:xing_parity} is equal to \[\cright(\Pi)=\frac{2^{n-n'}}{2^{n-1}}  \sum_{(\pi, J)}(-1)^{\xing(\pi)},\] 
where the sum is over all pairs $(\pi, J)$ such that $\pi \in \Pi$ and $J\in\oesd$ is compatible with $\pi$. 
We claim that this sum equals
\[\text{
\scalebox{0.95}{
$\displaystyle\sum_{(\pi, J)}(-1)^{\xing(\pi)} = 
\begin{cases}
(-1)^{\xing(\Pi)} 2^{n'-1} & \text{if $\Pi\cap\MresAB=\emptyset$ and  $\Pi \perp A \cap B$;}\\
- (-1)^{\xing(\pi|_{A'})+ \xing(\pi|_{B'})} 2^{n'-1 + \cyc(\Pi)} & \text{if  $\pi\in\Pi\cap\MresAB$ and $\Pi \not \perp A \cap B$;}\\
0 & \text{if  $\Pi\cap\MresAB=\emptyset$ and $\Pi \not \perp A \cap B$.}
\end{cases}$}}\]

Consider the first case $\Pi\cap\MresAB=\emptyset$, $\Pi \perp A \cap B$. Let $\pi$ be the unique element of $\Pi$. Pick some $j \in A \cap B$, and let $\{k, k+1\}:=\codd^{-1}(j)$. For each pair $\{i,i'\} \in \pi$ except for $\{k,k+1\}$, choose arbitrarily which of $i$ and $i'$ belongs to $J'$ and which does not. There are total $2^{n'-1}$ ways to do this. For each of the $2^{n'-1}$ ways, the condition $\sum_{i \in J\cap \double B} i\equiv \epsilon\pmod 2$ uniquely determines whether $k$ or $k+1$ must belong to $J'$ in order for $J$ to belong to $\oesdgood$. We are done with the first case.

\def\swp{\gamma}
Consider now the third case $\Pi\cap\MresAB=\emptyset$, $\Pi \not \perp A \cap B$. It follows that there is a pair $\{i,i'\}$ common to all $\pi \in \Pi$ such that $\codd(i) \in A \setminus B$ and $\codd(i') \in B \setminus A$. There is also a pair $\{k,k+1\}  = \codd^{-1}(j)$ for some $j \in A \cap B$ such that $k$ and $k+1$ are not connected to each other in any $\pi \in \Pi$. Consider a map $\swp:\oesdgood\to\oesdgood$ defined as follows. We put $\swp(I)=J$ for $I,J\in\oesdgood$ if $J'=\symdiff{I'}{\{i,i',k,k+1\}}$. Let $\pi'$ be obtained from $\pi$ by flipping at $j$. We claim that $I\in\oesdgood$ is compatible with $\pi$ if and only if $\swp(I)\in\oesdgood$ is compatible with $\pi'$. Moreover, $\xing(\pi')$ differs from $\xing(\pi)$ by $1$. 
Thus, we have a sign-reversing involution that cancels all the terms in $\sum_{(\pi, J)}(-1)^{\xing(\pi)}$, proving that it is equal to $0$ in the third case.

 Finally, consider the second case $\pi\in\Pi\cap\MresAB$, $\Pi \not \perp A \cap B$. We are going to show that 
\[\sum_{(\pi, J)}(-1)^{\xing(\pi)} = - (-1)^{\xing(\pi|_{A'})+ \xing(\pi|_{B'})} 2^{n'-1 +\cyc(\Pi)}.\]
Fix a matching $\pi\in\Pi\cap\MresAB$. We claim that for any $\pi'\in\Pi$, there exists $\epsilon_{\pi'}\in\{0,1\}$ such that for all $I\in\oesd$ compatible with $\pi'$, we have $I\in\Sumparity{B}{\epsilon_{\pi'}}$, that is,
\[\sum_{i\in I\cap \double B}i\equiv \epsilon_{\pi'} \pmod2.\]
Indeed, each component of $\Graph{\pi'}$ is a bipartite graph (a path or a cycle with an even number of vertices) so let us color its vertices black and white in a bipartite way. It is easy to check that $I\in\oesd$ is compatible with $\pi'$ if and only if for each connected component of $\Graph{\pi'}$, $I'$ contains either all white vertices or all black vertices of this component. Let $S\subset[2n']$ be the set of vertices of a connected component of $\Graph{\pi'}$, and let $J\in\oesd$ be such that $J'=\symdiff{I'}{S}$ (thus $J'$ is obtained from $I'$ by switching from white to black inside the component $S$). It is straightforward to check that because $\pi'$ is equivalent to $\pi\in\MresAB$, we have 
\[\sum_{i\in I\cap \double B}i\equiv \sum_{i\in J\cap \double B}i\pmod2.\]
We thus define $\epsilon_{\pi'}:=\sum_{i\in I\cap \double B}i$ for some $I\in\oesd$ compatible with $\pi'$, and we have shown that $\epsilon_{\pi'}$ does not depend on the choice of $I$.

Next, flipping $\pi'$ at some $j\in A\cap B$ changes $\epsilon_{\pi'}$ into $1-\epsilon_{\pi'}$. Thus we have $\epsilon_{\pi'}=\epsilon$ for precisely half of the matchings  $\pi'\in\Pi$, and for each such matching $\pi'$, there are $2^{|\Conn(\Graph{\pi'})|}=2^{n'-|A\cap B|}$ sets $J\in\oesd$ compatible with $\pi'$. Since $\Pi\cap\MresAB\neq\emptyset$, we have $\Pi\not\perp j$ for each $j\in A\cap B$, and thus  $|\Pi|=2^{|A\cap B|}$. Therefore the total number of pairs $(\pi',J)$ such that $\pi'\in \Pi$ and $J\in\oesdgood$ equals $2^{n'-1}$, and for each of them, the parity of $\xing(\pi')$ is the same, because it satisfies
\[\epsilon_\pi-\epsilon\equiv \xing(\pi)-\xing(\pi').\]
Thus in order to finish the proof, it suffices to show that
\begin{equation}\label{eq:Griffiths_need_1}
\xing(\pi)-\xing(\pi|_{A'})-\xing(\pi|_{B'})\not\equiv \epsilon_\pi-\epsilon\pmod2.
\end{equation}
Let $J\in\oesd$ be compatible with $\pi$. Then by the definition of $\epsilon_\pi$ and $\epsilon$, we have
\[\epsilon_\pi-\epsilon\equiv \sum_{i\in J\cap \double B}i + \sum_{i\in B'}i+1\pmod2.\]
Combining this with Lemma~\ref{lemma:xing_parity}, Equation~\eqref{eq:Griffiths_need_1} transforms into
\[|B'|/2+\sum_{i\in B'}i\equiv \sum_{i\in J\cap \double B}i + \sum_{i\in B'}i\pmod2,\]
equivalently, $|B'|/2\equiv \sum_{i\in J\cap \double B}i\pmod2$, which follows in a straightforward way since $\pi\in\MresAB$, $J$ contains either all white or all black vertices in each connected component of $\Graph\pi$, and hence the contribution of each connected component to the left and right hand side is the same. We are done with the proof of Theorem~\ref{thm:pfaffians_Griffiths}, which implies Theorem~\ref{thm:generalized_Griffiths} as discussed previously.
\end{proof}

\section{Open problems and future directions}\label{sec:conjectures}
In this section, we briefly list several questions that in our opinion would be worth exploring further.

According to~\eqref{eq:disjoint_OG}, $\OGtnn(n,2n)$ is a union of cells labeled by matchings $\medpa$ on $[2n]$, and each such cell $\pc_\medpa\cap \OGtnn(n,2n)$ is homeomorphic to $\R^{\xing(\medpa)}$. It would be nice to understand the topology closures of these cells. In fact, we have a conjecture, analogous to~\cite[Conjecture~3.6]{Pos}.

\begin{conjecture}\label{conj:regular}
The cell decomposition~\eqref{eq:disjoint_OG} gives a regular CW complex structure on $\OGtnn(n,2n)$. In other words, the closure of each cell $\pc_\medpa\cap \OGtnn(n,2n)$, given by~\eqref{eq:closures}, is homeomorphic to a closed $\xing(\medpa)$-dimensional ball.
\end{conjecture}

As we have already mentioned, the poset $\P_n$ of cells in $\OGtnn(n,2n)$ has been studied in the context of electrical networks. In particular, it has been shown to be shellable and Eulerian by~\cite{LamEuler, HK}, which shows that $\P_n$ is the face poset of \emph{some} regular CW complex by a result of~\cite{Bjo}. This leads to our next question.

\begin{question}\label{question:electrical}
Does there exist a natural stratification-preserving homeomorphism between the compactification $E_n$ of the space of response matrices of planar electrical networks (as studied in~\cite{Lam}) and the space $\Closure_n$ of boundary correlation matrices of planar Ising networks?
\end{question}

Recall that both spaces have cell decompositions into cells indexed by matchings on $[2n]$, and both spaces are homeomorphic to a closed $n\choose 2$-dimensional ball by Theorem~\ref{thm:main} and~\cite[Theorem~1.3]{GKL}. Similarly to Conjecture~\ref{conj:regular}, the space $E_n$ is believed to be a regular CW complex with face poset $\P_n$. There are many more surprising analogies between the two spaces:
\begin{itemize}
\item In both cases, a planar graph yields a point in the cell corresponding to its medial pairing.
\item Two reduced planar graphs yield the same point if and only if they are connected by the corresponding $Y-\Delta$ (or \emph{star-triangle}) moves.\footnote{In fact, under our map $G\mapsto \Gbip$, applying a $Y-\Delta$ move to $G$ corresponds to applying the \emph{superurban renewal} of~\cite{KP} to $\Gbip$.}
\item In both cases, there is an embedding of the space of boundary measurements into the totally nonnegative Grassmannian, as in Theorem~\ref{thm:main} and~\cite[Theorem~5.8]{Lam}.\footnote{The corresponding decorated permutations differ by a ``shift by $1$'', i.e., if $\pidec:[n]\to[n]$ is a fixed-point free involution then Lam embeds the electrical response matrix into the cell $\pc_{\pidec'}$ of $\Grtnn(n-1,2n)$, where $\pidec'(i):=\pidec(i)-1$ modulo $n$ for all $i\in[n]$. An analogous construction in the context of the \emph{amplituhedron} of~\cite{AHT} is related to going from the \emph{momentum space} to the \emph{momentum-twistor space}, where one performs a ``shift by $2$''. It remains an open problem to define the amplituhedron and related objects in the context of ABJM amplitudes. We thank Thomas Lam for pointing this out to us.}
\item The cyclic shift inside the corresponding Grassmannian amounts to the duality operation for Ising networks as in Section~\ref{sec:cyclic_intro}, and for electrical networks it corresponds to taking the dual graph and replacing each conductance by its reciprocal, as easily follows from the results of~\cite[Section~5]{Lam}.
\item Adding boundary spikes and boundary edges translates into adding pairs of bridges to the corresponding plabic graph, see Theorem~\ref{thm:inverse} and~\cite[Proposition~5.12]{Lam}.
\end{itemize}

Our next question is related to Remark~\ref{rmk:very_close}.

\begin{problem}
Explain rigorously the relationship between the scaling limit of planar Ising networks at critical temperature  and the unique cyclically symmetric point $X_0\in\OGtnn(n,2n)$ from Proposition~\ref{prop:X_0}.
\end{problem}

Our main result establishes a correspondence between total positivity and planar Ising networks, and thus potentially allows to apply results and intuition from one area to another. For example, asymptotic properties of plabic graphs have not yet been studied, while asymptotic properties of planar Ising networks have rich and important well-studied structure. Similarly, the space $\Grtnn(k,n)$ is usually studied in the context of \emph{cluster algebras} and \emph{canonical bases} of Lusztig, see e.g.~\cite{FZ,Lus3}. For instance, Theorem~\ref{thm:generalized_Griffiths} expresses Griffiths' inequalities as positive linear sums of minors of $\doublemap(M)$. But the theory of cluster algebras gives a much larger family of rational functions of the minors that all take positive values on $\Grtnn(k,n)$. 
\begin{problem}
Give an interpretation of the values of other cluster variables in the cluster algebra of the Grassmannian in terms of the planar Ising model.
\end{problem}

Another direction is related to Question~\ref{question:tests} and the discussion after it: what is the minimal number of minors one needs to check in order to test whether a given element $X\in\OG(n,2n)$ belongs to $\OGtnn(n,2n)$? A similar question for electrical networks has been discussed in~\cite[Section~4.5.3]{KenyonCDM}. This question also makes sense when $X$ belongs to a lower-dimensional cell inside $\OGtnn(n,2n)$. Note also that in the case of the Grassmannian, collections of such minors have a very nice structure~\cite{OPS} as they form \emph{clusters} in the associated cluster algebra. It is not clear to us whether there exists a similar ``cluster structure'' on $\OGtnn(n,2n)$.

Finally, there has been a rich interplay between the areas of scattering amplitudes and total positivity, giving rise to \emph{canonical differential forms} on positroid cells inside $\Grtnn(k,n)$, see~\cite{AHT,abcgpt,ABL,GL}. A similar result for electrical networks can be found in~\cite[Theorem~4.13]{KenyonCDM}, which gives an explicit expression for the Jacobian of a certain natural map. In~\cite[Section~2.4.2]{HWX}, an expression for another Jacobian was given for $\OGtnn(n,2n)$ in the context of ABJM scattering amplitudes. It would thus be interesting to understand their Jacobian in the language of planar Ising networks, as well as develop an analog of the \emph{amplituhedron}~\cite{AHT} for which $\OGtnn(n,2n)$ plays the role of $\Grtnn(k,n)$.

\newcommand{\arxiv}[1]{\href{https://arxiv.org/abs/#1}{\textup{\texttt{arXiv:#1}}}}

\bibliographystyle{alpha}
\bibliography{ising}

\end{document}